%% file: main.tex
\documentclass[a4paper,
			UKenglish]{article}

\usepackage[textwidth = 140mm,
 		     textheight = 222mm,
		     margin = 32mm,
 		     headheight = 3mm,
 		     headsep = 10mm,
		     footskip = 5mm]{geometry}
\usepackage{macros}
\usepackage{mathtools}
\usepackage{amsthm}
\usepackage{stmaryrd,
  amsfonts,
  newtxmath,
		    dsfont,
            	    bbm,
                     caption,
                     tikz-cd, 
                     mathrsfs,
            	    float}
\usepackage{comment}
\usepackage{orcidlink}
\usepackage{hyperref}
\usepackage[nameinlink,capitalize]{cleveref}		%
\usepackage[createShortEnv,commandRef=Cref,conf={no link to proof}]{proof-at-the-end}
\pratendSetGlobal{text proof={Proof of \string\pratendRef{thm:prAtEnd\pratendcountercurrent} on \string\cpageref{thm:prAtEnd\pratendcountercurrent}}}

\usetikzlibrary{automata, positioning, calc}
\usepackage{todonotes}
\newcommand*{\critical}[1]{}
\newcommand{\takeout}[1]{\empty}
\numberwithin{equation}{section}

\theoremstyle{plain}
\newtheorem{theorem}{Theorem}[section]

\newtheorem{proposition}[theorem]{Proposition}

\theoremstyle{definition}
\newtheorem{definition}[theorem]{Definition}
\newtheorem{example}[theorem]{Example}
\newtheorem{remark}[theorem]{Remark}

\renewcommand*{\emph}[1]{\textbf{#1}}

\begin{document}

\title{Central Limits via Dilated Categories}

\author{
Henning Basold\orcidlink{0000-0001-7610-8331}\thanks{LIACS, Leiden University,
				    \href{mailto:h.basold@liacs.leidenuniv.nl}{h.basold@liacs.leidenuniv.nl}} \and
Ois\'in Flynn-Connolly\orcidlink{0009-0000-1693-3356}\thanks{LIACS, Leiden University,
  \href{mailto:flynncoo@tcd.ie}{flynncoo@tcd.ie}} \and
Chase Ford\orcidlink{0000-0003-3892-5917}\thanks{LIACS, Leiden University,
		              \href{mailto:m.c.ford@liacs.leidenuniv.nl}{m.c.ford@liacs.leidenuniv.nl}}  \and
Hao Wang\orcidlink{0000-0002-4933-5181}\thanks{LIACS, Leiden University,
  \href{mailto:h.wang@liacs.leidenuniv.nl}{h.wang@liacs.leidenuniv.nl}}
}

\maketitle

\begin{abstract}
  The Central Limit Theorem (CLT) establishes that sufficiently large sequences of independent and identically distributed random variables converge in probability to a normal distribution
  This makes the CLT a fundamental building block of statistical reasoning and, by extension, in reasoning about computing systems that are based on statistical inference such as probabilistic programing languages, programs with optimisation, and machine learning components.
However, there is no general theory of CLT-like results currently, which forces practitioners to redo proofs without having a good handle on the essential ingredients of CLT-type results.
In this paper, we introduce dilated seminorm-enriched category theory as a unifying framework for central limits, and we establish an abstract central limit theorem within that framework.
We illustrate how a strengthened version of the classical CLT and the law of large numbers can be obtained as instances of our framework.
Moreover, we derive from our framework a novel central limit theorem for symplectic manifolds, the CLT for observables, which finds applications in statistical mechanics.
\end{abstract}

\section{Introduction}
\label{sec:intro}

The Central Limit Theorem (CLT) is a fundamental result of classical probability theory, which guarantees that the normalisation of a sufficiently large sequence of independent and identically distributed random variables with finite first and second moments converges in probability to a normal distribution. Starting with Giry~\cite{Giry82}, researchers have become increasingly interested in compositional viewpoints on probability. This point of view is very useful in reasoning about the behaviour of computing systems that are based on statistical inference. The body of work resulting from this has proven powerful enough to capture functorial analogues of many classic structural results such as de Finetti's theorem \cite{Fritz21}. A fundamental gap in the literature has been categorical descriptions of the classical limiting results of probability such as the CLT~\cite{Seewoo24}. This is due to a broader dearth of frameworks for reasoning \emph{quantitatively} about iterative analytic convergence in category theory. The purpose of this paper is to develop such methods, powerful enough to capture the CLT, through a combination of enriched category theory and the fixed point theory of normed spaces.

The first contribution of this paper is an axiomatization of the Banach fixed point theorem (BFPT) (\cref{thm:fpt}) using quantale theory and distance spaces \cite{Flagg97}. We define seminorm and dilated categories as enrichments over two related cosmoi of such distance spaces. Morphisms in a seminorm category are equipped with a notion of magnitude. The payoff is that, in the presence of such an enrichment, one can lift the BFPT to the (enriched) categorical level (\cref{thm:banach_fixed_point}), which has a form vaguely reminiscent of Lawvere's fixed point theorem \cite{lawvere69, lawvere06}. This version of BFPT seems to provide a convenient categorical machine for reasoning about limiting results from analysis in a straightforward way
As we show in the subsequent sections, it is at least enough to capture the CLT.

It has been widely recognized that the difficult part in formulating and proving central limit theorems is that one needs normalisation and rescaling operators~\cite{ breullart_thesis, Breuillard25}.
We propose \emph{dilated categories} as a solution, in which morphism can be measured by a seminorm valued in a quantale and rescaled by elements of an underlying quantale, offering a broad framework for treating this genre of problem systemically.
As our main example, we introduce a novel metric approach to categorical probability theory via enriched lax monoidal functors (\cref{Prob_is_monoidal_closed}). We show that the convolution product may be axiomised via monoidal dilated functors on pointwise additive categories (\cref{prop:axiomisedconvolution}).

We then state and prove a general structural central limit theorem (\cref{thm:gCLT}) in dilated categories. More specifically, we use an enriched functor $F$ between dilated categories, where the domain is equipped with an addition operation. This addition is pushed by the functor $F$ onto objects in its image.  In the known probabilistic central limit theorem, the functor $F$ sends a vector spaces to a metric spaces of probability measures with expectation 0 and finite variance on it.  This algebraic structure, along with the diagonal maps, is used to define a binary operator on the endomorphisms of the object. This corresponds, in the classical case, to the self-convolution of probability measures. However, such an operator will generally not be globally contractive on any given object $F(X)$, which permits a direct application the BFPT.
We remedy this by introducing a second grading functor $G$, corresponding to the assignment of a space of positive definite matrices, and posit the existence of a natural transformation $p\colon F\Rightarrow G$, corresponding to variance. Using enriched pullbacks, we use this transformation to decompose $F(X)$ into fibres.
On each fibre, the rescaled operator $\theta$ restricts to an endomorphism that is contractive. It follows from the categorical Banach theorem that repeated iterations of it converge to a central limit in each fibre. In the probabilistic case, this is the Gaussian distribution with fixed variance matrix.
This yields a strengthening of the CLT (\cref{thm:functorialCLT}), namely that the assignment of the normal distribution $\mathcal N(0, M)$ to each possible variance value $M$ in a finite vector space $V$ can be lifted to a natural transformation.
Finally, we prove a novel CLT for observables, illustrating how this framework is useful for constructing more complicated CLTs from simpler ones and, in particular, constructing higher order CLTs via compositional reasoning .

Our work falls in the area of categorical probability theory, which plays an increasingly important role in reasoning about probabilistic processes.
We see our work particularly as a first step to devise logical reasoning techniques for stochastic differential equations.
These play a crucial role in, for instance, optimisation and have a stochastic term given by Gaussian distributions.
Another motivation to reason about sampling processes on, for instance, smooth manifolds by using analogues of the law of large numbers.
Finally, we would like to point out that quantale-valued seminorms are also of independent interest for studying rates of algorithmic convergence, as the seminorm of a multistep process estimates its worst case rate of convergence.
Therefore, our framework may have applications in reasoning \textit{quantitatively} about limiting procedures.

\paragraph*{Related work.} Given the vast literature on CLTs, we restrict our review to recent developments most pertinent to the present work.

\paragraph*{Categorical probability theory.}
There have been several attempts to study probability theory using the tools of category theory, with early ideas already appearing in an unpublished manuscript of Lawvere in the 1960s. Of particular note, is the notion of probability monads, which appeared first in \cite{Giry82}. Subsequently, these have been extended from measurable spaces $\Meas$ to other closely related categories \cite{Keimel08,FP20}. We also use a probability functor $\Prob$ in our main example but are more interested in its lax monoidal structure rather than its monadicity.

A second, more recent and much more general approach to categorical probability theory has been Markov categories, introduced in \cite{Fritz20}. Markov categories subsume the previous approach due to the fact that the Kleisli categories of probability monads generally form Markov categories. Limiting processes have been studied synthetically in Markov categories in \cite{Fritz21} and via an enrichment in \cite{Perrone24}.  We are aware of an attempt to study probability via enriched categories within Markov categories, this being \cite{Perrone24}, which enriched Markov categories to study entropy.  In particular, we recently became aware of \cite{fritz26}, which offers an alternative approach to the law of large numbers that is more synthetic than ours.

We emphasise that our framework is complementary to the Markov approach. Markov categories excel at capturing the compositional and structural aspects of probability, enabling synthetic reasoning about copying and marginalisation, but they are not designed to capture the quantitative aspects of the theory, such as rates of convergence. Our enriched framework provides these analytical tools, enabling categorical reasoning about limit theorems, iterative processes, and convergence rates. We view these as currently orthogonal perspectives on categorical probability; however, they suggest a path toward a future unified theory of dilated Markov categories.

Other recent approaches include effectuses \cite{Jacobs18,ChoEA15} which were introduced to study quantum logic. A Cartesian closed category for higher order categorical probability theory was introduced in \cite{Heunen17}. 

\paragraph*{Weighted and normed category theory.} Enrichments similar to ours have appeared variously in other contexts. The notion of a weighted category first appears in \cite{grandis07}. This is quite close to the categories we use as bases of enrichment, except that ours take values in a quantale rather than the extended reals exclusively.  The work of \cite{CHT25} relates to ours in that their enrichment is also rich enough to support a Banach fixed point theorem, though flavoured rather differently to ours. There is also a categorical Banach theorem in coalgebra \cite{Ademek94}. Their theorem applies to contracting functors between complete metrically enriched categories. The category of complete metrically enriched categories therefore likely carries a seminorm category structure. On a slightly more removed note, our dilated categories seem to be somewhat related to the type system appearing in \cite{reed10}. It would be interesting to investigate this further.

\paragraph*{Fixed point approaches to the CLT} Our paper was directly inspired by, and owes a considerable debt to, the fixed point approach to the CLT~\cite{HW84}, which was based on that from~\cite{trotter59}.

\paragraph*{Overview}
After a brief review of the necessary background on enriched category theory and quantales in~\Cref{sec:prelim}, 
we establish a generalisation of the Banach Fixed Point Theorem for complete distance spaces over a quantale in~\Cref{sec:spaces}.
In~\Cref{sec:seminormspaces}, we introduce a category of seminorm spaces. In~\Cref{sec:seminormcategories}, we enrich over this category and in \Cref{sec:catBanach} establish a categorical Banach Fixed Point Theorem in the enriched setting. 
\Cref{sec:convolution} contains the material on our notion of convolution product in seminorm-categories. 
Next, in \Cref{sec:CLT}, we establish the main result of this work, a central limit theorem native to dilated-categories, and illustrate how the classic CLT and   law of large numbers  is obtained by instantiation.  We conclude by introducing a novel CLT, the CLT for observables, inspired by statistical mechanics on symplectic manifolds in \Cref{sec:pushforward}.
\section{Preliminaries}
\label{sec:prelim}
\pratendSetLocal{category=prelim}

We briefly gather the necessary background on enriched category theory~\cite{Kelly05} and quantales~\cite{Flagg97},
and fix notation.
We assume familiarity with basic concepts of category theory~\cite{AHS90,MacLane98} such as functors, natural transformations, adjunctions, and monoidal categories. 

Throughout this work, we adopt the perspective of \emph{generalised elements}.  If a category $\Cat{C}$ possesses a terminal object $\1$, a \emph{point} of an object $X$ is defined as a morphism $x\colon \1 \to X$. 
This generalizes the set-theoretic case where elements $x \in X$ correspond bijectively to functions $\{*\} \to X$. 
\subsection{Enriched category theory}
We begin with a review on enriched categories~\cite{Kelly05}. 

A \emph{monoidal category} is a category~$\Cat{V}$ equipped with a \emph{tensor} bifunctor 
$\Cat{V}\times\Cat{V}\xra{\tens}\Cat{V}$, a \emph{unit object}~$\tensU$, and natural isomorphisms
\begin{equation*}
\alpha\colon (A \tens B) \tens C \xra{\cong} A \tens (B \tens C),
\quad
\lambda\colon \tensU \tens A \xra{\cong} A, 
\quad
\text{and}
\quad
\rho\colon A \tens \tensU \xra{\cong} A
\end{equation*}
(the \emph{associator}, \emph{left unitor}, and \emph{right unitor}, respectively) subject to coherence laws.
We call~$\Cat{V}$ \emph{symmetric} if there is additionally a coherent natural isomorphism~$A\tens B\cong B\tens A$.
A symmetric monoidal category~$\Cat{V}$ is \emph{closed} if the functor~$-\tens B\colon\Cat{V}\to\Cat{V}$ has a right adjoint~$[B, -]\colon\Cat{V}\to\Cat{V}$ for each~$B\in\Cat{V}$.
In this case, we have a natural isomorphism of hom-sets
\begin{equation*}
\VC(A\tens B, C)\cong\VC(A, [B, C]),
\end{equation*}
and we refer to~$[B, C]\in\Cat{V}$ as the \emph{internal hom} from~$B$ to $C$.

\emph{Unless explicitly mentioned otherwise, assume that $\VC$ is a complete and cocomplete closed symmetric monoidal category.}
Examples include the category~$\Set$ of sets and maps, and the category~$\Pos$ of partially ordered sets and monotone maps with their cartesian closed structure.
An example in which the tensor is not Cartesian includes the category~$\Vect{\R}$ of real vector spaces and linear maps with the usual tensor product.

A \emph{$\VC$-enriched category} (or a \emph{$\VC$-category})~$\eAC$, which will always be denoted with an underline, comprises a collection of objects~$\ob{\eAC}$ along with the following data for all objects~$A,B,C\in\ob{\eAC}$: an object~$\eAC(A, B)\in\VC$ of \emph{morphisms from~$A$ to~$B$} for all~$A,B\in\ob{A}$, a \emph{composition morphism} 
\begin{equation*}
\comp\colon\eAC(B, C)\tens\eAC(A, B)\to\eAC(A, C),
\end{equation*}
and an \emph{identity morphism}~$j_A\colon \tensU\to\eAC(A, A)$.
This data is subject to coherence laws (involving the structure of~$\VC$ as a monoidal category) expressing that composition is associative and unital.
Given a $\VC$-category $\eAC$, we define its \emph{underlying category} $\AC$ to be the category with the same objects as~$\eAC$ and hom-sets defined by $\AC(A, B) := \VC(\tensU, \eAC(A, B))$.
An \emph{enriched functor} $\enr{F} \from \eAC \to \eBC$ between $\VC$-categories $\eAC$ and $\eBC$ consists of a map $F \from \ob{\eAC} \to \ob{\eBC}$ and for all $A, B \in \eAC$ a morphism 
\begin{equation*}
\enr{F}_{A,B} \from \eAC(A,B) \to \eBC(FA, FB)
\end{equation*}
in $\VC$ that preserves identity and composition.
We will typically omit the subscripts if they are clear from the context.
Given an enriched functor $\enr{F}$, we obtain a functor $F \from \AC \to \BC$ that acts as $\enr{F}$ on objects and on morphisms $f \from \tensU \to \eAC(A, B)$ by $Ff = f \comp \enr{F}_{A,B}$.
An \emph{enriched natural transformation} $\enr{F} \to \enr{G}$ is just a natural transformation $F \to G$.
We write $\VCat$ for the 2-category of $\VC$-enriched categories, enriched functors and natural transformations.

\begin{example}
A~$\Set$-category~$\AC$ is an (ordinary) locally small category. 
Indeed,~$\AC$ consists of the data of collection of objects together with the assignment of a \emph{set}~$\eAC(A, B)$ of morphisms for all object pairs. 
Composition is just a functional assignment of a composite morphism to all composable pairs, and the identity morphism is just the assignment of a morphism~$\id := j_A(*)$ to each object viewed as a generalised point~$\{*\}\to\eAC(A, A)$.
The coherence laws of~$\AC$ express that composition is associative and that composition by identity morphisms is unital. 
\end{example}
\subsection{Quantales and Distance Spaces}
Quantales are ordered algebraic structures that support a generalised concept of metric space~\cite{Flagg97}, which we review next. 

By a sequence in a set~$X$, we always understand a map~$x\colon\omega\to X$, where~$\omega$ denotes the poset of natural numbers with the standard ordering.
Whenever it is convenient, we use the point-wise notation~$(x_n)_{n\in\omega}$, or even just~$(x_n)$, for the sequence $x$.

A \emph{complete lattice} is a set~$V$ equipped with a partial order~$\leq$ such that the supremum (or \emph{join})~$\join A$ exists for every $A\subseteq V$.
Then~$V$ also admits arbitrary infima (or \emph{meets}) characterised for each~$A\subseteq V$ by
$\meet A = \join\up{A},$
where~$\up{A}$ is the set of upper bounds of~$A$~\cite[Theorem 4.2]{BS81:CourseUniversalAlgebra}.
For an index set~$I$ and a family of points~$x_i\in V$~$(i\in I)$, we write~$\join_{i\in I} x_i$ (or just~$\join x_i$) instead of~$\join\setDef{x_i}{i\in I},$ we denote binary joins by~$x\lor y$, and we write~$\bot$ and~$\top$, respectively, for the bottom and top elements of~$V$.

A \emph{quantale}~$(\V, \qMul,\unit)$ is a complete lattice~$\V$ equipped with the structure of a commutative monoid with tensor~$\qMul$ and unit~$\unit$ such that for each~$x\in\V$ the functor~$x \qMul - \colon\V\to\V$  preserves joins:
\begin{equation*}
r\qMul\left(\bigvee s_i\right) = \bigvee(r\qMul s_i).
\end{equation*}
Since joins are colimits in~$\V$, viewed as a category, it follows from the adjoint functor theorem that each~$x\qMul -$ has a right adjoint~$[x,-]\colon\V\to\V$ characterised by a natural isomorphism
$\V(r\qMul s, t)\cong \V(s, [r,t]).$
This means that
\begin{equation*}
r\qMul s \leq t\text{ if and only if }s\leq [r, t].
\end{equation*}
We call~$\V$  \emph{contractive} if $q \qMul x < x$ when $q < e$ and $x \neq \top, \bot$. 
Note that~$\bot\qMul r = \bot$ for all~$r\in\V$ since $\qMul$ preserves joins and $\bot = \bigvee\varnothing$; we refer to this property as \emph{absorbtion}.

\begin{example}
\label{expl:boolean}
The complete lattice~$\2 = \bot\leq\top$ carries the structure of a quantale with~$x\qMul y := x\land y$ and $\unit = \top.$ 
Then~$[x, y]:= (v_1\to v_2)$.  
Then~$(\2, \land, \top)$ is trivially contractive.
\end{example}
\begin{example}
\label{expl:reals}
Let~$\gR$ denote the set of non-negative reals extended by a point~$\infty$.
We extend the standard ordering, addition, and multiplication on the non-negative reals to~$\gR$ by defining
\begin{equation*}
r\leq \infty,
\qquad
r+\infty = \infty = \infty+r, 
~~\text{and}\qquad
r\mult\infty = \infty = \infty\mult r
\end{equation*}
for all~$r\in\gR$.
This ordering makes~$\gR$ into a complete lattice with top element~$\infty$. 
Moreover,~$\gR$ carries the structure of a quantale $\gR_{\times}$ with multiplication as tensor and the unit~$1$. 
Then $[x, y] = \frac{y}{x}$ for $x,y \neq 0, \infty$ in~$\gRm$. 
By equipping~$\gR$ with the complete lattice structure given by the opposite order~$\geq$ and extended addition, one obtains the so-called Lawvere quantale. 
Note that~$\gR_{\times}$ is contractive and, for trivial reasons, so is the Lawvere quantale.
\end{example}

\paragraph*{Generalised Distance Spaces}
Fix a quantale~$(\V, \qMul, \unit)$ for the remainder of this section.
We are going to axiomatically describe a notion of generalised distance space, which involves distance values from~$\V$, as a suitable class of $\V$-fuzzy relations~\cite{Zadeh71}.

\begin{definition}
\label{defn:space}
A \emph{$\V$-relation} is a set~$X$ equipped with a function~$d\colon X\times X\to\V,$ which assigns a \emph{distance} $d(x, y)\in\V$ to every pair of points~$x,y\in X$.
A map~$f\colon X\to Y$ between~$\V$-spaces~$X$ and~$Y$ is \emph{nonexpansive} if
\begin{equation*}
d_Y(f(x), f(y))\leq d_X(x, y)
\end{equation*}
for all~$x,y\in X$. 
We write
$\VRel$
for the category of~$\V$-relations and nonexpansive maps.
A \emph{$\V$-space} is a $\V$-relation~$(X, d)$ such that
\begin{equation*}
d(x,x)\leq\bot
\quad
\text{and}
\quad
d(x,y) = d(y,x)
\end{equation*}
for all~$x,y\in X$.
We write~$\VSp$ for the category of~$\V$-spaces as objects and non-expansive maps as morphisms.

A \emph{limit point} of a sequence~$x$ in~$X$ is a point~$y\in X$ such that, for any~$r>\bot$ in~$\V,$ there exists~$N_r\in\omega$ such that~$d(y, x_n)\leq r$ for all~$n\geq N_r$.  A $\V$-space is \emph{(sequentially) Hausdorff} if whenever~$y,z\in X$ are limit points of a sequence~$x,$ then~$y = z$.
In this case, we write~$\lim(x)$ for the unique limit point of~$x$.
\end{definition}

We note that our concept of~$\V$-space differs from the~$\V$-categories of Flagg~\cite{Flagg97} (also see Lawvere~\cite{Lawvere02} for the case of metric spaces) since our notion of distance is required to be symmetric and is not required to satisfy the triangle inequality. 
Whenever a $\V$-space is Hausdorff, one obtains the usual point separation condition for metric spaces as in the following.

\begin{lemmaE}
\label{lem:separation}
If~$(X, d)$ is Hausdorff, then $d(x,y) = \bot$ iff~$x = y$.
\end{lemmaE}
\begin{proofE}
Consider the constant sequence defined by~$x_n = x$ for all~$n\in\N$. 
Then, if~$d(x, y) = \bot,$ we have~$d(x_n, x) = d(x_n, y)$ for all~$n\in\N.$
In particular,~$d(x_n, x), d(x_n, y)\leq r$ for all~$r\in\V$. 
Since~$X$ is Hausdorff, it follows that~$x = y,$ as desired.
\end{proofE}

\begin{lemmaE}
\label{lem:convergence}
Nonexpansive maps preserve limit points.
\end{lemmaE}
\begin{proofE}
Let $x$ be a sequence in $X$ with limit point $y$. Let $f: X\to Y$ be a nonexpansive map. Then for any~$r>\bot$ in~$\V,$ there exists~$N_y\in\omega$ such that~$d(y, x_n)\leq r$ for all~$n\geq N_r$. It follows that
$$
d(f(y), f(x_n)) \le d(y, x_n) \leq r.
$$
We conclude that $f(x)$ converges to $y$.
\end{proofE}
\noindent
We recover our familiar examples from this formalism.
\begin{propositionE}
Every metric space is both a Hausdorff $\gRa$-space and a Hausdorff $\gRm$-space.
\end{propositionE}
\begin{proofE}
The reflexive axiom of metric spaces implies the reflexive axiom of distance spaces since $\bot$ in $\gR$ is 0. Distance is a symmetric function for metric spaces. Finally to prove the Hausdorff condition, suppose towards contradiction that a sequence $\{a_i\}_{i \in \mathbb N}$ in a metric space has two limits $L, M$. Then, for all $\epsilon \in \gRa$ there exists an $N>0$ such that $i>N$ implies that $d(L, a_i), d(a_i, M) < \epsilon/2$. Then, by the triangle inequality, $d(L, M) \leq d(L, a_i) + d(a_i, M)) \epsilon/2 + \epsilon/2 = \epsilon $.  Therefore
$$
d(L, M) < \epsilon
$$
for all $y' \in \gRa$. We can conclude that $L = M$.

Exactly the same argument establishes that every metric space is a Hausdorff $\gRm$-space
\end{proofE}

The category~$\VSp$ carries the structure of a symmetric monoidal category with the tensor $A\tens_{\VSp} B$ of $\V$-distance spaces $A, B$ being given by equipping the Cartesian product $A\times B$ with the \textbf{maximum metric}, i.e.
\begin{equation*}
d_{A\tens B}\left((a_1, b_1), (a_2, b_2) \right) = d_A(a_1, a_2) \vee d_B(b_1, b_2).
\end{equation*}

\begin{remark}
Distance spaces are the quantalic analogue of usual addition-based metric space theory. For example, multiplication is a much more natural operation to consider in many computing contexts; for example (symmetric) divergences in statistics \cite{kullback51, cisis68}, neural networks \cite{Piotr17} and  quantitative algebra \cite{radu16}.
\end{remark}

\section{Fixed Points in Distance Spaces}
\label{sec:spaces}
\pratendSetLocal{category=spaces}

In this section, we describe a concept of geometric completeness for~$\V$-spaces, i.e. in which every geometric sequence has a unique limit point, and consider the full subcategory of~$\VSp$ spanned by these spaces. 
In this setting, we establish a generalisation of the Banach Fixed Point Theorem, which we recall first.

\begin{theoremE}[Banach Fixed Point Theorem]
\label{thm:BFPT}
Let $(X, d)$ be a non-empty complete metric space and~$f\colon X \to X$ a map such that for some non-negative real~$q<1,$ the inequality
\begin{equation*}
d(f(x), f(y)) \leq q\cdot d(x,y)
\end{equation*}
holds for all~$x,y\in X.$ 
Then $f$ has a unique fixed-point $\fix(f)$. 
Moreover, for any~$x\in X$, we have~$\fix(f) = \lim f_x$ where~$f_x$ is the sequence defined by~$(f_x)_0 = x$ and~$(f_x)_{n+1} = f((f_x)_n)$. 
\end{theoremE}

In order to generalise~\Cref{thm:BFPT} to~$\V$-spaces, we first define a suitable concept of \emph{complete} $\V$-space, based on the idea of geometrically complete~$\V$-spaces in which every geometric series has a unique limit point.
There are some subtleties which emerge when replacing~$\gR$ by a more general quantale, which we address first.
\emph{In what follows, we assume that $\V$ is a contractive quantale.}

\subsection*{Completeness in Distance Spaces}
It is important in applications to guarantee the existence of limits of sufficiently well-behaved sequences. In complete metric space theory, all Cauchy sequences converge. In this work, as we need only the Banach fixed point theorem, we restrict attention to geometric sequences. A sequence $x: \omega \to X$ is \textbf{geometric} if there are $r < e$ and $q < \top$, such that for all $i ,j \in \mathbb N$ 
\begin{equation*}
  d(x_{i}, x_{i+1}) \leq r^i \qMul q
  \quad \text{and} \quad
  d(x_i, x_j) < \top \, .
\end{equation*}
In this case, we refer to the value~$r$ as a \emph{constant} of~$x$.
We say that~$(X, d)$ is \emph{geometrically complete} (or just \emph{complete}) if every geometric sequence in~$X$ has a unique limit point $L$ satisfying $d(x_i, L) \le \join_{n \in \mathbb N} \meet_{j > n} d(x_i, x_j).$
We write~$\CVSp$ (or just~$\CDS$ if~$\V$ is clear from context) for full subcategory of~$\VSp$ spanned by all complete $\V$-spaces.
\begin{example}
\label{expl:geometric}
For our purposes, the key example of a geometric sequence can be summarised as follows.
First, define the \emph{Lipschitz constant}~$\Lnorm{f}$ (or just~$\norm{f}$) of a nonexpansive map $f\colon X\to Y$ as the meet over all~$c\in\V$ such that, for all~$x,y\in X$,
\begin{equation*}
d(f(x), f(y)) \leq c\qMul d(x,y).
\end{equation*}
Note that since~$f\in\DS(X,Y)$ is nonexpansive, we have~$\Lnorm{f}\leq\unit$.

Now, for an endomorphism~$g\colon X\to X$ on a~$\V$-space~$(X, d)$, let~$g_x\colon\N\to X$ denote the sequence of iterates of~$g$ on~$x\in X$ defined by~$(g_x)_0 = x$ and~$(g_x)_{n+1} = g((g_x)_{n}),$ i.e.~$g_x = (g^n(x)).$
Then~$g_x$ yields a geometric series with constant~$\Lnorm{g}$.
\end{example}

\begin{example}
\label{expl:ERquantale}
Any Cauchy complete metric space over~$\gR$ is a complete~$\gRm$-space because Cauchy convergence implies geometric convergence in a metric space. This is a standard argument using the Cauchy criterion for series convergence and the formula for the sum of a  geometric sequence.
\end{example}
\begin{propositionE}
If $X,Y\in\CDS,$ then~$X_1\tens X_2\in\CDS$. 
In particular, the functor~$(-) \tens X_2$ given by tensoring with~$X_2$ restricts to~$\CDS$.
\end{propositionE}
\begin{proofE}
Suppose one has a geometric sequence $z\colon\N\to X_1\tens X_2,$ let~$z^i:= \pi_i\comp z\colon\N\to X_i$ denote the induced sequence in component~$i\in\{1,2\}$.
Then each $c_n = (a_n, b_n)$ and we have that
 \begin{multline*}d_A(a_n, a_{n+1}),d_B(b_n, b_{n+1})  \leq \vee\left(d_A(a_n, a_{n+1}), d_B(b_n, b_{n+1}) \right) \\ = d_{A\tens B}(c_n, c_{n+1}).
  \end{multline*}
 In particular, this implies that the sequences $\{a_n\}_{n \in \mathbb N}$ and $\{b_n\}_{n \in \mathbb N}$ are both geometric. By assumption, they both converge to $a$ and $b$ respectively for $a\in A$ and $b\in B$. So one has 
 for any~$r>\bot$ in~$\V,$ there exists~$N_y, M_y\in\omega$ such that~$d(a, a_n)\leq r$ for all~$n\geq N_r$ and $d(b, b_n)\leq r$ for all~$n\geq M_r$. Choosing $K_r = \max{(N_y, M_y)},$ one has  for $n > K_r$
 $$
 d_{A\tens B}((a,b), c_{n}) = \vee\left(d_A(a, a_{n}), d_B(b, b_{n})\right) \leq r
 $$
  which implies that $\{c_n\}_{n \in \mathbb N}$ converges to $(a,b).$
\end{proofE}
\begin{example}
Every convergent sequence in a metric space has a geometric subsequence, and in this context, all geometric sequences are Cauchy. However, the argument for this does not work in the absence of the triangle inequality. Our Banach fixed point theorem below therefore strictly subsumes the classical one.
\end{example}

\subsection*{Fixed Points in Complete~$\V$-spaces}

With this notion of completeness fully developed, our fixed point theorem for complete~$\V$-spaces can be formulated as follows. A  $\V$-space $X$ is said to be \emph{metrically small} if $d(x,y) < \top$ for all $x,y \in X.$

\begin{theoremE}
\label{thm:fpt}
Let $X$ be an inhabited, metrically small, complete $\V$-space %
and $f\colon X\to X$ a nonexpansive map such that~$\Lnorm{f} <\unit$.
Then~$f$ has a unique fixed point~$\fix(f)$, i.e.~$f(\fix(f)) = \fix(f)$.
Moreover, for any~$x\in X$,~$\fix(f)$ is characterised by 
\begin{equation*}
\fix(f) = \lim(f^{n}(x))_{n\in\N}.
\end{equation*}
\end{theoremE}
\begin{proofE}
Let $x\in X$ and consider the sequence~$f_x$ with~$(f_x)_n = f^n(x)$ for all~$n\in\N$.  
We first show that~$f_x$ is geometric.
To this end, we proceed by induction on~$n\in\N$ to show that
\begin{equation}
d(f^n(x), f^{n+1}(x)) < \Lnorm{f}^n\qMul d(x, f(a))
\end{equation}

By induction, we have that 
\begin{equation*}
d(f^n(a), f^{n+1}(a)) < q^{\qMul n}\qMul d(a, f(a)).
\end{equation*}
By assumption $d(x_i, x_j) < \top$ for all $x_i, x_j \in X.$ It follows that the sequence is geometric and therefore, by completeness, converges to some $\ol{a}.$ 

We next show that~$\ol{a}$ is a fixed point of~$f$.
To this end, note that, since $\ol{a} = \lim (f^i(a)),$ for each $y< \bot$, there exists $N\in\N$ such that
\begin{equation*}
 d(\ol{a}, f^{i}(a)) < y. 
\end{equation*}
for all~$i> N$. 
Since~$\Lnorm{f}\leq q$, it now follows that 
\begin{equation*}
d(f(\ol{a}), f^{i+1}(a)) < q\qMul d(\ol{a}, f^{i}(a)) < q\qMul y.
\end{equation*}
Thus, $f(\ol{a})$ is a limit point of~$(f^n(a))$. 
We conclude that $f(\ol{a}) = \ol{a}$ since sequences have unique limit points by the Hausdorff property of distance spaces.

It remains to be shown that~$\ol{a}$ is unique.
To this end, suppose that $b\in A$ is a fixed point of~$f$. 
Then 
\begin{equation*}
d(a, b) = d(f(\ol{a}), f(b)) \leq q\qMul d(a, b) < d(a, b),
\end{equation*}
where the inequality follows from the contractive property of the quantale. This is a contradiction.
\end{proofE}

\begin{example}
\label{expl:fixedpoint}
Consider the Boolean quantale~$\2,$ which has tensor given by~$\land$ and unit~$1 = \top.$
\Cref{thm:fpt} tells us that any nonexpansive map~$f\colon X\to X$ on a $\2$-space~$X$ with~$\Lnorm{f} = 0$ has a unique fixed point~$\fix f = \lim f_0.$
Therefore the Banach fixed point theorem in this setting tells us that all maps with Lipschitz seminorm 0 on discrete metric spaces are constant. 
\end{example}

\section{Seminorm and dilated spaces}
\label{sec:seminormspaces}
\pratendSetLocal{category=seminormspaces}
In this section, we describe a category of~$\V$-seminorm spaces for a quantale $\V$ and Lipschitz continuous maps which will serve as a base of enrichment for the subsequent developments. 
We illustrate this machinery on the core examples of seminorm categories that will be employed in the sequel.%
\begin{definition}
\label{def:SNorm}
An \textbf{seminorm space} is a complete $\V$-space~$(A,d)$ with map~$\abs{-}\colon A\to\V$, called the \emph{seminorm}. A morphism of  seminorm spaces from~$(A, |\cdot |_A)$ to~$(B, |\cdot |_B)$ is a nonexpansive map~$f\colon A\to B$ such that~$|f(x)|_B\leq \abs{x}_A$. We write $\SNorm$
for the category of  seminorm spaces and their morphisms. This is a monoidal category, see \cref{prop:tensor} below, where the carrier of $A\tens B$ is the product of the carriers in $\CDS$. The seminorm is given by
\begin{equation*}
  \abs{(a,b)} = \abs{a}\qMul\abs{b} \, .
\end{equation*}
\end{definition}
\begin{example}
For any quantale $\V$, we have the one element set $\{x\}$, equipped with the distance space structure $d(x,x) = \bot$, and where $\abs{x} = e.$ This is the unit of the monoidal structure.%
\end{example}
\subsection*{Action spaces}
We are going to develop the theory of seminorm spaces, but where we are able to appropriately rescale morphisms. These will be our dilated spaces, and we develop the theory of this enrichment next.

The \textbf{unit interval monoid} of a quantale $\V$ is the submonoid of~$(\V, \qMul, \unit)$ on 
\begin{equation*}
\Vint= \setDef{x\in\V}{x \leq\unit}.
\end{equation*}
Note that $\Vint$ inherits joins from $\V$: given a family of points~$x_i\in\Vint$ we have~$x_i\leq e$ for all~$i$ hence also~$\join x_i\leq e$.
Moreover, the multiplications of $\V$ restricts to $\Vint$ with $e$ as unit, which makes $\Vint$ a quantale.

\begin{definition}
\label{defn:action}
An \emph{action} on a complete $\V$-space~$(A,d)$ is a nonexpansive map
$\act\colon\Vint\times A\to A$
such that, for all $a\in A$ and $x,y\in\Vint$,
\begin{equation*}
\unit\act a = a
\qquad
\text{and}
\qquad
x\star(y\act a) = (x\qMul y)\act x.
\end{equation*}
That is,~$\act$ is a left action of the interval monoid on~$A$.
A \emph{($\V$-)action space} is a complete $\V$-space equipped with an action.
An action space $(A, \act)$ is \emph{conical} at~$a\in A$ if the map $\bot\act -\colon A\to A$ is constant with value $a$.
In this case, we call $a$ the \emph{base point}.
A \emph{equivariant map} from~$(A, \star_A)$ to~$(B, \star_B)$ is a nonexpansive map~$f\colon A\to B$ which preserves monoid actions, i.e.~$(x\star_A a) = x\star_B f(a)$ for all~$x\in\Vint$ and all~$a\in A$.
We write $\VAct$ for the category of~$\V$-action spaces and equivariant maps. 
\end{definition}

Note that an equivariant map~$f\colon (A, a)\to (B, b)$ between conical action spaces preserves base points:
\begin{equation*}
f(a) = f(\bot\act_A a) = \bot\act_B f(a) = b.
\end{equation*}
Thus, we may form the full subcategory
$
\bVAct
$
of~$\VAct$ spanned by all conical action spaces and basepoint preserving maps.
\begin{definition}
A \emph{dilated space} is a conical action space equipped with map~$\abs{-}\colon A\to\V$ (the \emph{(dilated) seminorm}) such that
\begin{equation*}
\abs{x \act a} = x \qMul\abs{a}
\end{equation*}
for every $x \in V$ and every~$a \in A$. 
A morphism of conical seminorm spaces from~$(X, |\cdot |_X)$ to~$(Y, |\cdot |_Y)$ is a nonexpansive map~$f\colon X\to Y$ such that~$|(f(x))|_Y\leq \abs{x}_X$.
We write $\DNorm$ for the category of conical seminorm spaces and their morphisms. 
\end{definition}
We next explain how to define a monoidal structure on this category. This involves making one more choice.
\begin{definition}
\label{def:unit_interval_dilated}
A  \textbf{unit interval dilated space} is a pointed seminorm structure on the unit interval monoid $\Vint$ of the quantale.
We assume $\Vint$ is equipped with a quantale-valued metric $d_{\Vint}$ satisfying the following axioms:
\begin{enumerate}
    \item The quantale multiplication $\cdot\colon \Vint \times \Vint \to \Vint$ is non-expansive with respect to the metric;
    \item The seminorm is the identity $\abs{-} = \id_{\Vint}$ and the basepoint is $\bot$;
    \item The space $(\Vint, d_{\Vint})$ is geometrically complete.
\end{enumerate}
\end{definition}
\begin{example}
In the quantale $\gRm$, the unit interval $[0,1]$ admits two distance space structures. One can consider the Euclidean distance $d(x,y) = |x-y|$, or the distance $d(x,y) = 0$ if $x=y$ and $\infty$ otherwise.
\end{example}
\begin{propositionE}
For any choice of unit interval dilated space, the symmetric monoidal category $\DNorm$ is equivalent to the category of Eilenberg-Moore algebras over pointed seminorm spaces $\SNorm_{\bullet}$ for the monad $T$ defined by $T(X) = \Vint \wedge X$.
\end{propositionE}
\begin{proofE}
First we remark that by \cite[Lemma 4.20]{Elmendorf09} that as the category $\SNorm$ has finite limits and colimits (which will be proven in the proof of Theorem \ref{thm:cosmoi}), the closed symmetric monoidal structure on $\SNorm$  lifts to a closed symmetric monoidal structure $\SNorm_\bullet.$ The monoidal product lifts to a smash product $\wedge.$

For completeness, we describe the action monad structure explicitly. Let $X_+$ denote the space $X \sqcup \{\ast\}$ with a basepoint. The functor $T$ is defined by the smash product $T(X) = \Vint \wedge X_+$, which is the quotient of the product $\Vint \times X_+$ by the subspace $(\{\bot\} \times X_+) \cup (\Vint \times \{\ast\})$. In other words, it is the following coequaliser
\begin{equation*}
(\{\bot\} \times X_+)  \sqcup (\Vint \times \{\ast\})\rightrightarrows \Vint \times X_+
\end{equation*}
where the first map is the constant map to $(\bot, \ast)$ and the second is the inclusion map. The quotient therefore exists in $\SNorm$ by Theorem~\ref{thm:cosmoi}.

The unit $\eta_X\colon X \to T(X)$ is induced by the map $x \mapsto (e, x)$, and the multiplication $\mu_X\colon T(T(X)) \to T(X)$ is induced by the monoid multiplication in $\Vint$.

An algebra for this monad is an object $A \in \SNorm_{\bullet}$ equipped with a morphism $\xi\colon T(A) \to A$ satisfying the unit and associativity axioms.
The map $\xi$ is equivalent to a map $\act\colon \Vint \times A \to A$ (where $q \act a = \xi(q, a)$) subject to the quotient condition that $\xi$ is constant on the identified subspace.
Specifically, the identification of $\{\bot\} \times A$ implies that $\bot \act a$ is constant for all $a \in A$.
Let $a_0 = \xi(\bot, a)$. This recovers the conical condition of a dilated space.

The unit axiom $\xi \circ \eta_A = \id_A$ implies that $e \act a = a$, which is the unital axiom.
The associativity axiom $\xi \circ T(\xi) = \xi \circ \mu_A$ implies that for all $p, q \in \Vint$, we have $p \act (q \act a) = (p \qMul q) \act a$, which recovers the associativity axiom.

Finally, a morphism of algebras $f\colon (A, \xi_A) \to (B, \xi_B)$ is a morphism in $\SNorm_{\bullet}$ such that $f \circ \xi_A = \xi_B \circ T(f)$.
This condition translates to $f(q \act a) = q \act f(a)$, which is precisely the definition of an equivariant map in $\DNorm$.
Thus, the category of algebras is always isomorphic to $\DNorm$ regardless of the monoid structure on $\Vint$.

Finally, to ensure that the monoidal structure on $\SNorm$ lifts to $\DNorm,$ we shall apply \cite[Theorem 2.2]{kock71}. The object $\Vint$ is a  commutative monoid as it is a submonoid of $\V$ and the structure maps $\Vint \wedge \Vint \to \Vint $ are nonexpansive. The monoid $\Vint$ is commutative since the quantale multiplication is assumed to be commutative. On the category $\SNorm_{\bullet}$, $\Vint\wedge -$ is the action monad of this commutative monoid. It follows by \cite[Example 6.3.12)]{brandenburg}, that it is a commutative monad, and we can thus apply \cite[Theorem 2.2]{kock71}.
\end{proofE}
\begin{example}
The unique morphism~$\act\colon\Vint\times\term\to\term$ on a terminal object~$\term$ of~$\CDS$ (which is equipped with seminorm $\bot$) yields the structure of an action on~$\term$.

For any quantale $\V$, the terminal object in seminorm spaces is simply the one element set $\{x\}$, equipped with the distance space structure $d(x,x) = \bot$, and a trivial monoid action. However, this is \textbf{not} the unit of the monoidal structure, this will be $\Vint$ itself equipped with its canonical action on itself via the quantalic multiplication.
\end{example}
\subsection*{Closed monoidal structure}
We describe the tensor product on $\DNorm$ induced by the product on $\SNorm.$
 \begin{definition}
 \label{definition:tend_dNorm}
The \emph{tensor} $(A, a_0) \tens (B, b_{0})$ of a conical action spaces $(A, a_0)$ and $(B, b_0)$ is carried by the set%
$
\Vint \times |A|\times |B| / \sim
$
where we identify the tuples $(q, a, b)$
\begin{enumerate}
\item
we identify all tuples $(\bot, a, b), (q, a_0, b )$ and $(q, a, b_0)$; 
\item
we identify $(q\qMul r , a, b), (q, r \act a, b_2)$ and $(q, a, r\act b)$.
\end{enumerate}
This is equipped with the quantale valued distance space structure as follows.
The free metric is%
\begin{equation*}
d_F\left((q_1, a_1, b_1), (q_2, a_2, b_2) \right) =  d(q_1, q_2) \vee d_A(a_1, a_2) \vee d_B(b_1, b_2)
\end{equation*}
and the quotient metric is given by
\begin{equation*}
d_{A\tens B}\left((q_1, a_1, b_1), (q_2, a_2, b_2) \right) = \meet d_F\left((q'_1 , a'_1,  b'_1), (q'_2, a_2,  b'_2, ) \right).
\end{equation*}
where the meet is taken over all equivalence class representatives for $(q_1, a_1, b_1)$ and $(q_2, a_2, b_2)$ respectively.
We then complete the resulting space to ensure all geometric sequences converge. The points added in the completion step all have seminorm $\top$. The result is a quantale valued metric.  The action is given by
$$
\act\colon \Vint \times A\tens B \to A\tens B
$$
$$
\left(q, (p, a, b)\right) \mapsto (q\qMul p, a, b)
$$
Finally, the seminorm is given as
\begin{equation*}
\abs{-}\colon A\tens B\to \V, \qquad \abs{(q, a,b)} := q\qMul  \abs{a}_A\qMul\abs{b}_B
\end{equation*}
yields a seminorm on~$A\tens B$. The object $\Vint$ is the unit of this monoidal structure.
\end{definition}
\begin{propositionE}
\label{prop:tensor}
The assignments~$(A, B)\mapsto A\tens B$ is the object-part of a functor
$\tens\colon\SNorm\times\SNorm\to\SNorm$ and $\tens\colon\DNorm\times\DNorm\to\DNorm$, in the latter case, with the action on pair of morphisms given point-wise.%
Moreover,~$\tens$ equips both ~$\SNorm$ and $\DNorm$ with the structure of a symmetric monoidal category with units~ $\{x\}$ and $\Vint$ respectively.
\end{propositionE}
\begin{proofE}
We explicitly check all properties for the tensor product in $\SNorm$.

One can easily see that~$A\tens B$ is a complete distance space. 
The reflexive and symmetric properties of the quantale-valued metric are satisfied by the definition of the product.

The only non-trivial verification is that $A\tens B$ is geometrically complete.
Let $\{c_n\}_{n \in \mathbb N}$ be a geometric sequence in $A\tens B$.
Then each $c_n = (a_n, b_n)$ and we have that
\begin{multline*}
d_A(a_n, a_{n+1}), d_B(b_n, b_{n+1})\\ \leq \vee\left(d_A(a_n, a_{n+1}), d_B(b_n, b_{n+1}) \right) = d_{A\tens B}(c_n, c_{n+1}).
\end{multline*}
In particular, this implies that the sequences $\{a_n\}_{n \in \mathbb N}$ and $\{b_n\}_{n \in \mathbb N}$ are both geometric.
By assumption, they both converge to $a$ and $b$ respectively for $a\in A$ and $b\in B$, which implies that $\{c_n\}_{n \in \mathbb N}$ converges to $(a,b)$.

We next show that for morphisms~$f\colon A\to C$ and~$g\colon B\to D$, the map~$f \tens g\colon A\tens B\to C\tens D$  defined for~$(a, b)\in A\times B$ by
$$
(f\tens g)(a, b): = (f(a), g(b))
$$
is a morphism of seminorm spaces. This is nonexpansive as 
\begin{multline*}
d_{C\tens D}\left( (f(a_1), g(b_1)), (f(a_2), g(b_2)) \right)
\\ = \vee\left(d_C(f(a_1), f(a_2)), d_D(g(b_1), g(b_2)) \right)
\\
\leq \vee\left(d_A(a_1, a_2), d_B(b_1, b_2) \right)
\\
= d_{A\tens B}\left((a_1, b_1), (a_2, b_2) \right) 
\end{multline*}
We next need to prove that $\tens$ is a bifunctor. It suffices to prove that $\tens$ preserves composition and identities. To get the the identity, $(\id_A \tens \id_B)(a, b) = (a, b) = \id_{A \tens B}(a, b)$. For the composition, let $f'\colon C \to E$ and $g'\colon D \to F$. Then
\begin{multline*}
((f' \tens g') \circ (f \tens g))(a, b)
= (f' \tens g')(f(a), g(b)) \\
= (f'(f(a)), g'(g(b))) = ((f'\circ f)(a), (g'\circ g)(b)) \\= ((f'\circ f) \tens (g' \circ g))(a, b).
\end{multline*}

For the symmetry, the map $\sigma_{A,B}\colon A \tens B \to B \tens A$ defined by $\sigma(a, b) = (b, a)$ is clearly an isometric isomorphism preserving the seminorm. Associativity and the pentagonal law are similar as all maps involved are isometric isomorphisms of seminorm $e$.

We define the left unitor $\rho_A\colon I  \tens A   \to A$ (where $I=\{x\}$ is the unit object) by
\begin{equation*}
    \rho_A(x, a) = a.
\end{equation*}
This is a morphism in $\SNorm$ of seminorm $e$ and it is non-expansive because
$$
    d_A(\rho_A(c_1), \rho_A(c_2)) = d_A(a_1, a_2)  = \vee\left( \bot, d_A(a_1, a_2) \right)  = d_{I \tens A }(c_1, c_2).
$$
where $c_i  = (x, a_i).$ The inverse map $\rho_A^{-1}\colon A \to I \tens A$ is given by $a \mapsto (x, a)$.
This is an isometry because
\begin{align*}
    d_{I \tens A }((x, a_1), (x, a_2)) &= \vee\left( d_I(x, x), d_A(a_1, a_2) \right) \\
    &= \vee\left( \bot, d_A(a_1, a_2) \right) = d_A(a_1, a_2).
\end{align*}
The left unitor axiom is precisely the same.

To obtain that $\DNorm$ is a symmetric monoidal category, it suffices to observe that $\SNorm_{\bullet}$ is symmetric monoidal and has coequalisers. This is because categories of pointed objects are complete and cocomplete if the underlying category is complete and complete, and we shall prove in Theorem \ref{thm:cosmoi} that this indeed the case.  It is equipped with the smash monoidal product turning it into a symmetric monoidal category. Then one can immediately apply \cite[Theorem 2.1]{Keigher78} to deduce that it is symmetric monoidal.

The description of the tensor product in Definition \ref{definition:tend_dNorm} comes from the explicit definition in \cite[Theorem 2.1]{Keigher78}, the tensor product in $\DNorm$ is the coequaliser of the two maps
$$
\begin{tikzcd}
A \tens_{\SNorm_{\bullet}}   \Vint  \tens_{\SNorm_{\bullet}}  B 
  \arrow[r, shift left=.5ex, "l "] 
  \arrow[r, shift right=.5ex, "r"'] 
  & A \tens_{\SNorm}  B
\end{tikzcd}
$$
where one of the maps acts on the left $l(a,v, b) = (a\act v, b)$ and the other on the right $r(a,v,b) = (a, v\act b)$.
\end{proofE}
Both $\SNorm$ and $\DNorm$ are convenient bases of enrichment.  
 \begin{theoremE}
 \label{thm:cosmoi}
 Both $\SNorm$ and $\DNorm$ are cosmoi: ie.\ closed, symmetric monoidal categories that are complete and cocomplete.
\end{theoremE}
\begin{proof}
\noindent\textbf{Symmetric monoidal closed}: To show that $\SNorm$ is closed symmetric monoidal, we construct the adjoint to the tensor product. Given $A, B, C \in \SNorm$ one has 
$$
[B, C] = \left\{f \colon B \to C : \exists q \in \V \mbox{ such that } |f(b)|_C \leq q \qMul \abs{b}_B \right\}
$$
This is equipped with the supremum metric $$d_{[B,C]}(f,g) = \sup_{b\in B} d_C(f(b, g(b))$$ and the Lipschitz seminorm 
$$
\abs{f}_{[B,C]} = \meet\{q \in \V | \forall b \in B. |f(b)|_C \leq q \qMul \abs{b}_B \}
$$
To conclude, we must verify the adjunction
$$
\SNorm(A \tens B, C) \cong \SNorm(A, [B,C]).
$$
Given a morphism $f\colon A \to [B,C]$, one has $|f(a)|_{[B,C]} < \abs{a}_A$ for all $a \in A$ by nonexpansiveness of maps in $\SNorm$ with respect to Lipschitz constant. But, by the definition of the Lipschitz seminorm, this is true if and only if for for all $b\in B$
$$
|f(a)(b)|_C \leq |f(a)|_{[B,C]} \qMul \abs{b}_B \leq \abs{a}_A \qMul \abs{b}_B
$$
It follows that the map defined by
$$(a, b) \mapsto f(a)(b)$$
satisfies precisely the nonexpansiveness condition coming from Definition \ref{def:SNorm}. We conclude that it defines an object of  $\SNorm(A \tens B, C)$ as required.

Next, we establish that $\DNorm$ is closed monoidal. For $B, C \in \DNorm$, we define the internal hom $[B, C]$ in $\DNorm$ as the subspace of $[B, C]$ in ${\SNorm}$ consisting of equivariant morphisms with respect to the action of $\Vint$.
The internal seminorm and metric from $\SNorm$ then restrict to this subspace.
This object is automatically conical because for any $f \in [A, B]$, we have that $(\bot \act f)(a) = \bot \act f(a)$, which must be the basepoint of $B$ as $B$ is conical, making $\bot \act f$ a constant map with value the base element. The adjunction then follows from the adjunction in $\SNorm$ restricted to equivariant maps.
\\
\noindent\textbf{$\SNorm$ is Complete:} 
To show that $\SNorm$ is complete, we must prove that it has products and equalisers.

\noindent\textbf{Products:} Binary products have already been constructed in the proof of Theorem \ref{prop:tensor} (though rest of the monoidal category structure is not Cartesian because of how the unit is defined). To construct the infinite products, we generalise this construction. For any indexing set $\mathcal I$, we have $\prod_{i \in \mathcal I} X_i$ is carried by the underlying product in $\Set.$ The distance between two elements $x, x' \in \prod_{i \in \mathcal I} X_i$ is then given by $\join \{d_{X_i}(x, x') \mid i \in \mathcal I\}.$ This clearly satisfies identity of the indiscernables and is symmetric. This space is sequentially Hausdorff. Suppose a sequence $x^{(n)}$ has limit points $y$ and $z$. For any $r > \bot$, there is an $N_r$ such that for all $n \geq N_r$, we have $d(y, x^{(n)}) \leq r$ and $d(z, x^{(n)}) \leq r$. By the definition of the supremum, this implies $d_{X_i}(y_i, x^{(n)}_i) \leq r$ and $d_{X_i}(z_i, x^{(n)}_i) \leq r$ for every coordinate $i \in \mathcal{I}$. Therefore, $y_i$ and $z_i$ are both limit points of the sequence $x^{(n)}_i$ in $X_i$. Because each $X_i$ is a $\V$-space and thus Hausdorff, limits are unique, so $y_i = z_i$ for all $i$. Consequently, $y = z$.

Finally, we must show that $\prod_{i \in \mathcal I} X_i$ is geometrically complete. Let $(x^{(n)})_{n \in \mathbb{N}}$ be a geometric sequence in the product space, meaning there exist $r < e$ and $q < \top$ such that $d(x^{(n)}, x^{(n+1)}) \le r^n \qMul q$. By the definition of the supremum metric, this implies $d_{X_i}(x_i^{(n)}, x_i^{(n+1)}) \le r^n \qMul q$ for every coordinate $i \in \mathcal{I}$. Thus, each projected sequence $(x_i^{(n)})$ is geometric in $X_i$. 

Since each $X_i$ is geometrically complete, $(x_i^{(n)})$ converges to a unique limit point $L_i \in X_i$ satisfying 
$$d_{X_i}(x_i^{(n)}, L_i) \le \join_{m \in \mathbb N} \meet_{j > m} d_{X_i}(x_i^{(n)}, x_i^{(j)}).$$ 
Because coordinate distances are bounded by the product distance, $d_{X_i}(x_i^{(n)}, x_i^{(j)}) \le d(x^{(n)}, x^{(j)})$. Substituting this bound and taking the supremum over all $i \in \mathcal{I}$ yields:
$$d(x^{(n)}, L) = \join_{i \in \mathcal{I}} d_{X_i}(x_i^{(n)}, L_i) \le \join_{m \in \mathbb N} \meet_{j > m} d(x^{(n)}, x^{(j)}).$$ 
Therefore, $L = (L_i)_{i \in \mathcal{I}}$ is the unique limit point. So the infinite product is indeed geometrically complete.

\noindent\textbf{Equalisers:} So it suffices to show that has equalisers. We can explicitly construct them:  given $f, g: A \to B$  morphisms in $\SNorm$, their equaliser $E$ is the subspace of $A$ with underlying set  $E = \{ x \in A : f(x) = g(x) \}$. This inherits a metric and seminorm given by the subspace metric.   Finally the inclusion $i: E \hookrightarrow A$ is clearly non-expansive in terms of distances and preserves also preserves the seminorm ($|i(x)|_A = \abs{x}_E$), meaning that it is non-expansive on seminorms. 

We also need to show that $E$ is geometrically complete with respect to the subspace metric. Let $(a_n)_{n\in\mathbb{N}}$ be some geometric sequence in $E$. Since $E \subseteq A$ and $A$ is geometrically complete, we have that the sequence converges to a unique limit $L$ in $A$. We must show that $L$ is $E$. Recall from \Cref{lem:convergence} that morphisms in $\SNorm$ are non-expansive. It follows that the morphisms preserve limits. Because $f(x_n) = g(x_n)$ for all $n$, we conclude $f(L) = \lim f(x_n) = \lim g(x_n) = g(L)$. Thus $L \in E$, establishing that $E$ is complete.
\\
\noindent\textbf{$\SNorm$ is Cocomplete:} 

Next, we must show that $\SNorm$ it has all coproducts and coequalisers.

\noindent\textbf{$\SNorm$ has coproducts:} The underlying set is the disjoint union of copies placed at distance $\top$ from each other. The seminorm and distance within each copy is inherited directly from the original sets. This clearly satisfies the universal property of the product as the inclusion maps exist and are isometries. It is also complete as the copies are distance $\top$ from each other so each geometric sequence must eventually end up in one copy.

\noindent\textbf{$\SNorm$ has coequalisers:}  Finally we show that $\SNorm$ has coequalisers. This construction is involved; we summarize the main steps below. Let $f, g\colon A \to B$ be morphisms in $\SNorm$. We construct the coequaliser $C$ in two steps: first we form a quotient space, making it Hausdorff, and then apply a construction analogous  to the Cauchy completion.

 Let $\sim$ be the equivalence relation on the set $B$ given by identifying $f(a) \sim g(a)$ for all $a \in A$. We let $B_0 = B / \sim$ be the quotient. There is a projection  $\pi_0:B \to B_0$.  We define a distance function $d_{B_0}$ on $B_0$ by taking the following meet. For $\mathbf{x},\mathbf{x}\in B_0$:
    \begin{equation*}
    d_{B_0}(\mathbf{x},\mathbf{x}) = \meet \left\{ d_B(x_0, y_0) : x_0 \in x, y_0 \in y \right\}
    \end{equation*}
    where the meet is taken over all possible representatives for $x_0 \in\mathbf{x}$ and $y_0 \in \mathbf{x}.$ We define seminorm on $B_0$ by taking meets over the same class
    \begin{equation*}
    |x_{B_0} | = \meet \left\{ |x_0|_B : x_0\in\mathbf{x} \right\}.
    \end{equation*}
    As $f$ and $g$ have the property that $\abs{f(x)} \leq \abs{x}$, $\pi_0$ is a morphism of seminorm spaces.
    
Next, we must discuss the completion $C$ of $B_0$. We warn the reader that this behaves very differently to the completion in metric spaces.  Let $\mathcal{G}(B_0)$ be the set of geometric sequences in $B_0$ with constants less than $e$. 

\textbf{Sequences with limit points:} If a geometric sequence $(x_i)$ has a limit point $L \in B_0$, we identify the sequence with $L$ and set $d([x_i], y) = d(L, y)$ and $|[x_i]|_C = |L|_{B_0}$ for all $y \in C$.

\textbf{Sequences without limit points:} Consider the remaining collection of geometric sequences without limit points in $B_0$, denoted $\overline{\mathcal{M}(B_0)}$. We quotient this by the relation $\sim$ of eventual agreement, i.e., $(x_i) \sim (y_i)$ if there exists some $N>0$ such that $x_i = y_i$ for all $i>N$, to obtain a set $\mathcal{M}(B_0)$.

We define the completion $C$ as the disjoint union $C = B_0 \sqcup \mathcal{M}(B_0)$. The distance and seminorm on $C$ are defined as follows:

\begin{enumerate}
\item For $x, y \in B_0$: $d_C(x, y) = d_{B_0}(x, y)$ and $\abs{x}_C = \abs{x}_{B_0}$.

\item For distinct equivalence classes $\mathbf{x}, \mathbf{y} \in \mathcal{M}(B_0)$, where $\mathbf{x} = (x_i)$ and $\mathbf{y} = (y_i)$:
\begin{equation*}
d_C(\mathbf{x}, \mathbf{y}) = \top
\end{equation*}
and
\begin{equation*}
|\mathbf{x}|_C =  \top
\end{equation*}

\item For $x \in B_0$ and $\mathbf{y} \in \mathcal{M}(B_0)$:
\begin{equation*}
d_C(x, \mathbf{y}) =  \meet_{(y_i) \in \mathbf{y} } \join_{n \in \mathbb N} \meet_{i>n} d(x, y_i)
\end{equation*}
where we take the meet across all equivalence classes for $(y_i)\in \mathbf{y} $ for which $x = y_1.$
\end{enumerate}

\textbf{Geometric completeness of $C$:} We now show that $C$ is geometrically complete. Consider a sequence in $C$, we have three cases to consider:

\textit{Case 1: All but finitely many terms in $B_0$.} If $(z_n)$ has all but finitely many terms in $B_0$, we truncate after the last point in $ \mathcal{M}(B_0)$ and then it is a geometric sequence in $B_0$. If it has a limit $L \in B_0$, we are done. Otherwise, $[(z_n)] \in \mathcal{M}(B_0) \subseteq C$ serves as its limit point.

\textit{Case 2: Infinitely many terms in $\mathcal{M}(B_0)$.} Once the sequence enters $\mathcal{M}(B_0)$, it cannot transition to a distinct point within $\mathcal{M}(B_0)$, as the distance between distinct points is infinite. It therefore must remain there.

\textit{Case 3: Infinitely alternating between $B_0$ and $\mathcal{M}(B_0)$.} A sequence $(z_n)$  cannot have more than one distinct term in $\mathcal{M}(B_0)$. This is because for $z_i, z_i+ k$ such terms we would have $d(z_i, z_i+ k) = \top$ and that contradicts our assumptions about distances in the sequence always being less than $\infty.$ We conclude that $(z_n)$ must contain the same point of $\mathcal{M}(B_0)$ infinitely often which implies it converges to that value.

\textbf{Universal property:} Lastly, we must show that object we have defined is the coequalizer. Let $q = i \circ \pi_0 \colon B \to C$ where $i: B_0 \hookrightarrow C$ is the obvious inclusion that we have just defined. This is a morphism in $\SNorm$, and by construction $q \circ f = q \circ g$.

Suppose $h \colon B \to Z$ is a morphism in $\SNorm$ such that $h \circ f = h \circ g$. The morphism $h$  must factors uniquely through $B_0$, the simple quotient, vianon-expansive map $h_0 \colon B_0 \to Z$.

For any $\mathbf{x} \in \mathcal{M}(B_0)$ where $(x_n) \in \mathbf{x}$ is an equivalence class of geometric sequences in $B_0$  without a limit, the sequence $(h_0(x_n))$ is 
$$
d_Z(h_0(x_i), h_0(x_{i+1})) \leq d_{B_0}(x_i, x_{i+1}) < r^i \qMul q.
$$
We also have 
$$
d_Z(h_0(x_i), h_0(x_j))) \leq d_{B_0}(x_i, x_j) < \top
$$
The sequence $(h(x_i))$ is therefore geometric and has a unique limit point $L$. We can therefore extend the function by defining $\overline{h}(\mathbf{x} ) = L$   This is well-defined. Suppose $(y_n) \in \mathbf{x}$ is another representative sequence. By the definition of $\mathcal{M}(B_0)$, there exists $N$ such that $x_n = y_n$ for all $n > N$. Consequently, $h_0(x_n) = h_0(y_n)$ for all $n > N$, and thus their limits in $Z$ must coincide.

We must now check that $\overline{h}$ is a morphism in $\SNorm$. The map $h$ is automatically nonexpansive on norms by construction. For the seminorm condition, if $\mathbf{x} \in \mathcal{M}(B_0)$, then $|\mathbf{x}|_C = \top$. It follows immediately that $|\overline{h}(\mathbf{x})|_Z \leq \top = |\mathbf{x}|_C$. 

Next, we must check that $\overline{h}$ is nonexpansive on distances. This holds automatically for pairs of points in $B_0$ and $\mathcal M(B_0).$ Therefore the only nontrivial verification is showing the inequality for $x \in B_0$ and $\mathbf{y} \in \mathcal M(B_0).$

Now $h_0$ is non-expansive on $B_0$, we have for all $k$:
$$d_Z(h_0(x), h_0(y_k)) \leq d_{B_0}(x, y_k).$$

Since the join and meet operations in the quantale $\mathcal{V}$ are order-preserving, we can apply the operation $\bigvee_{n \in \mathbb N} \bigwedge_{k>n}$ to both sides of the inequality:
$$\join_{n \in \mathbb N} \meet_{k>n} d_Z(h_0(x), h_0(y_k)) \leq \join_{n \in \mathbb N} \meet_{k>n} d_{B_0}(x, y_k).$$
Taking the meet over all representatives of $(y_i)$ of $\mathbf{y}$ we notice that the right-hand side is precisely the definition of $d_C(x, \mathbf{y})$ and the same is true for the left-hand side.

Thus, one has
$$d_Z(h_0(x), L) \leq d_C(x, \mathbf{y}).
$$

We conclude that $\overline{h}$ is non-expansive and therefore a morphism in $\SNorm.$ 

The uniqueness of $\overline{h}$ follows from the fact that $h$ is fixed on $B_0$. Then any nonexpansive map, then by Lemma \ref{lem:convergence}, limits of sequences are sent to limits. This completes the proof that $\SNorm$ has coequalisers.

\noindent\textbf{$\DNorm$ is complete and cocomplete:} Categories of pointed objects are complete and cocomplete if the underlying category is complete and complete, as pointed objects are a comma category where both functors are continuous and cocontinuous.

The category $\DNorm$ is thus an Eilenberg-Moore category over a complete category and the forgetful functor creates limits. Therefore it is complete. 

To show that it is cocomplete, we note that, by  \cite[Lemma 4.20]{Elmendorf09} as the category $\SNorm$ has all finite limits and colimits, the underlying functor of the action monad $\Vint \wedge - $ has a right adjoint $[\Vint, -]$; this the set consisting of basepoint preserving maps $f: \Vint  \to X$. This is equipped with the sub-distance space structure. It is therefore colimit preserving. So the Eilenberg-Moore category has all colimits.
\end{proof}
Most of the details of the previous verification are routine. We use the monadicity of $\DNorm$ over $\SNorm$ to establish our co(completeness) results for it. The main difficulty is establishing that $\SNorm$ has coequalisers, as it is necessary to take geometric completions; this is a delicate procedure that differs significantly from the metric case. 

Via the closed structure we can compute the seminorm of the identity maps. 
\begin{propositionE}
The isomorphisms in $\SNorm$ and  $\DNorm$ have seminorm $e$.
\end{propositionE}
\begin{proofE}
The seminorm of the identity map is the unit of the quantale $e$. This follows from its definition as
$$
|\id_A|_{[A,A]} = \meet \{q \in \V | \forall b \in B. |f(b)|_A = \abs{b} \leq q \qMul \abs{b}_A \}
$$ 
as clearly $e$ is in the set the join above is taken over as $e \qMul \abs{b}_A  = \abs{b}_A \ge \abs{b}_A$  and $\abs{b} > q \qMul \abs{b}_A $ for any $q < e$ by the semi-monotonicity property. By the same argument and the nonexpansiveness of maps in $\SNorm$ in seminorm it follows that $\abs{f} \le e$ for all maps $f$.

Suppose $f$ has an inverse $f^{-1}$. Then
$$
f^{-1} \circ f = \id_A.
$$
We compute the seminorm of both sides.   Using the sub-multiplicative property of the seminorm, we get:
$$
e = \abs{\id_A} = \abs{f^{-1} \circ f} \leq \abs{f^{-1}} \qMul \abs{f}.
$$
To conclude, assume that $\abs{f} <  e$. Then, multiplying by $|f^{-1}|$  we obtain $\abs{f} \qMul |f^{-1}| <  |f^{-1}| \le e$ where we use semi-monotonicity. This is a contradiction and the conclusion follows.
\end{proofE}
\section{Seminorm and Dilated Categories}
\label{sec:seminormcategories}
\pratendSetLocal{category=seminormcategories}
In this section, we introduce dilated categories as the main vehicle for an abstract formulation of central-limit-type theorems.

\begin{definition}
\label{def:enrichment}
  A \emph{seminorm-category} (resp.\ a \emph{dilated category}) $\UCat{C}$ is an $\SNorm$-category (resp.\ a $\DNorm$-category).
  Since $\SNorm$  (resp.\ $\DNorm$) is a concrete category, a morphism $f \from X \to Y$ in the underlying category of $\UCat{C}$ is an element of $\UCat{C}(X, Y)$.
  An object $X$ in a seminorm category $\eCC$ is said to be \emph{metrically right trivial} (resp.\ \emph{left trivial}) if for all $Y \in \eCC$, the space $\eCC(X,Y)$  (resp.\ $\eCC(Y,X)$) is discrete, ie.\ $d(g_1, g_2) = \top$ for all $g_1, g_2 \in \eCC(X,Y)$ (resp.\ $g_1, g_2 \in \eCC(Y,X)$) such that $g_1 \neq g_2$. Otherwise it is said to be \emph{metrically right non-trivial}.

  We define the \emph{left} and \emph{right composition} with $f$, respectively, $\postComp{f} \from \UCat{C}(Z, X) \to \UCat{C}(Z, Y)$ and $\preComp{f} \from \UCat{C}(Y, Z) \to \UCat{C}(X, Z)$ by $\postComp{f}(g) = f \circ  g$ and $\preComp{f}(g) = g \circ f$.
  We say that $\UCat{C}$ is \emph{left} (respectively, \emph{right}) \emph{compatible}, if for all morphisms $g_{1}$ and $g_{2}$ such that $d(g_{1}, g_{2}) < \top$, we have
  \begin{equation*}
    d(\postComp{f}g_{1}, \postComp{f}g_{2}) \leq \abs{f} \cdot d(g_{1}, g_{2})
    \quad \text{resp.} \quad
    d(\preComp{f}g_{1}, \preComp{f}g_{2}) \leq \abs{f} \cdot d(g_{1}, g_{2}).
  \end{equation*}
\end{definition}
Unfolding the definition of enriched categories, a seminorm category comprises a collection of objects and for all $X, Y\in \UCat{C}$ a seminorm space $\UCat{C}(X, Y)$ of morphisms.
Moreover, the composition is \emph{sub-multiplicative}, that is, for all $f \from X \to Y$ and $g \from Y \to Z$ we have $\abs{g \comp f} \leq \abs{g} \abs{f}$.
We also note that left and right compatibility are not redundant: even though we have by definition that $\postComp{f}g \leq \abs{g} \abs{f}$ and that $d(\postComp{f}g_{1}, \postComp{f}g_{2}) \leq d(g_{1}, g_{2})$, there is a priori not a relation between the metric and the norm.
The compatibility can be understood as ensuring that composition with a $g$ is $\abs{g}$-Lipschitz continuous. In a dilated category, we can also act on morphisms with $\Vint$, rescaling the morphisms.
\begin{example}
\label{ex:eContVect}
The category $\eEBan$ of \emph{homogeneous Banach spaces} has Banach spaces (over~$\R$) as objects and all homogenous maps (ie. those maps such that $f(c\cdot x) =c \cdot f(x)$ for all $c \in [0,1]$ and $x \in \mathbb R^n$ ) between them as morphisms. These maps are not necessarily either contractive, linear nor bounded. The hom-objects $\eEBan(A, B)$ can be equipped with a conical action of $\mathbb R$  via the formula
$$
(c \cdot f) = c \cdot_B f(x) \mbox{ for } f \in \eEBan(A, B)
$$
where we use the scalar multiplication on $B$. Furthermore, these are (extended) metric spaces via the usual supremum metric of function spaces. The seminorm of a continuous morphism $f\in \eEBan(V,W)$ is its Lipschitz constant, 
$$
\Lnorm{f} :=  \inf \left\{ c \in \mathbb R^{\infty}_{\geq 0} \colon d_W(f(v_1), f(v_2))\leq c \cdot d_V(a_1, a_2) \forall v_1, v_2\in V  \right\}.
$$
In particular, if $f$ is not Lipschitz in the classical sense, then its seminorm will be $\infty$.
\end{example} 
\begin{propositionE}
The category $\eEBan$ is a dilated category.
\end{propositionE}
\begin{proofE}
We have already explained the enriched structure in the body of Example \ref{ex:eContVect}. It remains to check the submultiplicative and, the left/right compatibility.

\textit{Submultiplicative axiom:} By definition, one has that $d_B(f(a_1), f(a_2))\leq \Lnorm{f} \cdot d_A(a_1, a_2)$ and $d_C(g(b_1), g(b_2))\leq \Lnorm{g} \cdot d_B(b_1, b_2)$  for all $a_1, a_2\in A$ and $b_1, b_2 \in B.$ It follows that, in particular, $d_C(g \circ f(a_1), g \circ f(a_2))  \leq \Lnorm{f}\Lnorm{g} \cdot d_A(a_1, a_2)$ for all $a_1, a_2\in A$. We conclude that $ \Lnorm{f}\Lnorm{g} \geq \Lnorm{g\circ f}$ as $\Lnorm{g\circ f}$ is a lower bound for constants $c$ with the property that $d_C(g \circ f(a_1), g \circ f(a_2))  \leq c \cdot d_A(a_1, a_2)$ for all $a_1, a_2\in A$

\textit{Left/right compatibility:} Without loss of generality, it suffices to check the left composition metric axiom. The distance between $g_1, g_2 \in \eEBan(A, B)$ is defined to be 
$
\sup_{a \in A} d(g_1 (a), g_2(a)).
$
One has  $ d(f\circ g_1 (a),f\circ  g_2(a)) \leq \Lnorm{f}\cdot d(g_1 (a), g_2(a)).
$ It follows that $
 d(f\circ g_1, f\circ g_2) =  \sup_{a \in A} d(f\circ g_1 (a),f \circ  g_2(a)) \leq \sup_{a \in A}  \Lnorm{f}\cdot d(g_1 (a), g_2(a)) =  \Lnorm{f}\cdot   \sup_{a \in A}  d(g_1 (a), g_2(a))  =  \Lnorm{f}\cdot d(g_1, g_2).
$
\textit{Composition property} We have to show that composition is contractive and respects multiplication. This is straightforward:
$$
f\circ (c \cdot g) = f( c\cdot - ) \circ g = c (f\circ g)
$$
by homogeneity.
\end{proofE}
\begin{definition}
Let $\eCC$ and $\UCat{D}$ be seminorm (resp.\ dilated) categories.  A \emph{seminorm functor} (resp.\ \emph{dilated functor})  $\enr{F}\colon \eCC \to \UCat{D}$ is a functor on the underlying concrete categories such that for all $f \in \eCC(A, B)$, one has 
$$
\abs{F(f)} \le \abs{f}
$$
\end{definition}
\begin{example}
A very simple example of a dilated functor on $\eEBan$  is  $\enr{H}: \eEBan \to \eEBan$ defined by $\enr{H}(V) = V$ and $\enr{H}(f) = \frac{1}{2}\cdot  f.$
\end{example}
We conclude by noting that we are immediately able to compute seminorms of many identity maps. \begin{propositionE}
\label{prop:identites}
Let $\UCat{C}$ be a dilated or a seminorm category. Then, for any metrically right (or left) nontrivial $A\in \eCC$, the seminorm of an identity map $\id_A\in \UCat{C}(A,A)$ is always equal to $e$.
\end{propositionE}
\begin{proofE}
The identity is a morphism of seminorm spaces $\{x\} \to \UCat{C} (A, A)$ in seminorm categories and $\Vint \to \UCat{C} (A, A)$ in dilated categories. In both cases, as the maps are nonexpansive we obtain that $ |\id_A| \leq e.$ 

The identity map has the property that for any $B$ and  $A$ $d(g_{1}, g_{2}) \leq \abs{\id_A} \cdot d(g_{1}, g_{2})$.   Therefore, by the contractive property of the quantale, $e \leq |\id_A|$ if there are any distinct $\bot < d(g_1, g_2) < \top$.   We conclude that $ |\id_A| = e$ if 
\end{proofE}

\subsection*{More examples of seminorm \& dilated categories}
We conclude this section by discussing some explicit examples of dilated categories.  
The largest class of examples are $\Ban$-enriched categories.
\begin{propositionE}
\label{prop:ban_enrich}
Any $\Ban$-enriched category $\UCat{C}$ is a dilated category.
\end{propositionE}
\begin{proofE}
The quantale that the distance spaces is valued in  is  $(\mathbb R^{\infty}, \cdot, 1)$. Banach spaces $V$ are always equipped with a norm $|-|_V$ and we define the seminorm of an element $v\in V$ to be equal to this norm. Banach spaces are also, definitionally, complete metric spaces with respect to the induced metric $d(u,v)= |u-v|$ for $u,v \in V$. 

 Any vector space $V$ naturally carries an action of canonical action of $\mathbb R$ coming from the vector space structure. This action is conical as $0\cdot v = 0_V$ for all all $v\in V$ where $0_V$ is the origin of $V$.

The submultiplicative axiom follows directly from the definition of a Banach enrichment. Morphisms in the category of Banach spaces are contractive with respect to the norm.  We therefore require
$$
\UCat{C}(A, B) \tens \UCat{C}(B, C) \to \UCat{C}(A, C) 
$$
$$
(v, w) \mapsto w\circ v
$$
to be contractive. More precisely, one has
$$
|v|\cdot|w| =  |(v, w)| \geq |w\circ v|
$$
which gives the submultiplicative axiom as required.

Lastly, we verify the left and right compatibility of the enrichment. Without loss of generality we  verify only the left compatibility. One has
$$
d(v \circ w_1, v \circ w_2) = |v \circ w_1 - v \circ w_2| 
$$
By bilinearity, one has 
$$
|v \circ w_1 - v \circ w_2|  = |v \circ (w_1-w_2)|
$$
Finally, by the contractive property of morphisms in Banach spaces, one has 
$$
|v \circ (w_1-w_2)| \leq |v| \cdot |w_1-w_2| = |v| \cdot d(w_1, w_2)
$$
Finally, we need to verify the dilation property. Given $c\in [0,1],$ we have act on a morphism by $(c\cdot f)$. If the hom-set is Cauchy complete, it is geometrically complete. It follows that the hom-objects in the category are complete.

To be a morphism in $\DNorm$, composition must be compatible with the tensor product $\tens_{\DNorm}$, which identifies actions of $\mathbb{V}_{int}$.
Explicitly, we require that for any $c \in [0,1]$, $(c \cdot g) \circ f = c \cdot (g \circ f)$ and $g \circ (c \cdot f) = c \cdot (g \circ f)$.
This holds immediately from the bilinearity of composition in Banach spaces (associativity of scalar multiplication).
Thus, the composition map descends to the quotient space defining the tensor product.

This completes the proof.
\end{proofE}
It follows that many well-known categories are dilated.
\begin{example}
The category $\eBan$ of Banach spaces is enriched over itself and, by \Cref{prop:ban_enrich} is a dilated category. The category $\enr{\Hilb}$ of Hilbert spaces is also enriched in $\DNorm$ as it is a full subcategory of $\eBan$. The same is true of the category of finite vector spaces $\eFinVect,$ with the $l_{\infty}$-norm.
\end{example}
There are also many examples that are not $\Ban$-enriched.
\begin{example}
Consider the category $\eCMet$ of complete extended metric spaces. This is a seminorm category as $\CMet$ is enriched over itself. The seminorm of a map $f\colon A\to B$ is defined to be the Lipschitz norm%
$$
\inf \{c \in \mathbb R^+: d(f(x),f(y)) \leq c \cdot d(x, y) \forall x,y \in A \}
$$
Morphisms cannot be rescaled, so therefore it is not a dilated category.
The category of algebras $\ekCMet$ over the monad $\mathbb [0,1] \wedge(-)$\footnote{The monad structure here is the obvious one induced by multiplication.} is $\DNorm$-enriched. Unravelling this, in $\ekCMet$, the objects are pointed complete metric spaces $(M, m_0, d)$ equipped with a rescaling operator $\cdot\colon  \mathbb{R}^+ \times M \to M$ such that; $d(c \cdot x, c \cdot y) = c \cdot d(x, y)$ for all $c \in \mathbb{R}^+$ and $x, y \in M$; the metric $c \cdot f(x) = f(c\dot x)$ and $0\cdot f$ is the constant function at $m_0$. The morphisms then are the metric maps $f\colon  M \to N$ that satisfy the identity $d(c\cdot f(x), c \cdot f(y))= c \cdot   d(x, y) $ for $x,y \in M$ and $c \in \mathbb R^+$.  and for a scalar $c \in \mathbb{R}^+$ and a morphism $f\colon  A \to B$, the map $(c \cdot f)$ is taken pointwise: $(c \cdot f)(x) := c \cdot f(x)$.  
\end{example}
\begin{propositionE}
The category $\eCMet$ is a seminorm category and $\ekCMet$ is a dilated category.
\end{propositionE}
\begin{proofE}
We have already described the enrichment so now we check the axioms.

\textit{Submultiplicative axiom:}
If we have maps $f\colon  A \to B$ and $g\colon  B \to C$, then for any $x, y \in A$, we have
$$
d_C(g(f(x)), g(f(y))) \leq \Lnorm{g} d_B(f(x), f(y)) \leq \Lnorm{g}\Lnorm{f} d_A(x, y).
$$
Taking the infimum over constants satisfying the Lipschitz condition yields $\Lnorm{g \circ f} \leq \Lnorm{g}\Lnorm{f}$.

\textit{Left/right compatibility:}
For left compatibility, let $h \in \eCMet(B, C)$ and $f, g \in \eCMet(A, B)$. Then
\begin{multline*}
d(h \circ f, h \circ g) = \sup_{a \in A} d_C(h(f(a)), h(g(a))) \\ \leq \sup_{a \in A} \Lnorm{h} d_B(f(a), g(a)) = \Lnorm{h} d(f, g).
\end{multline*}
Right compatibility follows by a completely symmetric argument

The proofs of the two above properties for $\ekCMet$ are \textit{mutatis mutandis} the same.

For $\ekCMet$, the enrichment is in $\DNorm$ as the objects are metric spaces equipped with a scaling action of $\mathbb{R}^+$. The hom-objects inherit the metric and seminorm structure from $\eCMet$ but are additionally equipped with a conical action of $\mathbb{R}^+$ given by $(c \cdot f)(x) := c \cdot_B f(x)$.

\textit{Dilation:} We must check that $\Lnorm{c \cdot f} = c \cdot \Lnorm{f}$. Using the property that the action on objects scales distances, $d_B(c \cdot y_1, c \cdot y_2) = c \cdot d_B(y_1, y_2)$, we verify:
$$
d_B((c \cdot f)(x), (c \cdot f)(y)) = c \cdot d_B(f(x), f(y)) \leq c \cdot \Lnorm{f} d_A(x, y).
$$
This implies $\Lnorm{c \cdot f} \leq c \cdot \Lnorm{f}$. The case of equality follows from the minimality of $\Lnorm{f}$.

\textit{Composition:} Now that we have verified that we have the hom objects are enriched, we must check composition is a morphism in the enriching category. Explicitly, we require that for $c \in \mathbb{R}^+$, $f \in \ekCMet(A,B)$ and $g \in \ekCMet(B,C)$, we have $g \circ (c \cdot f) = c \cdot (g \circ f)$.
Evaluating at a point $x \in A$:
$$
(g \circ (c \cdot f))(x) = g((c \cdot f)(x)) = g(c \cdot f(x)).
$$
Recall that objects in $\ekCMet$ are algebras for the scaling monad and morphisms are algebra morphisms, implying that $g$ is equivariant with respect to the action. Thus:
$$
g(c \cdot f(x)) = c \cdot g(f(x)) = (c \cdot (g \circ f))(x).
$$
The compatibility with the left action, $(c \cdot g) \circ f = c \cdot (g \circ f)$, follows immediately from the definition of the action on the hom-object. Thus, the composition map descends to the quotient space defining the tensor product.
\end{proofE}

\subsection*{Changes of basis of enrichment}
We conclude by discussing the relationship between our three types of category; seminorm, dilated and usual categories. Succinctly, this is summed up in the following commutative diagram of forgetful lax monoidal functors, which can be seen as inclusions.
 \begin{equation}
 \label{my_diagram}
 \begin{tikzcd}
 \DNorm \arrow[rd] \arrow[r, " \DtoSNorm"] & \SNorm \dar{}
 \\ 
 \ & \Set
 \end{tikzcd}
 \end{equation}
There is a change of basis from dilated to seminorm categories and also between dilated (or seminorm) categories and $\Set$.
 \begin{theoremE}
 \label{thm:change_of_basis}
 
Every dilated category $\UCat{C}$ is canonically a seminorm category $\DtoSNorm_\ast\UCat{C}$.
\end{theoremE}
\begin{proofE}
It suffices to show that the forgetful functor $\DtoSNorm : \DNorm \to \SNorm$ is a lax monoidal functor. Let $U \colon \DtoSNorm$. We define the structure maps and verify the axioms.

Let $I_{\SNorm} = \{x\}$ be the unit object in $\SNorm$ (where $\abs{x}=e$) and let $I_{\DNorm} = \Vint$ be the unit in $\DNorm.$ Define the morphism $\psi \colon \{x\} \to U(\Vint)$ by setting $\psi(x) = e$. This is a valid morphism in $\SNorm$ because it preserves the seminorm, $|\psi(x)| = |e| = e \le \abs{x}$, and it is clearly non-expansive since $d_{\Vint}(e, e) = \bot$.

Given any $A, B \in \DNorm$, we need a morphism $\phi_{A,B} : U(A) \tens_{\SNorm} U(B) \to U(A \tens_{\DNorm} B)$. The carrier of the domain is $A \times B$, while the codomain is a quotient of $\Vint \times A \times B$. We define the map by $\phi_{A,B}(a, b) = [(e, a, b)]$. We check this is a morphism in $\SNorm$. For the seminorm, observe that $|\phi(a,b)| = |[(e, a, b)]| = e \qMul \abs{a}_A \qMul \abs{b}_B = \abs{a}_A \qMul \abs{b}_B$, which matches the domain. For the metric, let $(a_1, b_1)$ and $(a_2, b_2)$ be pairs. Since $((e, a_1, b_1), (e, a_2, b_2))$ is a valid representative in the pre-quotient space, and $d(e,e)=\bot$, the distance in the quotient is bounded by $\join (d_A(a_1, a_2), d_B(b_1, b_2))$, which is exactly the distance in the domain. Thus $\phi$ is non-expansive and so is a morphism in $\SNorm.$

The verification of associativity and unitality is straightforward.
\end{proofE}
Another important change of basis that we shall make is to $\Set$, ie.\ we obtain standard categories from dilated and seminorm categories.  
 \begin{theoremE}
 \label{thm:Set}
Every dilated (or seminorm) category $\UCat{C}$ has an underlying category, denoted $\Cat{C}$ and taking the underlying category commutes with the change of basis from dilated to seminorm categories.
\end{theoremE}
\begin{proofE}
There is a strong monoidal functor from $\SNorm$ to $\Set$ and a lax monoidal $\DNorm \to \Set.$ Both are given by sending the monoidal product in each category to their carriers in $\Set$. The structure maps%
$$
\phi_{X, Y}\colon \abs{a}\times \abs{b} \to |A \tens_{\SNorm} B|
$$
on $\SNorm$ is an isomorphism. On $\DNorm$ it is given by sending 
$$
(a,b) \mapsto (e, a, b)
$$
For $\SNorm$, the unit map $\{\ast\} \to \{x\}$ is the unique morphism in $\Set$. For $\DNorm$, the map is given by 
$$\{\ast\} \to \Vint \qquad \ast \mapsto e.$$
That the strong monoidal structure diagrams commute for the functor $\SNorm \to \Set$ is clear.

The forgetful functor to $\Set$ factors through $\DtoSNorm$ by (\ref{my_diagram}). Since the composite of lax monoidal functors is lax monoidal it follows that it is lax monoidal.

The fact that taking the underlying category commutes with the change of basis from dilated to seminorm categories follows from the factorisation in (\ref{my_diagram}).
\end{proofE}
\begin{remark}
These change of basis formulae formalise notation that we shall use throughout the text. That is; we can manipulate points in internal-homs of dilated categories as though they were morphisms in a category. There is no ambiguity over when changing base to the underlying seminorm category. A subtlety for conical diagrams to be well defined, the morphisms must all have seminorm less than or equal to $e$ to guarantee the existence of a map $I \to \eCC(X, Y)$ ie.\ from the monoidal unit of the base of enrichment to $\eCC(X, Y)$.
\end{remark}
\section{A categorical Banach fixed point theorem}
\label{sec:catBanach}
\pratendSetLocal{category=catBanach}
We now discuss fixed point theorems for dilated and seminorm categories.
\subsection*{Morphisms of seminorm $\bot$}
Morphisms with seminorm $\bot$ will play an important role in our fixed point theorems, so we study them further.

A \emph{terminal object} in a seminorm category $\eCC$ is an object $\1_{\eCC} \in \eCC$ such that $\eCC(X, \1_{\eCC}) = \{y\}.$ where $\abs{y} = \bot.$ It follows that such an object is metrically right trivial.  An object $X \in \UCat{C}$ is said to be \emph{metrically small} if $\UCat{C}(\1_{\eCC},X)$ is metrically small as a complete $\V$-space.

We say that $\UCat{C}$ is a \textbf{Lipschitz category} if $\abs{\id_{\1_{\eCC}}} = \bot.$
\begin{propositionE}
\label{prop:constant_morph}
Let $\UCat{C}$ be a Lipschitz category. Then, for all $X \in \eCC$ and generalised points $p \in \eCC(\1_{\eCC}, X )$, we have $\abs{f} = \bot.$
\end{propositionE} 
\begin{proofE}
By the sub-multiplicative property of the seminorm, we have:
$$
    \abs{p} = \abs{p \circ \id_{\1_{\eCC}}} \leq \abs{p} \qMul \abs{\id_{\1_{\eCC}}}.
$$
Since $\UCat{C}$ is a Lipschitz category, by definition we have $\abs{\id_{\1_{\eCC}}} = \bot$.
Using the absorbing property of the quantale, we know that $x \qMul \bot = \bot$ for any $x$. Thus:
$$
    \abs{p} \leq \abs{p} \qMul \bot = \bot.
$$
\end{proofE}

We say that a morphism $f \in \UCat{C}(X,Y)$ is \textbf{constant} if factors through the terminal object, and so defines an generalised element of $Y$: $\1_{\eCC} \xrightarrow{p} Y$

Recall that a \emph{well-pointed category} is a category with a terminal object $\1$ and such that that if two morphisms $f,g \colon A \to B$ in the category differ, there must be a generalised point $w\colon \1 \to X$ such that $f \circ w \neq g \circ w.$ Lipschitz categories generalise well pointed categories.
\begin{theoremE}
Let $X \in \eCC$ be a metrically small object in a Lipschitz seminorm category such that $\Cat{C}$ is well-pointed. Then for all $Y \in \eCC$, a morphism $f \colon X\to Y$ is constant if and only if it has seminorm $\bot$.
\end{theoremE}
\begin{proofE}
If $f$ factors as $p \circ g$ through a generalised element $p,$ then $g: X \to \1_{\eCC}$ is the unique map to the constant object. $\bot \qMul \abs{g} = \abs{f} \qMul \abs{g} \geq \abs{f}.$ It follows from the absorbing assumption on quantales that $f$ has seminorm $\bot.$

We must show that if the seminorm of $f$ is $\bot$, it factors through the terminal object.
Consider any pair of generalised points $x, y\colon 1 \to X$.
Since $d(x, y) < \top$ by the metrically small assumption, it follows that by the right compatibility axiom, we have
$$
d(f \circ x, f \circ y) \leq \abs{f} \qMul d(x, y) = \bot \qMul d(x, y) = \bot.
$$
By the Hausdorff property of distance spaces, $d(f \circ x, f \circ y) = \bot$ implies $f \circ x = f \circ y$.
Since this holds for all $x, y$, $f$ is constant on generalised points. Let $y\colon 1 \to Y$ be this constant value (i.e., $f \circ u = y$ for all $u$).

Now consider the factorisation $h = y \circ t_X$, where $t_X\colon X \to 1$ is the unique morphism to the terminal object.
To conclude, we use the well-pointed condition. If $h$ and $f$ differ, they must be distinguished by a generalised point $w\colon 1 \to X$.
However,
$$
h \circ w = y \circ t_X \circ w = y \circ \id_1 = y = f \circ w.
$$
Therefore $h$ and $f$ agree on all points, so $h = f$.
Thus $f$ factors through the terminal object, and so is constant.
\end{proofE}

\subsection*{Fixed points in seminorm categories}

We next prove a generalisation of the Banach Fixed Point Theorem native to dilated and seminorm categories.

\begin{theoremE}
\label{thm:banach_fixed_point}
Let $J \in \UCat{C}$ be metrically small object in a seminorm (or dilated) category such that $\UCat{C}(\1_{\eCC}, J)$ is nonempty. Consider the left composition operator $f_\ast\colon \UCat{C}(\1_{\eCC}, J) \to \UCat{C}(\1_{\eCC}, J)$, where $f \in \UCat{C}(J, J)$ is such that~$\Lnorm{f} <\unit$. Then~$f_\ast$ has a unique fixed point~$\fix(f) \in \UCat{C}(\1_{\eCC}, J)$, i.e.~$f(\fix(f)) = \fix(f)$.
Moreover, for any~$x\in \UCat{C}(\1_{\eCC}, J)$,~$\fix(f)$ is characterised by 
\begin{equation*}
\fix(f) = \lim(f^{n}(x))_{n\in\N}.
\end{equation*}
\end{theoremE}
\begin{proof}
By Definition \ref{def:enrichment},  $\UCat{C}(\1_{\eCC}, J)$, is a complete distance space.  We will show that $f_\ast$ is a contraction mapping on the complete metric space $(\UCat{C}(\1_{\eCC}, J), d)$. Let $x, y\in \UCat{C}(\1_{\eCC}, J)$. Then by the fact $J$ is metrically small, one has $d(x, y) < \top.$ So by Definition \ref{def:enrichment}, we have
$$
d(f_\ast x, f_\ast y) \leq \Lnorm{f} \qMul d(x, y)
$$
It follows that $f_\ast\colon \UCat{C}(\1_{\eCC}, J) \to \UCat{C}(\1_{\eCC}, J)$ has seminorm $\Lnorm{f}$ in $\SNorm$. We can therefore apply  \Cref{thm:fpt} to deduce existence and uniqueness of the desired fixed point. 
\end{proof}

\section{Probability Distributions and convolution}
\label{sec:probabilityandconvolution}
\pratendSetLocal{category=probabilityandconvolution}
We pivot from discussing the enrichment proper towards a treatment of the central limit theorem. As a running example, we introduce  \emph{probability functors on finite dimensional vector spaces}. We also introduce the convolution product from the seminorm perspective. 

\subsection*{Probability functors on vector spaces}

 We discuss a metric structure on the probability measures on finite dimensional (f.d.) vector spaces. In this setting, we fix the quantale $\gRm$ where $[0,1]$ is equipped with the Euclidean metric. 

Recall that every f.d.\ vector space is equipped with the standard topology given by open balls, and hence with the Borel $\sigma$-algebra. We can therefore consider the set of probability measures on it. We equip these with the additional structure of extended metric spaces via the Fourier $l$-distance; for more on this choice of metric see~\cite{gabetta95, HW84, carrillo07}.

\begin{definition}
Let $V$ be a f.d.\ vector space equipped with the $l_{\infty}$-norm. The Fourier $l$-distance $d_{l}$ for $l \geq 1$, is a metric on probability distributions on $V$ given by:
$$ d_{l}(\mu, \tau) =  \left( \sup_{t \in V^\ast, t \neq 0} \frac{\left\|  \phi_{\mu}(t)  - \phi_{\tau}(t)  \right\|}{\|t\|^l} \right)^{1/l}$$
Here $\|-\|$ is the $l_1$ norm on $V^\ast$, which is dual to the $l_{\infty}$-norm on $V$, and $\phi_{\mu}(t) $ denotes $\int_{V} e^{-i\langle t,x \rangle} d\mu$, which is usually called the \emph{characteristic function} of $\mu.$
\end{definition}
The constant $l$ in the definition is very important. For the law of large numbers it will be taken to be in the interval $(1, 2)$ and in $(2,3)$ for the central limit theorem. We shall use this norm to define an extended metric structure on suitably well-defined probability measures on vector spaces.
\begin{propositionE}
\label{prop:metric_properties}
 Consider the collection $\mathcal P_l(V)$ of probability measures defined on a finite dimensional  space $V$ with 
$$
\int_V \|x\|^{l } d\mu < \infty.
$$
Then $(\mathcal P_l(V), d_l)$ is an extended metric space. Moreover,  we have that $d(\mu, \tau) < \infty$ if the measures $\mu, \tau$ satisfy
$$
\int_V \|x\|^{n} d\mu  = \int_V \|x\|^{n} d\tau
$$
for all integers $n <l.$
\end{propositionE}
\begin{proofE}
To simplify the computations, we verify the inequality for the $l$-th power of the Fourier distance. This is an extension of the proof of \cite[Proposition 1]{HW84} from $\mathbb R$ to the general case of finite-dimensional vector spaces.  We remark before a strict verification that essentially all the tools of this proof;  Taylor expansion, L\'{e}vy Continuity Theorem and Fatou's lemma - remain valid in $\mathbb{R}^n$, ensuring the result holds in this broader context..

Let $\phi_\mu$ and $\phi_\tau$ denote the characteristic functions of $\mu$ and $\tau$.
The non-negativity and symmetry of $d_l$ follow directly from the properties of the absolute value and the supremum norm. 
The triangle inequality is immediate: for any $\nu$,
$$
\frac{\|\phi_\mu(t) - \phi_\tau(t)\|}{\|t\|^l} \leq \frac{\|\phi_\mu(t) - \phi_\nu(t)\|}{\|t\|^l} + \frac{\|\phi_\nu(t) - \phi_\tau(t)\|}{\|t\|^l}.
$$
Taking the supremum over $t \neq 0$ yields $d_l(\mu, \tau) \leq d_l(\mu, \nu) + d_l(\nu, \tau)$. 
For the identity of indiscernibles axiom, if $d_l(\mu, \tau) = 0$, then $\phi_\mu(t) = \phi_\tau(t)$ for all $t \neq 0$. By continuity of characteristic functions at the origin, $\phi_\mu(0) = \phi_\tau(0) = 1$, implying $\phi_\mu = \phi_\tau$ everywhere. Since characteristic functions uniquely determine the original measure, one can immediately conclude that  $\mu = \tau$.

The difficult step in the proof is to show that $d_l(\mu, \tau) < \infty$ under the moment condition.
Let $n = \lfloor l \rfloor$. Since $\mu$ and $\tau$ have finite moments of order $l$, their characteristic functions are $n$-times differentiable. 

The condition that moments match for all integers $p < l$  implies that the partial derivatives of the characteristic functions at the origin match. Consequently, the Taylor polynomials $Q_n(t)$ of degree $n$ for $\phi_\mu$ and $\phi_\tau$ centered at $0$ are identical.
Using the Taylor remainder theorem for characteristic functions, we have:
$$
|\phi_\mu(t) - Q_n(t)| \leq C \|t\|^l \int_V \|x\|^l d\mu
$$
for some constant $C$ dependent on $l$. A similar bound holds for $\tau$.
Substituting this into the definition of the metric:
\begin{align*}
\frac{|\phi_\mu(t) - \phi_\tau(t)|}{\|t\|^l} &= \frac{|(\phi_\mu(t) - Q_n(t)) - (\phi_\tau(t) - Q_n(t))|}{\|t\|^l} \\
&\leq \frac{|\phi_\mu(t) - Q_n(t)|}{\|t\|^l} + \frac{|\phi_\tau(t) - Q_n(t)|}{\|t\|^l} \\
&\leq C \left( \int_V \|x\|^l d\mu + \int_V \|x\|^l d\tau \right).
\end{align*}
Since the $l$-th moments are finite by assumption, the ratio is bounded uniformly in $t$. Thus $d_l(\mu, \tau) < \infty$.

If the moment condition is not met, the Taylor expansions differ at some order $k < l$, causing the ratio to diverge as $t \to 0$, resulting in an infinite distance. Thus, $(P_l(V), d_l)$ is an extended metric space as required.
\end{proofE}
There is clearly inclusion of sets $\mathcal P_l(V)  \hookrightarrow \mathcal P_k(V) $ if $ l < k.$
\begin{propositionE}
The space $\mathcal P_n(V)$ is a complete metric space if $n$ is an integer. Moreover the completion of $\mathcal P_l(V)$ is a subspace of $\mathcal P_{\lfloor l \rfloor}(V)$, where $\lfloor l \rfloor$ is $l$ rounded down to an integer.
\end{propositionE}
From the previous result, one may construct a pair of functors.
\begin{proofE}
We follow the proof of \cite[Proposition 2]{HW84}. 

First, assume $n$ is an integer. Let $\{\mu_k\}_{k \in \mathbb{N}}$ be a Cauchy sequence in $(P_n(V), d_n)$.
Let $\phi_k$ denote the characteristic function of $\mu_k$.
For any fixed $\epsilon > 0$, there exists $K$ such that for $k, m \geq K$,
$$
\frac{|\phi_k(t) - \phi_m(t)|}{\|t\|^n} < \epsilon
$$
for all $t \neq 0$. This implies that $\{\phi_k(t)\}$ is a Cauchy sequence of complex numbers for each $t$, so it converges pointwise to a function $\phi(t)$

Since $d_n(\mu_k, \mu_m) < \infty$, the measures share the same moments up to order $n$ (see the proof of Proposition \ref{prop:metric_properties}). Let $Q(t)$ be the Taylor polynomial of degree $n$ corresponding to these moments. Note that $Q(0)=1$ and the higher order terms account for the derivatives at zero.
Define the remainder functions for $t \neq 0$:
$$
h_k(t) = \frac{\phi_k(t) - Q(t)}{\|t\|^n}.
$$
From the Cauchy condition on $d_n$, the sequence $\{h_k\}$ is uniformly Cauchy on $V \setminus \{0\}$. Thus, it converges uniformly to a continuous function $h(t)$.
We can write the limit characteristic function as:
$$
\phi(t) = Q(t) + \|t\|^n h(t).
$$
Since $Q(t)$ is a polynomial and $\|t\|^n h(t) \to 0$ as $t \to 0$ it follows that, $\phi(t)$ is continuous at the origin with $\phi(0)=1$.
By the Lévy Continuity Theorem on finite-dimensional vector spaces, $\phi$ is the characteristic function of a probability measure $\mu$.
Since $\mu_k \to \mu$ weakly, by Fatou's lemma we have:
$$
\int_V \|x\|^n d\mu \leq \liminf_{k \to \infty} \int_V \|x\|^n d\mu_k < \infty.
$$
Thus $\mu \in P_n(V)$. The convergence $d_n(\mu_k, \mu) \to 0$ follows from the uniform convergence of $h_k \to h$. Hence $P_n(V)$ is complete.

The second statement is immediate since a Cauchy sequence in $P_{l}$ is also Cauchy in $P_{\lfloor l \rfloor}$.
\end{proofE}
\begin{propositionE}
\label{prop:eqipped_with_metric}
Let $\Prob(V)$ be the completion of the metric space $P_l(V)$ inside $P_{\lfloor l \rfloor}(V).$
This defines a functor $\Prob\colon  \FinVect \to \kCMet$ and $\ProbZ \colon  \FinVect \to \kCMet$.
\end{propositionE}

\begin{proofE}
Given a function $f\colon  V \to W$ between f.d.\ vector spaces and $l \geq 1$, define $\Prob(f)\colon  \Prob(V) \to \Prob(W)$ to be the function $\Prob(f)(\mu) = [U \mapsto \mu(f^{-1}U)]$ for all $\mu \in  \Prob(V).$  It suffices to verify that the image of a measure with finite $l$-th moments under $\Prob(f)$ stays within also has finite $l$-th moments. We can then complete the metric space and $\Prob(f)$ will extend to the completions.

Let $f\colon  V \to W$ be a morphism in $\eFinVect$ and let $L = \Lnorm{f}$ denote its Lipschitz seminorm. 
By the definition of the seminorm, for any $x \in V$, we have $\abs{f(x)}_W \leq L\abs{x}_V$. 
Raising this to the power of $l \geq 1$, we obtain $\abs{f(x)}_W^l \leq L^l\abs{x}_V^l$.
Integrating with respect to $\mu \in \Prob(V)$ yields:
\begin{multline*}
\int_W \abs{y}_W^l \, d(f_*\mu)(y) = \int_V \abs{f(x)}_W^l  d\mu(x) \le L^l \int_V \abs{x}_V^l d\mu(x)  
\end{multline*}
As $\int_V \abs{x}_V^l \, d\mu(x)$ is finite by the assumption that $\mu \in \Prob(V)$, it follows that $f_*\mu$ has a finite $l$-th moment. 
Thus, the functor $\Prob$ is well-defined on the objects of the category.

This metric space $\Prob$ is equipped with a rescaling operator via the map $\act\colon  (c, \mu) \mapsto \Prob(x \mapsto c \cdot x)(\mu)$ where $\cdot$ is the usual scalar action on $x$ coming from the vector space structure. We must verify that this operation is compatible with the metric when we choose the Fourier $l$-distance ($d_l$). The characteristic function of a rescaled measure $S_c(\mu)$ is $\phi_{S_c(\mu)}(t) = \phi_\mu(ct)$. Therefore, the distance transforms as:
\begin{align*}
d_l(c\act \mu, c\act \tau) &= \left(\sup_{t \neq 0} \frac{|\phi_\mu(c^{-l/l}t) - \phi_\tau(c^{-l/l}t)|}{|t|^l}\right)^{1/l} \\
&= \left(\sup_{u \neq 0} \frac{|\phi_\mu(u) - \phi_\tau(u)|}{|c^{-1/l} t|^l}\right)^{1/l} \quad (\text{letting } u = c^{-1/l}t) \\
&= \left( c^l \sup_{u \neq 0} \frac{|\phi_\mu(u) - \phi_\tau(u)|}{|u|} \right)^l= c\cdot d_l(\mu, \tau)
\end{align*}
The metric space $ \Prob(V)$ thus satisfies the condition $d(c \act x, c\act y) = c \act d(x, y)$ required for the morphisms in $\ekCMet$.
\end{proofE}
\begin{remark}
 The dilation action on $\Prob(V)$ is given by lifting the scalar multiplication using functoriality:
\begin{equation*}
  c \act \mu = \Prob(c \cdot -)(\mu)
\end{equation*}
\end{remark}
Checking the properties of this functor, we can verify the following.
 \begin{propositionE}
\label{prop:monoidal_closed}
The functor $\Prob \colon\eFinVect \to \ekCMet$ is a dilated functor.
\end{propositionE}
\begin{proofE}
Let $f\colon V \to W$ be a linear map between finite-dimensional vector spaces with Lipschitz seminorm $\abs{f}$.  For any two probability measures $\mu, \nu \in \Prob(V)$,we have
$$
d_l(\Prob(f)(\mu), \Prob(f)(\nu)) = \sup_{t \in W \setminus \{0\}} \frac{|\phi_{\Prob(f)(\mu)}(t) - \phi_{\Prob(f)(\nu)}(t)|}{\|t\|^l}
$$
We have  $\phi_{\Prob(f)(\mu)}(t) = \phi_\mu(f(t))$, and we substitute this into the distance formula:
$$
d_l(\Prob(f)(\mu), \Prob(f)(\nu)) = \left( \sup_{t \neq 0} \frac{|\phi_\mu(f(t)) - \phi_\nu(f(t))|}{\|t\|^l} \right)^{1/l}
$$
Let $u = f(t)$. By the definition of the seminorm, we have $\|u\| = \|f(t)\| \leq \abs{f} \|t\|$. For $u \neq 0$, this implies $\|t\| \geq \frac{\|u\|}{\abs{f}}$. Since $l \geq 1$, raising to the power of $-l$ reverses the inequality:
$$
\frac{1}{\|t\|^l} \leq \frac{\abs{f}^l}{\|u\|^l}
$$
Substituting this bound into the supremum, we obtain:
$$
d_l(\Prob(f)(\mu), \Prob(f)(\nu)) \leq  \left( \sup_{u \neq 0} \left( \abs{f}^l \frac{|\phi_\mu(u) - \phi_\nu(u)|}{\|u\|^l} \right) \right)^{1/l}= \abs{f} \cdot d_l(\mu, \nu)
$$
This inequality demonstrates that the map $\Prob(f)$ is Lipschitz continuous with seminorm bounded by $\abs{f}$:
$$
|\Prob(f)| \leq \abs{f}
$$
It follows that $|\Prob(f)| \leq \abs{f}$, verifying that $\Prob$ is a dilated functor.
\end{proofE}
\subsection*{Convolution in Seminorm-categories}
\label{sec:convolution}
There is very natural operation on probability measures $\mu_1, \mu_2$ over a vector space $V$ called the \emph{convolution product} $\mu_1 \ast \mu_2$. This gives the distribution of the sum of two independent random variables drawn from $\mu_1$ and $\mu_2$. Explicitly, this is given by the formula
\begin{equation}
\label{eq1}
(\mu_1\ast \mu_2)(U) := \int_{V\times V} \chi_U(x+y)d\mu_1(x)d\mu_2(y)
\end{equation}
where $\chi_U$ is the indicator function for the the open set $U \subset V$, see \cite[Section 20]{Billingsley12} for more details.
We turn our attention to isolating the essential structural properties of this convolution product in the context of seminorm categories. 

To keep the exposition accessible, we take a concrete approach to enriched monoidal categories, introducing only the minimum machinery required for our results.  A \emph{monoidal dilated category} (resp.\ \emph{monoidal seminorm category} ) $(\UCat{C}, \ptens)$ is a dilated category (resp.\ seminorm category) with a tensor product $\ptens$ where all of the associator and unitor maps all have seminorm $e$ and are therefore nonexpansive. This is easily proven for all of the examples that follows, as these monoidal structures are Cartesian, and structure maps are isometries in the underlying category and hence have seminorm $e$.

A monoidal dilated (or seminorm) category $\UCat{C}$ is \emph{pointwise additive} if each object $X \in\ob{\UCat{C}}$ comes equipped with an \emph{addition morphism} $+ \from  X \ptens X \to X$ in the underlying concrete category. This morphism is not required to satisfy any properties with respect to the enrichment.
 
\begin{example}
In $\eFinVect$ the monoidal structure is the Cartesian product. The associator and unitor maps are all isometries and hence have seminorm $e$.  The pointwise addition is the usual vector addition. Note that the addition morphism here is expansive: $+ \from  X \ptens X \to X$ has seminorm 2.
\end{example}
\begin{example}
In $\eCMet$ the monoidal structure is induced by the usual Cartesian monoidal structure on $\CMet$, that is, the product is equipped with the maximum metric. As all the associator and unitor maps are all isometries, they have seminorm $e$ as required. The category $\ekCMet$ is an Eilenberg-Moore category of a adjoint monad on a symmetric monoidal category with coequalisers and so by \cite[Theorem 2.1]{Keigher78} it is itself symmetric monoidal.
\end{example}
The \emph{linear maps} in a pointwise additive seminorm category are precisely the maps that are equivariant with respect to the addition morphisms
  ie.\ morphism $f \in \UCat{C}(A, B)$ is \emph{linear} if the following diagram commutes:
  \begin{equation*}
    \begin{tikzcd}
      A \ptens A \arrow[r, "+"] \arrow[d, "f \ptens f"'] & A \arrow[d, "f"] \\
      B \ptens B \arrow[r, "+"] & B
    \end{tikzcd}
  \end{equation*}
  The \emph{linear subcategory} of $\UCat{C}$, denoted $\lin{\UCat{C}}$, is the seminorm category with the same objects as $\UCat{C}$, but whose morphisms are the linear morphisms of $\UCat{C}$. We note that $\lin{\UCat{C}}$ exists, as it can be constructed using equalisers in $\DNorm$.

\begin{example}
The subcategory $\lin{\eEBan}$ has the same objects, but its morphisms are the additive maps between Banach spaces. The maps are already assumed homogeneous, so $\lin{\eEBan}$ is the category of Banach spaces and bounded linear maps between them.
\end{example}
  Let $\UCat{C}$ and $\eDC$ be two monoidal seminorm-categories.
  A \emph{(lax) monoidal dilated functor} $\UCat{C} \to \eDC$ is a triple $(\enr{F}, \eta, \mu)$, where $\enr{F} \from \UCat{C} \to \eDC$ is an dilated functor and $(F, \eta, \mu)$ is a lax monoidal functor $\CC \to \DC$ between the underlying monoidal categories such that for all $X, Y \in \UCat{C}$, we have $| \mu_{X, Y}|_{\UCat{D}} \le e$ and $|\eta_V|_{\UCat{D}} \le e$.
  
 This definition of lax monoidal enriched functor is a special case of the usual one.
 
  \begin{propositionE}
  \label{Prob_is_monoidal_closed}
For $l \geq 1,$ the functor
$$
\left(\Prob \colon\eFinVect \to \ekCMet,  \operatorname{ind}, u\right)
$$
is a monoidal dilated functor, where the lax monoidal map is given by the product of independent probability measures:
$$ \operatorname{ind}\colon  \Prob V \times \Prob W \to \Prob(V \times W)$$
$$(\mu_1, \mu_2) \mapsto [U \times V \mapsto \mu_1(U)\mu_2(V)],$$
 and the unit map $u\colon  1 \to \Prob(1)$ is given by $0 \mapsto \delta_0$
where $\delta_0$ is the Dirac delta distribution on $0$.
  \end{propositionE}
\begin{proofE}
  To establish that $(\Prob, \operatorname{ind}, u)$ is a lax monoidal dilated functor, we must verify the associativity and unitality axioms for the lax structure and show that the structural morphisms are non-expansive.

  We first show that the multiplication $\operatorname{ind}$ is associative. This amounts to proving that the following diagram commutes for any $U, V, W \in \eFinVect$
  \begin{equation*}
    \begin{tikzcd}
      \Prob(U) \times \Prob(V) \times \Prob(W) 
        \arrow[r, "\operatorname{ind}_{U,V} \times \id"] 
        \arrow[d, "\id \times \operatorname{ind}_{V,W}"'] 
      & \Prob(U \times V) \times \Prob(W) 
        \arrow[d, "\operatorname{ind}_{U \times V, W}"] 
      \\
      \Prob(U) \times \Prob(V \times W) 
        \arrow[r, "\operatorname{ind}_{U, V \times W}"] 
      & \Prob(U \times V \times W)
    \end{tikzcd}
  \end{equation*}
For convenience, we have omitted the associators of the Cartesian product.
  
  Let $(\mu_U, \mu_V, \mu_W)$ be probability measures on $U, V, W$ respectively.
  Following the top path, we first form the product measure $\mu_{U} \tens \mu_{V}$ on $U \times V$, and then form $(\mu_{U} \tens \mu_{V}) \tens \mu_{W}$ on $(U \times V) \times W$
  Following the bottom path, we obtain $\mu_{U} \tens (\mu_{V} \tens \mu_{W})$.
  By the associativity of the product measure construction, ie.\ multiplication on $[0,1]$ these define the same measure on the product space $U \times V \times W$. Thus the diagram commutes.

 Let $1 = \{0\}$ be the zero vector space, which acts as the unit for the Cartesian monoidal structure on $\eFinVect$. The unit map $u\colon  1 \to \Prob(1)$ is given by $0 \mapsto \delta_0$ We must show compatibility with the left unitor $\lambda\colon  1 \times V \to V$. The required diagram is:
  \begin{equation*}
    \begin{tikzcd}
      1 \times \Prob(V) \arrow[r, "u \times \id"] \arrow[dr, "\lambda_{\Prob(V)}"'] 
      & \Prob(1) \times \Prob(V) \arrow[r, "\operatorname{ind}_{1, V}"] 
      & \Prob(1 \times V) \arrow[d, "\Prob(\lambda_V)"] 
      \\
      & \Prob(V) \arrow[r, "\id"] & \Prob(V)
    \end{tikzcd}
  \end{equation*}
  Let $\mu \in \Prob(V)$. The top path maps $(0, \mu)$ to $\delta_0 \tens \mu$. The pushforward of this product measure along the isomorphism $\lambda_V\colon  \{0\} \times V \cong V$ is exactly $\mu$, since for any measurable set $A \subseteq V$:
  $$
  (\Prob(\lambda_V)(\delta_0 \tens \mu))(A) = (\delta_0 \tens \mu)(\{0\} \times A) = \delta_0(\{0\}) \cdot \mu(A) = \mu(A).
  $$
  where we use $\delta_0(\{0\})= 1$.
  The verification for the right unitor is exactly the same.
  
  Finally, we need to show that the lax monoidal map $\operatorname{ind}$ has seminorm bounded by $1$. 
Let $(\mu_1, \nu_1)$ and $(\mu_2, \nu_2)$ be elements of $\Prob(U) \times \Prob(V)$.
Recall that the metric on the product space is the maximum metric
$$
D := d_{\Prob(U) \times \Prob(V)}\left((\mu_1, \nu_1), (\mu_2, \nu_2)\right) = \max\left( d_l(\mu_1, \mu_2), \, d_l(\nu_1, \nu_2) \right).
$$
We can bound the Fourier $l$-distance in the codomain:
\begin{equation}
\label{bound}
d_{l}(\mu_1 \tens \nu_1, \mu_2 \tens \nu_2) = \left( \sup_{(u, v) \neq 0} \frac{\left| \phi_{\mu_1}(u)\phi_{\nu_1}(v) - \phi_{\mu_2}(u)\phi_{\nu_2}(v) \right|}{\left|(u, v)\right|^l}. \right)^{1/l}
\end{equation}
Here, we are implicitly using Fubini's theorem to separate the integrals. Specifically, for the product measure $\mu_1 \tens \nu_1$ on $V \times W$, the characteristic function is:
\begin{multline*}
\phi_{\mu_1 \tens \nu_1}(u, v)= \int_{V \times W} e^{-i(\langle U, Z \rangle)} d(\mu_1 \tens \nu_1)(Z)
\\
 = \int_{V \times W} e^{-i(\langle u, x\rangle + \langle v, y\rangle)} d(\mu_1 \tens \nu_1)(x, y)
\end{multline*}
Here, we decomposed the inner product $\langle U, Z\rangle$ in $V\times W$ as follows. First, we split $U = u+v$ where $u\in V\times 0_W$ and $v = 0_{V}\times W.$ Then we can split $Z = x+y$ in the same way. Expanding this and noticing orthogonal inner products vanish, we obtain $\langle u, x\rangle + \langle v, y\rangle) = \langle U, Z\rangle$.
By Fubini's theorem, we can split this double integral:
\begin{align*}
\phi_{\mu_1 \tens \nu_1}(u, v) 
&= \int_V \int_W e^{-i\langle u, x\rangle} e^{-i\langle v, y\rangle} \, d\nu_1(y) \, d\mu_1(x) \\
&= \int_V e^{-i\langle u, x\rangle} d\mu_1(x) \cdot \int_W e^{-i\langle v, y\rangle} d\nu_1(y) \\
&= \phi_{\mu_1}(u) \cdot \phi_{\nu_1}(v).
\end{align*}
 We apply the triangle inequality the numerator of the right hand side of (\ref{bound}), and use the bound $|\phi| \leq 1$ (the usual one for characteristic functions, but proven in the proof of Theorem \ref{prop:LLN}):
\begin{multline*}
\left| \phi_{\mu_1}\phi_{\nu_1} - \phi_{\mu_2}\phi_{\nu_2} \right| 
\\ \leq |\phi_{\mu_1}(u)| \cdot |\phi_{\nu_1}(v) - \phi_{\nu_2}(v)| + |\phi_{\nu_2}(v)| \cdot |\phi_{\mu_1}(u) - \phi_{\mu_2}(u)| \\
\leq 1 \cdot d_l(\nu_1, \nu_2)^l |v|^l + 1 \cdot d_l(\mu_1, \mu_2)^ l|u|^l 
\\ \leq D^l \left( |u|^l + |v|^l \right).
\end{multline*}
To go from the second to the third line, we used the inequality
$$
 |\phi_{\nu_1}(v) - \phi_{\nu_2}(v)| \leq d_l(\nu_1, \nu_2)^l|v|^l 
$$
which we obtain by raising Fourier distance to the power of $l$ and observing that any value is less than the supremum.

Substituting this calculation into the Equation (\ref{bound}), we can compute the supremum of the the left hand side as \textbf{(here we use the fact that we have the $l_1$-norm on the dual space)}:
$$
\frac{d_{l}(\mu_1 \tens \nu_1, \mu_2 \tens \nu_2)}{D} \leq \left( \sup_{(u,v) \neq 0} \frac{|u|^l + |v|^l}{(|u|^l+ |v|^l)} \right)^{1/l}
$$
so it follows that $|\operatorname{ind}| \leq 1$.
\end{proofE}
Given a lax monoidal dilated functor $(\enr{F} \from  \UCat{C} \to \eDC, \eta, \mu)$ where the domain $\UCat{C}$ is pointwise additive category, the \emph{convolution product} $\conv_{X}$ is the pushforward of the addition on $\UCat{C}$ to a product on the essential image of $F$ in $\eDC$, given by
  \begin{equation*}
    \conv^{\enr{F}}_{X} \colon  \enr{F}X \ptens \enr{F}X \xrightarrow{\mu_{\enr{F}X, \enr{F}X}} \enr{F}(X\ptens X) \xrightarrow{\enr{F}(+_X)} \enr{F}(X).
  \end{equation*}
  \begin{propositionE}
  \label{prop:axiomisedconvolution}
 Let  $\Prob\colon  \eFinVect \to \ekCMet$. Given a finite dimensional vector space $V$, the above convolution product
$$
\Prob(V) \times \Prob(V) \to  \Prob(V)
$$
is given by the formula
$$
(\mu_1\ast \mu_2)(U) = \int_{(x,y) \in V\times V} \chi_U(x+y)d\mu_1(x)d\mu_2(y).
$$ 
  \end{propositionE}
  \begin{proofE}
 The tensor product in $\eEBan$ is $\times.$ By definition, the functor 
$$
\Prob(V) \times \Prob(V)  \xrightarrow{} \Prob(V\times V)
 $$
has the formula
$$
(\mu_1, \mu_2) \mapsto \left[ U \mapsto \int_{(x,y) \in V\times V} \chi_U(x,y) \mu_1(x)\mu_2(y) \right]
$$
for $U \subseteq V\times V $
Similarly, 
 $$
 \Prob (V\times V) \xrightarrow{\Prob(V) (+_V)} \Prob (V).
 $$
 has the formula
 $$
 \mu \mapsto \left[ U \mapsto  \int_{(x,y) \in V\times V} \chi_U(x+y)d\mu_1(x, y)d\mu \right]
 $$
 for $U \subseteq V$.
 Combining both, we obtain the desired formula for the convolution product.
\end{proofE}
   For each $X\in \UCat{C}$ and some choice of \emph{diagonal map} $\Delta_F \to 
   \Delta_F \times  \Delta_F$, we may define the \emph{$\enr{F}$-convolution operator} as the pre-composition of this with a diagonal map.

For each $X$ in a dilated category $\UCat{C}$ and some choice of \emph{diagonal map} $\Delta_{\enr{F}} \to \Delta_{\enr{F}} \ptens \Delta_{\enr{F}}$, we may define the \emph{$\enr{F}$-convolution operator} denoted $\theta^{\enr{F}}_X$ as the pre-composition of this with a diagonal map:
\begin{equation*}
\theta^{\enr{F}}_X\colon \enr{F}X \xrightarrow{\Delta_{\enr{F}} } \enr{F}X\ptens \enr{F}X\xrightarrow{\mu_{\enr{F}X, \enr{F}X}} \enr{F}(X\ptens X) \xrightarrow{\enr{F}(+_X)} \enr{F}X.
\end{equation*}
If $[\abs{\theta^{\enr{G}}_X},e] \act \theta^{\enr{G}}_X = \id_{ \enr{G} X}$, we say that the $\enr{G}$-convolution operator is \emph{perfectly rescalable}. To lighten the visual burden of notation, we shall normally denote the seminorm $\abs{\theta^{\enr{G}}_X} = c \in \V$ and call it the \emph{grading constant}. Note, however, that it is an invariant of $\enr{G}$ and not additional data.

\begin{example}
Consider the Cartesian diagonal on $\ekCMet$, the functor $\lvert-\rvert$ is perfectly rescalable with grading constant $2$. To see this, observe that the convolution operator:%
\begin{equation*}
\theta^{|-|}_V:   | V|  \xrightarrow{\Delta} |V |  \times | V | \xrightarrow{\Delta} |V\times V | \xrightarrow{|+|} |V |
\end{equation*}
simply takes an element $x \in |V|$ and returns  $2x \in |V|$.  So if we rescale by $\frac{1}{2} = [2, 1]$, we get the identity map. Similarly, if we consider the functor $\PosDef$, a similar calculation gives a matrix $M$ is sent to a matrix $\sqrt{2}M.$ So it has grading constant $\sqrt{2}$.

\end{example}
   \subsection*{Expectation}
We are going to describe a categorification of the classical notions of expectation of a probability distribution.

Let $|-| \colon \eFinVect \to\ekCMet$ be the forgetful functor, $|V| = V$.  This is a monoidal enriched dilated functor as the carrier of all the monoidal structure maps on $\FinVect$ is the diagrams in $\kCMet$.
 \begin{example}
The addition on  $\FinVect $ is sent to the carrier of the addition as a metric map in $\kCMet.$
\end{example}
Recall every probability measure $\mu$ on a f.d.\ vector space $V$  has expected value
$$
\exp[\mu] = \int_V x d\mu(x) \in V.
$$
It is straightforward to show the following.
\begin{proposition}
For $l \in (1,2),$ the map $\exp_V\colon  \Prob V \to |V|$ built by sending a probability measure $\mu$ to its expectation defines a natural transformation from $\Prob \to |-|$. We call this the \textbf{expectation natural transformation}.,
  \end{proposition}
 \subsection*{Variance}

Recall from basic probability theory that the \textbf{variance matrix} of a probability measure on $\mathbb R^n$ with expectation 0 is the matrix $\Var_V(\mu)$ with entries given by 
$$ (\Var_V(\mu))_{ij} = \int_V x_i x_j d\mu.$$
This is a \textbf{symmetric positive semidefinite matrix.} We shall use this to define the $\ProbZ$ functor. This is a restriction of the $\Prob$ functor which assigns to finite vector spaces $V$, the set of probability measures over $V$ which have expectation zero and 
$$
\int \|x \| ^{\lambda} d\mu < \infty \mbox{ for some } \lambda > 2.
$$
This functor has a grading over the functor $\PosDef\colon \eFinVect \to \ekCMet$ which assigns to every vector space $V$ the set of \textbf{symmetric positive semidefinite matrices} $\PosDef(V)$ on it and acts on linear maps as $\PosDef(f) = f(-)f^T$. The metric on  $\PosDef(V)$ is the Bures-Wasserstein distance~\cite{bhatia19} given by the formula
$$
d(A,B) = \left( \operatorname{tr}(A) +\operatorname{tr}(B) -2 \operatorname{tr}\left(A^{1/2}B A^{1/2}\right)^{1/2}\right)^{1/2}
$$
where $A^{1/2}$ is the unique positive semidefinite square root of $A$.
For example, in the one dimensional case, this is $d(x,y) = |\sqrt{x}-\sqrt{y} |$.
The set $\PosDef(V)$ admits an action of $\mathbb R^+$ as $r \act M = r^2 \cdot M$, where $\cdot$ is the usual scalar action on matrices.  It has a monoidal structure given by forming sending $(M,N)$ to the matrix with top-left-corner $M$ and bottom-right corner $N$.
\begin{proposition}
With the definitions as above, $\PosDef(V)$ is an element of $\ekCMet$.
\end{proposition}
\begin{proofE}
It suffices to verify that
$$
d(c\act A, c\act B) = c \cdot d(A, B)
$$
One has the general formula $c \operatorname{tr}(M) = \operatorname{tr}(c\cdot M)$. Therefore
\begin{multline*}
d(c\act A,c\act B) = \left(c^2 \operatorname{tr}(A) + c^2\operatorname{tr}(B) -2 c^2 \operatorname{tr}\left(A^{1/2}B A^{1/2})^{1/2}\right)\right)^{1/2} \\= c\left( \operatorname{tr}(A) +\operatorname{tr}(B) -2 \operatorname{tr}\left(A^{1/2}B A^{1/2})^{1/2}\right)\right)^{1/2}
\end{multline*}
as required.
\end{proofE}
In fact, it defines an monoidal dilated functor.
\begin{propositionE}
The functor $\PosDef: \FinVect \to \kCMet$ extends to a lax monoidal dilated functor $\enr{\PosDef}(V): \eFinVect \to \ekCMet$.
\end{propositionE}
\begin{proofE}
To establish this, we first show that $\PosDef$ is a dilated functor and then it carries a lax monoidal dilated structure.

We must verify the inequality $\abs{\PosDef(f)} \leq \abs{f}$. Let $M, N \in \PosDef(V)$.
The Bures-Wasserstein distance satisfies
$$
d(f M f^T, f N f^T) \leq \abs{f} \cdot d(M, N)
$$
where $\abs{f}$ is the Lipschitz seminorm of $f$. This is because (see  \cite{bhatia19}) it is the same as the Wasserstein distance between the normal distributions with $M$ and $N$ as variance matrices. To conclude, observe (or see \cite{tommy1996q}) that the Wasserstein distance satisfies the inequality
$$
W_p(f_\ast \mu, f_\ast \nu) \leq \abs{f}_{op} \cdot W_p(f_\ast \mu, f_\ast \nu).
$$
for all $p$ and note the pushfoward of a Gaussian along a linear map transforms its variance matrix as $fMF^T.$
Thus, we have:
$$
|\PosDef(f)| \leq \abs{f}.
$$
Therefore $\PosDef$ is a dilated functor.

We have not yet defined the lax monoidal structure maps for $\PosDef$ so we do that next. 
The unit in $\eFinVect$ is the zero vector space $\{0\}$. The space $\PosDef(\{0\})$ is the trivial one-point metric space. There is thus a unique map $w\colon \1_{\ekCMet} \to \PosDef(\{0\})$, which we take to be the unit. 

  The monoidal map $\mu_{V,W}\colon \PosDef(V) \times \PosDef(W) \to \PosDef(V \times W)$ is the block diagonal embedding:
    $$
    (M, N) \mapsto \begin{pmatrix} M & 0 \\ 0 & N \end{pmatrix}
    $$
The distance between block diagonal matrices splits as follows as we can decompose the top and bottom corners separately:
$$
d\left(\begin{pmatrix} M_1 & 0 \\ 0 & N_1 \end{pmatrix}, \begin{pmatrix} M_2 & 0 \\ 0 & N_2 \end{pmatrix}\right)^2 = d(M_1, M_2)^2 + d(N_1, N_2)^2
$$
This means that as
$$
d(M_1, M_2)^2 + d(N_1, N_2)^2 \leq 2 \max(d(M_1, M_2)^2, d(N_1, N_2)^2)
$$
 we have that the map is bounded by $\sqrt{2}$.
\end{proofE}
The variance map connects probability measures and the underlying vector space \emph{functorially}.
 \begin{theoremE}
\label{lem:var_lipschitz}
For $l \in (1,2)$, the variance map $\Var \colon  \ProbZ \to \PosDef$ defines a natural transformation. 
\end{theoremE}
\begin{proofE}
For now we check that $\Var_V\colon  \ProbZ(V) \to \PosDef(V)$ is a morphism in $\ekCMet$ and $\abs{\Var_V} \leq e$ and defer the proof of naturality to Proposition \ref{prop:gradedpair}. We need to check equivariance and non-expansiveness.

Equivariance is trivially satisfied as for any probability measure $\mu \in \ProbZ(V)$,  the measure $c \star \mu$ has, by  the standard properties of variance:
$$
\Var_V(c \star \mu) = \Var(cX) = c^2 \Var(X) = c^2 \Var_V(\mu) = c \star \Var_V(\mu).
$$
Thus, the map is equivariant.

Second, we check non-expansiveness. We consider the case $l > 2$ used for the Central Limit Theorem.
Recall from Proposition \ref{prop:metric_properties} that for the Fourier $l$-distance $d_l(\mu, \nu)$ to be finite, the measures $\mu$ and $\nu$ have to share the same moments up to order $\lfloor l \rfloor$. Therefore they are sent to the same matrix. If $d_l(\mu, \nu) = \infty$, the non-expansiveness condition holds vacuously. Thus, $\Var_V$ is a morphism in $\ekCMet$.

To compute $\abs{\Var_V}$, we note that it sends everything at finite distance apart to distance 0. So it has seminorm $0$.
\end{proofE}

\section{Central Limits via Dilations}
\label{sec:CLT}
We now turn to setting up the main theorem of the paper. 
\pratendSetLocal{category=CLT}
\subsection*{Functor gradings and pre-CLT structures}
Expectation and variance have a same layered structure with respect to the probability functor. This can be described in the more abstract language of dilated functors. Fix a pair of lax monoidal dilated functors 
\begin{equation*}
\enr{F},\enr{G}\colon\UCat{C}\to \UCat{D}
\end{equation*} 
between monoidal dilated categories $\UCat{C}$ and $\UCat{D}$ unless explicitly mentioned otherwise. 
We assume that $\UCat{C}$ has point-wise addition and that the underlying category of $\UCat{D}$ has enriched binary pullbacks and a terminal object~$\1_{\Cat{D}}$.
From this we shall develop a general notion of a \emph{central limit.} 

 \begin{definition}
The \textbf{(enriched) fibre} of a morphism $f\colon X\to Y$ in~$\UCat{D}$ with $\abs{f}\le e$ over a generalised point $y: 1_{\eCC} \to Y$, denoted $X_y$, is the pullback in the following conical diagram.
\begin{equation*}
\begin{tikzcd}
    X_y \arrow[r] \arrow[d] & \1_{\UCat{D}} \arrow[d, "y"] \\
    X \arrow[r, "f"] & Y
\end{tikzcd}
\end{equation*}
This is well defined as a conical limit as generalised points have seminorm $\abs{p}= \bot.$
\end{definition}

We briefly note that the limits in seminorm categories are defined up to isometry as the maps in both $\SNorm$ and $\DNorm$ are non-expansive and have the property that $\Lnorm{f(x)}\le \Lnorm{x}$. Therefore the fibre is defined up to isometry and the seminorm of maps between fibres is defined up multiplication by  $e$, that is, on the nose.
We phrase this as a proposition.

\begin{propositionE}
\label{my_prop}
Let $X_y$ be the enriched fibre over $f\colon X\to Y$. 
Then
\begin{equation*}
  \UCat{C}(\1_{\eCC}, X_y)\cong\{x \in \Cat{C}(\1_{\eCC}, X)\mid p(x) = y \mbox{ in } \Cat{C}(\1_{\eCC}, Y)\},
\end{equation*}
where the latter is equipped with the subspace metric and subspace seminorm from~$X$.
\end{propositionE}
\begin{proofE}
The functor $\UCat{C}(\1_{\eCC}, -)\colon  \UCat{C} \to \SNorm$ (or $\DNorm$) preserves all enriched limits. Recall that the enriched fibre $X_y$ is defined as the enriched pullback of the diagram $\1_{\eCC} \xrightarrow{y} Y \xleftarrow{p} X$.Applying the functor $\UCat{C}(\1_{\eCC}, -)$ to this diagram yields a limit diagram in the base category $\SNorm$:
\begin{equation*}
\begin{tikzcd}
    \UCat{C}(\1_{\eCC}, X_y) \arrow[r] \arrow[d] & \UCat{C}(\1_{\eCC}, X) \arrow[d, "p_*" ] 
    \\
    \UCat{C}(\1_{\eCC}, \1_{\eCC}) \arrow[r, "y_*"] & \UCat{C}(\1_{\eCC}, Y)
\end{tikzcd}
\end{equation*}
We compute this limit in $\SNorm$. Since $\1_{\eCC}$ is the terminal object, $\UCat{C}(\1_{\eCC}, \1_{\eCC}) \cong \{x\}$ where $\abs{x} = \bot.$ The morphism $y_*\colon  \{x\} \to \UCat{C}(\1_{\eCC}, Y)$ picks out the point $y$.
 In the categories $\SNorm$, the limit of the cospan $I \xrightarrow{y} B \xleftarrow{p} A$ is calculated as the set-theoretic fibre $\{a \in A \mid p(a) = y\}$ equipped with the subspace metric and subspace seminorm induced from $A$.

Therefore, $\UCat{C}(\1_{\eCC}, X_y)$ is isomorphic in $\SNorm$ to the subspace $\{x \in \UCat{C}(\1_{\eCC}, X) \mid p(x) = y \mbox{ in }  \UCat{C}(\1_{\eCC}, Y)  \}$.Since isomorphisms in $\SNorm$ (and $\DNorm$) are isometries, this establishes the equality up to isometry as required.
\end{proofE}
\begin{definition}
\label{def: grading}
  A \textbf{grading} of $\enr{F}$ with respect to $\enr{G}$, is a natural transformation $p \from \enr{F} \to \enr{G}$ that commutes with lax monoidal maps and the sum, that is, the following two diagrams commute, and such that each $p_{X}$ is an epimorphism such that $\abs{p_{X}} \le e$.
  \begin{equation*}
    \begin{tikzcd}[column sep=large]
      \enr{F}_{lin}(X)\ptens \enr{F}_{lin}(Y) \arrow[d, "\mu(F)_{X,Y}"] \arrow[r, "p_{X}\ptens p_{Y}"]
      & \enr{G}_{lin}(X)\ptens \enr{G}_{lin}(Y) \arrow[d, "\mu(G)_{X,Y}"]
      \\
      \enr{F}_{lin}(X\ptens Y) \arrow[r, "p_{X\ptens Y}"]
      & \enr{G}_{lin}(X \ptens Y)
    \end{tikzcd}
    \end{equation*}
        and 
    \begin{equation*}
    \begin{tikzcd}
      \enr{F}_{lin}(X\ptens X) \arrow[d, "\enr{F}(+)"] \arrow[r, "p_{X\ptens X}"] & \enr{G}_{lin}(X \ptens X) \arrow[d, "\enr{G}(+)"]
      \\
      \enr{F}_{lin}(X) \arrow[r, "p_{X}"] & \enr{G}_{lin}(X).
    \end{tikzcd}
  \end{equation*}
  Here, $\mu(F)$ and $\mu(G)$ are the coherence maps for the lax monoidal functors $F$ and $G$.   
   A \textbf{grading diagonal} is a pair of natural transformations, $\Delta_{\enr{F}}\colon  F \to F \ptens F$ and $\Delta_{\enr{G}}\colon  G \to G \ptens G$ along with a modification $\mu\colon  \Delta_{\enr{F}} \to \Delta_{\enr{G}}.$ 
  \end{definition}
 \begin{example}
 The probability functor has two gradings.
 \begin{enumerate}
 \item
 For $l\in (1,2),$ the functor $\enr{\Prob}$ is graded with respect $\enr{|-|}.$ The grading is given by the map $\exp_V: \Prob(V) \to |V|$ which assigns its expectation to every probability measure in  $\Prob(V)$. 

 \item
 For $l\in (2,3),$ the functor $\enr{\Prob}$ is graded with respect $\enr{\PosDef}.$ The grading is given by the map $\Var_V: \Prob(V) \to \enr{\PosDef}(V)$ which assigns its variance matrix to every probability measure in  $\Prob(V)$. 
 \end{enumerate}
 In both cases, the operation that is being pushed forward is the addition on the vector space.
 \end{example}
  \begin{remark}
  The seminorm of $+$ is not required to have a seminorm less than $e$ and generally will not. For example, the seminorm of addition in $\eEBan$ is 2. This is the origin of the rescaling factor in the central limit theorem and law of large numbers. 
  \end{remark}

We arrive at our main structural definition.

\begin{definition}
A \emph{pre-CLT system} $(\enr{F}, \enr{G}, \Delta_{\enr{F}}, \Delta_{\enr{G}}, p, \mu)$ consists of data   
\begin{enumerate}
\item
a grading of $\enr{F}$  by $p$ with respect to a perfectly rescalable $\enr{G}.$
\item 
a grading diagonal $(\Delta_{\enr{F}}, \Delta_{\enr{G}}, \mu)$ such that for all $X\in \UCat{C}$ the following commutes
\begin{equation*}
\begin{tikzcd}
\enr{F}X \arrow[r, "\Delta_{\enr{F}}"] \arrow[d, "p_X"]& \enr{F}X\ptens \enr{F}X \arrow[d, swap, "p_X\ptens p_X"]
\\
\enr{G}X  \arrow[r, "\Delta_{\enr{G}}"]& \enr{G}X\ptens \enr{G}X.
\end{tikzcd}
\end{equation*}
\end{enumerate}
We say  $(\enr{F}, \enr{G}, p)$ is a \emph{Cartesian pre-CLT system} if $(\enr{F}, \enr{G}, \Delta_{\enr{F}}, \Delta_{\enr{G}}, \mu)$ is a pre-CLT system and $ \Delta_{\enr{F}}, \Delta_{\enr{G}}, \mu$ are induced by the Cartesian diagonal on $\UCat{D}.$
\end{definition}
Our main examples are pre-CLT-systems, the proofs are rather long, but essentially diagram chasing.
\begin{propositionE}
Both of these are Cartesian pre-CLT-systems:
\begin{enumerate}
\item
the tuple $(\Prob, |-|, \exp)$.
\item
the tuple $(\ProbZ,\PosDef, \Var)$.
\end{enumerate}
\end{propositionE}
\begin{proofE}
See the proof of Proposition \ref{preCLT_LLN} for (1) and Proposition \ref{prop:gradedpair} for (2).
\begin{propositionE}
\label{preCLT_LLN}
 The tuple  $\left(\Prob, |-|, \exp\right)$ is a Cartesian pre-CLT system.
\end{propositionE}
\begin{proofE}
The functor $|-|$ is monoidal enriched. The monoidal axioms follow from the monoidal structure of $\FinLip$. The map is metric-preserving and thus enriched.

Next we establish that that $(\Prob, |-|)$ is a graded pair, we must verify that the expectation map, $\exp$, is a natural transformation with respect to linear subcategory of $\lin{\eFinVect}$.
 
To begin, we need to establish the naturality of the expectation map. The map $\exp\colon  \Prob \to U$ is a natural transformation if for every morphism $f\colon  V \to W$ in the category, the following diagram commutes:
\begin{equation*}
\begin{tikzcd}
    \Prob(V) \arrow[r, "\exp_V"] \arrow[d, "\Prob(f)"'] & \lvert V \rvert \arrow[d, "f"] \\
    \Prob(W) \arrow[r, "\exp_W"] & \lvert W \rvert
\end{tikzcd}
\end{equation*}
The component of the transformation at an object $V$, denoted $\exp_V$, maps a probability measure $\mu \in \Prob(V)$ to its expected value, $E[\mu] \in V$. The action of the functor $\Prob$ on a morphism $f$ is the pushforward map, so $\Prob(f)(\mu) = f_*\mu$.

For the diagram to commute, we must have $$f(\exp_V(\mu)) = \exp_W(\Prob(f)(\mu))$$ for any measure $\mu$. We compute both sides of this equality. 
The left hand-side becomes $f(\exp_V(\mu)) =  f\left(\int_V x \,d\mu(x)\right)$. The right hand side is $\exp_W(\Prob(f)(\mu)) = \int_W y \,d(f_\conv\mu)(y) = \int_V f(x) \,d\mu(x)$.
The naturality condition is therefore the equality:
\[ f\left(\int_V x \,d\mu(x)\right) = \int_V f(x) \,d\mu(x) \]
This is precisely the equality condition of Jensen's inequality, which holds if $f$ is both concave and convex. In particular, this is true if the function $f$ is linear map.

Next, we verify the graded functor axioms.

The first diagram requires commutativity for the product of spaces.
\[
\begin{tikzcd}
    \Prob(V) \times \Prob(W) \arrow[r, "\exp_V \times \exp_W"] \arrow[d, "\mu_{V, W}"'] & \lvert V \times W \rvert \arrow[d, "\id"] \\
    \Prob(V \times W) \arrow[r, "\exp_{V \times W}"] & \lvert V \times W \rvert
\end{tikzcd}
\]
Starting with $(\mu_1, \mu_2) \in \Prob(V) \times \Prob(W)$, the top path yields the pair of expectations $(\exp_V[\mu_1], \exp_W [\mu_2])$. The bottom path first forms the product of two independent variables $\mu_1 \tens \mu_2$ and then takes its expectation, which is $\exp_{V\times W}[\mu_1 \tens \mu_2]$. The expectation of the product of two independent variables  is the product of the expectations. Therefore the diagram commutes.

The second diagram involves the addition map $+\colon  V \times V \to V$, which induces convolution.
$$
\begin{tikzcd}
    \Prob(V \times V) \arrow[r, "\exp_{V \times V}"] \arrow[d, "\Prob(+)"'] & \lvert V \times V \rvert \arrow[d, "+"] \\
    \Prob(V) \arrow[r, "\exp_V"] & \lvert V \rvert 
\end{tikzcd}
$$
The expectation of a sum of probability measures is always equal to the sum of expectations. Therefore this diagram also commutes.

Finally, we need to show that the seminorm of the expectation map is less than or equal to $1$. Recalling that we are working with the Lipschitz seminorm, so it suffices to prove that
$$
d_V(\exp(\mu), \exp(\nu)) \leq d_l(\mu, \nu)
$$
However, this is true because, if $l>1$, $d_V(\exp(\mu), \exp(\nu)) > 0$ implies that  $d_l(\mu, \nu) = \infty.$

Since all conditions are met under the restriction to linear maps, we conclude that $(\Prob, |- |)$ is a graded functor pair on this subcategory. Observe that the map
$$
\lvert V \rvert \xrightarrow{\triangle_{\lvert V \rvert }} \lvert V \rvert \times \lvert V \rvert\xrightarrow{\lvert + \rvert} \lvert V \rvert
$$
is precisely rescaling by 2. Therefore we conclude that the grading constant of $|-|$ is 2.
\end{proofE}
\begin{propositionE}
\label{prop:gradedpair}
The tuple $(\ProbZ,\PosDef, \Var)$ is a Cartesian pre-CLT system.
\end{propositionE}

\begin{proofE}
We must show that the pair of functors $(\ProbZ, \PosDef)$ forms a graded pair with respect to $\Var$ as per \Cref{def: grading}. Therefore we must demonstrate that the variance map, which we denote by $\Var\colon  \ProbZ \to\PosDef$, is a natural transformation that commutes with the lax monoidal structures.

Let $V \in \eFinVect$. The functor $\ProbZ$ maps $V$ to the space of probability measures on $V$ with zero expectation and finite, non-zero variance. The functor $\PosDef$ maps $V$ to the space of symmetric positive semidefinite matrices over it. The grading map $\Var_V\colon  \ProbZ(V) \to \PosDef(V)$ is defined for a measure $\mu \in \ProbZ(V)$ as the matrix of second moments:
$$ (\Var_V(\mu))_{ij} = \int_V x_i x_j \,d\mu $$
This matrix is the variance matrix because the expectation is zero and thus we do not need to renormalise.

\textbf{1. Naturality of Var:}
We must show that $\Var$ is a natural transformation. Let $f\colon  V \to W$ be a linear map. We need to show that the following diagram commutes:
\[
\begin{tikzcd}
\ProbZ(V) \arrow[r, "\Var_V"] \arrow[d, "\ProbZ(f)"'] &\PosDef(V) \arrow[d, "\PosDef(f)"] \\
\ProbZ(W) \arrow[r, "\Var_W"'] &\PosDef(W)
\end{tikzcd}
\]
Here, $\ProbZ(f)$ is the pushforward map $f_\ast$, so $\ProbZ(f)(\mu) = f_\ast\mu$. The map $\PosDef(f)$ transforms a matrix $M \in\PosDef(V)$ to $f M f^T \in\PosDef(W)$.

Let $\mu \in \ProbZ(V)$. We compute both paths of the diagram.
\begin{multline*}
(\PosDef(f) \circ \Var_V)(\mu) = f (\Var_V(\mu)) f^T = \\ \Var_W(f_\ast\mu) = (\Var_W \circ \ProbZ(f))(\mu)
\end{multline*}
To check the second equality above, we compute the $(i,j)$-th entry of the matrix of $\Var_W(f_\ast\mu)$. Let $y \in W$ with coordinates $y_k$, and let $y=f(x)$.
\begin{align*}
    (\Var_W(f_\ast\mu))_{ij} &= \int_W y_i y_j \,d(f_\ast\mu)(y) \\
    &= \int_V (f(x))_i (f(x))_j \,d\mu(x) \\
    &= \int_V \left(\sum_k f_{ik}x_k\right) \left(\sum_l f_{jl}x_l\right) \,d\mu(x) \\
    &= \sum_{k,l} f_{ik} f_{jl} \int_V x_k x_l \,d\mu(x) \\
    &= \sum_{k,l} f_{ik} (\Var_V(\mu))_{kl} f_{jl} \\
    &= (f (\Var_V(\mu)) f^T)_{ij}
\end{align*}
Since this holds for all entries $(i,j)$, we have $\Var_W(f_\ast\mu) = f (\Var_V(\mu)) f^T$. The diagram commutes, so $\Var$ is a natural transformation.

\textbf{2. Commutativity with Lax Monoidal Maps:}
We must verify the two diagrams from Definition 5.6 commute. The monoidal product on $\eFinVect$ is the Cartesian product $\times$.
    First, note that the lax monoidal map for $\ProbZ$ is the product of measures,
    $$
    \mu_{V, W}\colon  \ProbZ(V) \times \ProbZ(W)  \to \ProbZ(V \times W)
   $$
    The convolution product on $\PosDef(\mathbb{R})$ corresponds to the addition of variances. For the general case of $\PosDef(V)$, the product on the grading space is addition. If $\Sigma_1$ and $\Sigma_2$ are the variance matrices for independent random vectors, the variance of their sum is $\Sigma_1 + \Sigma_2$.
The first diagram in Definition 5.6 requires that for $\mu_1 \in \ProbZ(V)$ and $\mu_2 \in \ProbZ(W)$, the variance of the product measure $\mu_1 \tens \mu_2$ on $V \times W$ corresponds to the block-diagonal matrix of individual variances.
\[
\begin{tikzcd}
\ProbZ(V) \times \ProbZ(W) \arrow[r, "\Var_V \times \Var_W"] \arrow[d, "\mu_{\ProbZ}"'] &\PosDef(V) \times\PosDef(W) \arrow[d, "\mu_{\PosDef}"] \\
\ProbZ(V \times W) \arrow[r, "\Var_{V \times W}"'] &\PosDef(V \times W)
\end{tikzcd}
\]
The variance matrix of the product measure $\mu_1 \tens \mu_2$ is precisely the block diagonal matrix 
$$
\begin{pmatrix}
 \Var(\mu_1) & 0
 \\
 0 & \Var(\mu_2).
\end{pmatrix}$$
The map $\mu_{V, W}$ therefore corresponds to this block-diagonal construction. We conclude that the diagram commutes.

The second diagram involves the convolution product. The map $+\colon  V \times V \to V$ is vector addition. Let $\mu \in \ProbZ(V)$. The convolution $\mu \conv \mu$ corresponds to the distribution of $Z = X_1 + X_2$, where $X_1, X_2$ are independent random variables with distribution $\mu$. The variance of $Z$ is $\Var(X_1) + \Var(X_2) = 2 \cdot \Var(\mu)$. The map $\ProbZ(+)$ sends the product measure on $V \times V$ to the convolution on $V$. The commutativity of the second diagram is therefore precisely the fact that the variance of a sum of independent random variables is the sum of their variances. This holds by construction.

Since $\Var$ is a natural transformation and the required diagrams commute, the pair $(\ProbZ,\PosDef)$ is a graded pair with respect to $\Var$.  
\end{proofE}
\end{proofE}
The morphisms in dilated categories often have too few points to be useful in applications. Once we have rescaled our morphism, we no longer need the rescaling machinery, so we drop it and make the change of basis in Theorem \ref{thm:change_of_basis} to pass from a dilated $\UCat{C}$ to the underlying seminorm category $\DtoSNorm_\ast\UCat{C}$ which has more points.
\begin{theoremE}
  \label{structural_CLT}
  Given a pre-CLT system $(\enr{F}, \enr{G}, p, \Delta_{\enr{F}}, \Delta_{\enr{G}}, \mu)$,  the diagram
  \begin{equation*}
    \begin{tikzcd}[ampersand replacement=\&]
      \enr{F}X \arrow[rr, "{[c, e]\act \theta^F_X}"] \arrow[dr, "p_X"'] \& \& \enr{F}X \arrow[dl, "p_X"] \\
      \& \enr{G}X \&
    \end{tikzcd}
  \end{equation*}
  commutes.

Moreover, the fibre $\enr{F}(X)_p$ over $p_X$, taken in $\DtoSNorm_\ast \UCat{D}$, of each generalised point  $\1_{\UCat{D}} \xrightarrow{p} \enr{G}X$ of $\enr{G}X$ is such that $\UCat{D}(\1_{\UCat{D}}, \enr{F}(X)_p)$ is nonempty. The morphism $[c, e]\act \theta^F_X$ induces an endomorphism $$\theta_p \in\DtoSNorm_\ast\UCat{D}\left(\enr{F}(X)_p, \enr{F}(X)_p\right).$$
\end{theoremE}
where $c$ is the grading constant of $\enr{G}.$
\begin{proofE}
The commutativity of the diagram can be seen by writing the $\enr{F}$ and $\enr{G}$ convolution operator maps in terms of components. 
$$
\begin{tikzcd}
\enr{F}X \arrow[r, "\Delta_{\enr{F}}"] \arrow[d, "p_X"']
  & \enr{F}X \ptens \enr{F}X \arrow[r, "\mu_{\enr{F}X, \enr{F}X}"] \arrow[d, "p_X \ptens p_X"']
  & \enr{F}(X \ptens X) \arrow[r, "{ \enr{F}(-+-)}"] \arrow[d, "p_{X \ptens X}"']
  & \enr{F}X \arrow[d, "p_X"']
  \\
  \enr{G}X \arrow[r, "\Delta_{\enr{G}}"'] 
  & \enr{G}X \ptens \enr{G}X \arrow[r, "\mu(\enr{G})_{
  \enr{G}X, \enr{G}X}"'] 
  & \enr{G}(X \ptens X) \arrow[r, "{\enr{G}(-+-)}"'] 
  & \enr{G}X 
\end{tikzcd}
$$
The commutativity of the first square to the left holds by the definition of a grading diagonal. The second and third are by the two commutative diagrams  defining a grading of functors. Then, by the perfectly rescalable assumption, the composition on the bottom row is equal to $c\act \id_{\enr{G}X}.$ Therefore, by the assumption that $c\otimes [c,e] = e$ rescaling this diagram by $[c, e]$ produces an identity on the bottom and so we obtain the commutative square in the statement.

The next thing we note is that the seminorm of all the morphisms in the diagrams we take the fibre over have seminorm less than $e$. To check this, observe that generalised points have seminorm $\bot$, the diagonal maps have seminorm $e$, and the units of the lax monoidal transformations do too. So $\abs{\theta^F} \le \abs{\triangle_F}\abs{\mu}\abs{F(+)} = \abs{F(+)} = \abs{+} \leq c$. So $\abs{[c,e]\theta^F} \le e.$   So this is a well-defined conical diagram in $\eCC.$

To prove the claim about induced maps on fibres, observe that  our assumptions guarantee the commutativity of the following diagram. The labels of the arrows on the bottom row follow from a simple diagram chase.
  \begin{equation}
  \label{dia:cospan}
    \begin{tikzcd}
\enr{F}X \arrow[r, "\Delta_{\enr{F}}"] \arrow[d, "p_X"']
  & \enr{F}X \ptens \enr{F}X \arrow[r, "\mu_{\enr{F}X, \enr{F}X}"] \arrow[d, "p_X \ptens p_X"']
  & \enr{F}(X \ptens X) \arrow[r, "{[c, e] \act \enr{F}(-+-)}"] \arrow[d, "p_{X \ptens X}"']
  & \enr{F}X \arrow[d, "p_X"'] \\
\enr{G}X \arrow[r, "\Delta_{\enr{G}}"'] 
  & \enr{G}X \ptens \enr{G}X \arrow[r, "\mu(\enr{G})_{\enr{G}X, \enr{G}X}"'] 
  & \enr{G}(X \ptens X) \arrow[r, "{[c, e]\act \enr{G}(-+-)}"'] 
  & \enr{G}X  \\
\1_{\UCat{D}} \arrow[r] \arrow[u, "p"]
  & \1_{\UCat{D}} \arrow[r] \arrow[u, "\Delta_{G}\circ p"]
  & \1_{\UCat{D}} \arrow[r] \arrow[u, "\mu(\enr{G})_{\enr{G}X, \enr{G}X} \circ \Delta_{G}\circ p "]
  & \1_{\UCat{D}} \arrow[u, "p"]
\end{tikzcd}
\end{equation}
The full composite from left to right is a morphism of cospans, which therefore induces a corresponding map on its limit, the fibre at $p$. 

We remark that the claimed decomposition just below the statement of the theorem we are proving follows easily from the diagram above. Each pair of consecutive columns and the arrows between them produces a morphism in the limit through which $\theta_p$ factors.

The fibre is inhabited as $p_X$ is an epimorphism. In the category $\SNorm$ (and $\DNorm$), epimorphisms are surjective. Since $p_X$ is an epimorphism by assumption, the strict fibre over any point, which is a carrier in $\Set$ for the enriched fibre by Proposition \ref{my_prop}, is inhabited.
\end{proofE}
\begin{example}
For the Cartesian pre-CLT system $(\ProbZ,\PosDef, \Var)$ and a vector space $V$, the fibre over a particular choice matrix $M \in\PosDef(V)$ will be the metric space $\ProbZ(V)_M$ consisting of precisely the probability measures in $\ProbZ(V)$ with variance $M$. For $\left(\Prob, |-|, \exp\right)$, it is the metric space of probability measures in $\Prob(V)$ with expectation $x$. In both cases, we move from $\ekCMet$ to $\eCMet$ to do the further computations.
\end{example}
Having verified the structural conditions checked we can move on to checking the  analytic ones on the fibre.
\begin{definition}
Given a pre-CLT system $(F, G, p, \Delta_{\enr{F}}, \Delta_{\enr{G}}, \mu)$. The \textbf{$p$-restricted convolution operator} denoted $\left(\theta_p\right)_\ast$ is given by first by considering the restriction $[c, e] \act\theta^F_X$ to an element $\theta_p \in\DtoSNorm_\ast\UCat{C}\left(\enr{F}(X)_p, \enr{F}(X)_p\right)$. Then 
$$
\left(\theta_p \right)_\ast\colon  \DtoSNorm_\ast\UCat{D}\left(\1_{\UCat{D}}, \enr{F}(X)_p\right) \to \DtoSNorm_\ast\UCat{D} \left(\1_{\UCat{D}}, \enr{F}(X)_p\right)
$$
is the morphism in $\SNorm$ defined by postcomposition with $\theta_p$.
\end{definition}
\begin{example}
\begin{enumerate}
\item
In the LLN example,  the $x$-restricted convolution operator is the map 
$
\mu \mapsto \frac{1}{2}\mu\ast \mu
$
restricted to probability measures of expectation $x$.
\item
In the CLT example, the $M$-restricted convolution operator is the map 
$
\mu \mapsto \frac{1}{\sqrt{2}}\mu\ast \mu
$
restricted to probability measures of variance matrix $M$.
\end{enumerate}
\end{example}
The restricted $p$-convolution can be broken down further using the decomposition of $\enr{G}$-convolution operator as $\enr{G}(+)\circ \mu(G)\circ \Delta_G$. 
\begin{propositionE}
\label{prop:decomposition}
The $p$-restricted convolution operator $\theta_p$ on the fibre $\enr{F}(X)_p$ admits a factorisation%
\begin{equation*}
    \theta_p = \alpha_p \circ m_p \circ \delta_p
\end{equation*}
passing through the intermediate fibres induced by the grading structure.
If the seminorm of $\alpha_p$ is less than $e$ then $|\theta_p| < e$.
\end{propositionE}
\begin{proofE}
Consider the commutative diagram of cospans in the proof of Theorem \ref{structural_CLT}:
  \begin{equation}
    \begin{tikzcd}
    \label{nnnnn}
\enr{F}X \arrow[r, "\Delta_{\enr{F}}"] \arrow[d, "p_X"']
  & \enr{F}X \ptens \enr{F}X \arrow[r, "\mu_{\enr{F}X, \enr{F}X}"] \arrow[d, "p_X \ptens p_X"']
  & \enr{F}(X \ptens X) \arrow[r, "{[c, e] \act \enr{F}(-+-)}"] \arrow[d, "p_{X \ptens X}"']
  & \enr{F}X \arrow[d, "p_X"'] \\
\enr{G}X \arrow[r, "\Delta_{\enr{G}}"'] 
  & \enr{G}X \ptens \enr{G}X \arrow[r, "\mu(\enr{G})_{\enr{G}X, \enr{G}X}"'] 
  & \enr{G}(X \ptens X) \arrow[r, "{[c, e]\act \enr{G}(-+-)}"'] 
  & \enr{G}X  \\
\1_{\UCat{D}} \arrow[r] \arrow[u, "p"]
  & \1_{\UCat{D}} \arrow[r] \arrow[u, "\Delta_{G}\circ p"]
  & \1_{\UCat{D}} \arrow[r] \arrow[u, "\mu(\enr{G})_{\enr{G}X, \enr{G}X} \circ \Delta_{G}\circ p "]
  & \1_{\UCat{D}} \arrow[u, "p"]
\end{tikzcd}
\end{equation}
By the definition of the fibre, because it is a limit and therefore functorial, the columns of this diagram define three distinct fibres via pullback:
\begin{enumerate}
    \item The fibre $\enr{F}(X)_p$ over $p$;
    \item The fibre $(\enr{F}X \ptens \enr{F}X)_{\Delta_{G}\circ p} $ over $\Delta_{G}\circ p$ in $\enr{F}X \ptens \enr{F}X$;
    \item The fibre $\enr{F}(X \ptens X)_{\mu(\enr{G})\circ \Delta_{G}\circ p}$ over $\mu(\enr{G})\circ \Delta_{G}\circ p$ in $\enr{F}(X \ptens X)$.
\end{enumerate}
Explicitly, we obtain $\delta_p \colon \enr{F}(X)_p \to (\enr{F}X \ptens \enr{F}X)_{\Delta_{\enr{G}} \circ p}$ induced by $\Delta_{\enr{F}}$, $m_p \colon (\enr{F}X \ptens \enr{F}X)_{\Delta_{\enr{G}} \circ p} \to \enr{F}(X \ptens X)_{\mu(\enr{G}) \circ \Delta_{\enr{G}} \circ p}$ induced by $\mu_{\enr{F}}$, and $\alpha_p \colon \enr{F}(X \ptens X)_{\mu(\enr{G}) \circ \Delta_{\enr{G}} \circ p} \to \enr{F}(X)_p$ induced by the rescaled addition $[c, e] \act \enr{F}(+)$
The composition of these maps is precisely the induced endomorphism $\theta_p = \alpha_p \circ m_p \circ \delta_p$.
Since the morphisms between fibres are restrictions of the morphisms in $\enr{F}$, they satisfy $|\delta_p| \le |\Delta_{\enr{F}}| \le e$, $|m_p| \le |\mu_{\enr{F}}| = e$ and $|\alpha_p| \le |[c, e] \act \enr{F}(+)|$.
Using the submultiplicative property of the seminorm, we have:
\begin{equation*}
    |\theta_p| \le |\alpha_p| \qMul |m_p| \qMul |\delta_p|.
    \end{equation*}
If  $|\alpha_p|$  has seminorm strictly less than $e$, the contractive property of the quantale implies that the total product is strictly less than $e$ as the other two are bounded by $e$.
\end{proofE}
The second condition in \Cref{def:contractivity_on_Fibres} can also be made in terms of seminorm $|\alpha_p|$. However, we found this harder to compute in practical examples than the $|\theta_p|$ so we opted to use that instead.
\subsection*{The categorical central limit theorem via CLT-systems}
Succinctly, the categorical central limit theorem is fundamentally about the convergence of this $p$-restricted operator to a central limit in each fibre. To ensure analytic convergence, we must impose analytical conditions.

 \begin{definition}
\label{def:contractivity_on_Fibres}
A pre-CLT system $(F, G, p, \Delta_{\enr{F}}, \Delta_{\enr{G}}, \mu)$ is called a \emph{CLT-system} if for every generalised point $\1_{\UCat{D}} \xrightarrow{p} \enr{G}(X)$,
\begin{enumerate}
\item
the fibre $\enr{F}(X)_p$ is metrically small.
\item
the $p$-restricted convolution operator $\theta_p$ is strictly contractive i.e., $|\theta_p| < e$. 
\end{enumerate}
We say it is a  \emph{Cartesian CLT system} if it is Cartesian as a pre-CLT system.
\end{definition}

For our running examples, this condition is satisfied.
\begin{theoremE}
\label{llnclt}
Both of these are Cartesian CLT-systems:
\begin{enumerate}
\item
the tuple $(\Prob, |-|, \exp)$.
\item
the tuple $(\ProbZ,\PosDef, \Var)$.
\end{enumerate}
\end{theoremE}
\begin{proofE}
See the proof of  the below Theorem \ref{prop:LLN} for (1) and Theorem \ref{thm:clt_contractive} for (2).
\begin{theoremE}
\label{prop:LLN}
The tuple $(\Prob, |-|, \exp)$ is a Cartesian CLT-system.
\end{theoremE}
\begin{proofE}
First, we need to compute the fibre at each point. For $x \in |V|$ it follows by direct inspection that the fibre in the seminorm category $\eCMet$ is 
$$
P_x(V) = \{\mu \in\Prob(V):  \exp(\mu) = x \}.
$$
It follows from Proposition \ref{prop:metric_properties} that it is extended metric space and metrically small.

The constant $c$ is $2$ by the previous proposition. We show that the operator $\theta_x$ is a strict contraction on each fibre. The operator is $\theta\colon  P_x(V) \to P_x(V)$, defined by $\theta(\mu) = \frac{1}{2}(\mu \conv \mu)$.

Let $A  := \int_V e^{-i\langle t,y \rangle} d\mu$ and $C := \int_V e^{-i\langle t,y \rangle} d\nu$ be the characteristic functions of two measures $\mu, \nu \in P_x$. It is well known that the characteristic function of the convolution $\mu \conv \mu$ is $A^2$, see \cite[Section 26]{Billingsley12}.Using the algebraic identity $A^2 - C^2 = (A-C)(A+C)$ and the fact that characteristic functions are bounded by $1$ (implying $\|A+C\| \leq 2$), we derive the sharp bound by applying the definition of the Fourier distance directly:
\begin{align*}
d_{l}(\mu \conv \mu, \nu \conv \nu) &= \left( \sup_{t \neq 0} \frac{\left\| A^2 - C^2 \right\|}{\|t\|^l} \right)^{1/l} \\
&\leq \left( \sup_{t \neq 0} \frac{2 \left\| A - C \right\|}{\|t\|^l} \right)^{1/l} \\
&= 2^{1/l} d_{l}(\mu, \nu).
\end{align*}
By the definition of $\act$, the map $\Prob(\frac{1}{2} \cdot -) \colon \Prob(V) \to \Prob(V)$ is equal to $\frac{1}{2} \act - \colon   \Prob(V) \to \Prob(V)$.  Therefore:
$$
d_{l}(\theta(\mu), \theta(\nu)) = \frac{1}{2}d_{l}(\mu \conv \mu, \nu \conv \nu) \le \frac{1}{2} \cdot 2^{1/l} \cdot d_{l}(\mu, \nu) = 2^{l/1-1} d_{l}(\mu, \nu).
$$
Since $l \in (0,1]$, the exponent $2^{1/ l -1}$ is strictly negative. Thus, $|\theta_x| < 1$ as required.
\end{proofE}
\begin{theoremE}
\label{thm:clt_contractive}
The tuple $(\ProbZ,\PosDef, \Var)$ is a Cartesian CLT-system.
\end{theoremE}

\begin{proofE}
The fact that the grading constant $c$ is $\sqrt{2}$ follows from the fact that the sum of two independent variables with variance matrix $M$ results in a variance of $2M$. Scaling the sum by $\sqrt{2}$ multiplies the variance matrix by $2$.

The fibre is the same as in the proof of Proposition \ref{prop:LLN}.
For $M \in \PosDef(V)$ it follows by direct inspection that the fibre in the seminorm category $\eCMet$ is 
$$
P_x(V) = \{\mu \in\Prob(V):  \Var(\mu) = M \}.
$$
We have already shown earlier that this is a metric space and thus metrically small for $l\in [2,3)$ o to prove the theorem, we must show that the central limit operator is a contraction in this fibre

Let $P_M$ be the fibre of measures in $\ProbZ(V)$ with a given variance matrix $<$. The central limit operator is $\theta\colon \ProbZ(V)_M \to \ProbZ(V)_M$, defined as $\theta(\mu) = \frac{1}{\sqrt{2}}(\mu \conv \mu)$.
 Let us write
$$
A := \int_V e^{-i\langle t,x \rangle} d\mu \quad \text{and} \quad C := \int_V e^{-i\langle t,x \rangle} d\nu.
$$
The characteristic function of $\mu \conv \mu$ is $A^2$ by the fact that Fourier sends convolution to multiplication. Using the algebraic identity $A^2 - C^2 = (A-C)(A+C)$ and the fact that characteristic functions are bounded by 1, we obtain the sharp bound for the distance $d_{l}(\mu \conv \mu, \nu \conv \nu)$:
$$
\left( \sup_{t \neq 0} \frac{\| A^2 - C^2 \|}{\|t\|^l} \right)^{1/l} \le \left( \sup_{t \neq 0} \frac{\| A-C \| \cdot 2}{\|t\|^l} \right)^{1/l} = 2^{1/l} d_{l}(\mu, \nu).
$$
Now we study the operator $\theta$. We showed that the Fourier $l$-distance scales linearly with the measure, i.e., $d_{l}(c\mu, c\nu ) = c \cdot d_{l}(\mu, \nu)$.
Therefore,
$$
d_{l}(\theta(\mu), \theta(\nu)) = d_{l}\left(\frac{1}{\sqrt{2}}(\mu \conv \mu), \frac{1}{\sqrt{2}}(\nu \conv \nu)\right) = \frac{1}{\sqrt{2}} d_{l}(\mu \conv \mu, \nu \conv \nu).
$$
Combining these results, we find
$$
d_{l}(\theta(\mu), \theta(\nu)) \le \frac{1}{\sqrt{2}} \cdot 2^{1/l} d_{l}(\mu, \nu) = 2^{\frac{1}{l} - \frac{1}{2}} d_{l}(\mu, \nu).
$$
Since we chose $l > 2$, the term $1/l$ is strictly less than $1/2$. Consequently, the exponent $\frac{1}{l} - \frac{1}{2}$ is strictly negative, which means the coefficient $k=2^{\frac{1}{l} - \frac{1}{2}}$ is strictly less than 1. Thus, $\theta$ is a strict contraction, which satisfies the contractive condition of Theorem~\ref{thm:gCLT}.
\end{proofE}
\end{proofE}
\begin{proof}[Sketch of proof]
One first verifies that $\eCMet(1, \Prob(V)_x)$ (resp.\ $\eCMet(1, \ProbZ(V)_M)$) is a metric space, for each $x \in |V|$ (resp.\ $M \in \ProbZ(V)$) rather than an extended metric space. Then one must explicitly computes a Lipschitz bound on the operator
$
\mu \mapsto \frac{1}{\sqrt{2}}\mu\conv \mu
$
within  $\eCMet(1, \Prob(V)_x)$ for each $x$. 
\end{proof}

Once one has a CLT-system, by checking the conditions of Theorem \ref{thm:banach_fixed_point}, we can deduce the existence of a \emph{central limit} in each fibre.
\begin{theoremE}[Categorical Central Limit Theorem]
\label{thm:gCLT}
Suppose $(\enr{F}, \enr{G}, p, \Delta_{\enr{F}}, \Delta_{\enr{G}}, \mu)$ is a CLT system. Then for every generalised point $p\colon  1 \to \enr{G}(X)$ in $\DtoSNorm_\ast \UCat{D}$, the convolution operator $\theta_p$ on the fibre $\enr{F}(X)_p$ has a unique fixed point $\mathcal{N}_p \in \UCat{D}(1, \enr{F}(X)_p)$ which we call the \textbf{central limit}. For any initial point $x_0 \in \UCat{D}(1, \enr{F}(X)_p)$, this fixed point is given by the limit:
    $$
    \mathcal{N}_p = \lim_{n\to \infty} (\theta_p)^{\circ n}(x_0).
    $$%
\end{theoremE}
\begin{proofE}
We need to establish that $\mathcal N_p$ exists for each generalised point $p\colon  \1_{\UCat{D}} \to \enr{G}(X)$ in $\DtoSNorm_\ast \UCat{D}$. 

By Definition \ref{def:contractivity_on_Fibres}, the fibre $\enr{F}(X)_p$ taken in $\DtoSNorm_\ast(\UCat{D})_{\preceq e}$ is metrically small and the operator $\theta_p$ satisfies $|\theta_p| \prec e$. Finally, by Theorem \ref{structural_CLT}, $\UCat{D}(\1_{\UCat{D}}, J)$ is nonempty

Moreover, by the construction in Theorem \ref{structural_CLT}, $\theta_p$ is an endomorphism in the seminorm category $\DtoSNorm_\ast(\UCat{D})_{\preceq e}$. The underlying distance space $\UCat{D}(1, \enr{F}(X)_p)$ inherits the structure of a complete $\V$-space from $\enr{F}(X)_p$.

We can therefore apply Theorem \ref{thm:banach_fixed_point} to the operator 
$$
(\theta_p)_\ast\colon  \UCat{D}(\1_{\UCat{D}}, \enr{F}(X)_p) \to \UCat{D}(\1_{\UCat{D}}, \enr{F}(X)_p)
$$ 
defined by $(\theta_p)_\ast(\phi) = \theta_p \circ \phi$ for $\phi\colon  \1_{\UCat{D}} \to \enr{F}(X)_p$ to conclude that it has a unique fixed point $\mathcal N_p \in \UCat{D}(\1_{\UCat{D}}, \enr{F}(X)_p)$ characterised by 
$$
\mathcal N_p = \lim_{n \to \infty} (\theta_p)^{\circ n}(\phi_0)
$$
for any initial point $\phi_0\colon  \1_{\UCat{D}} \to \enr{F}(X)_p$. By Theorem \ref{thm:banach_fixed_point}, this limit exists and can be computed as  $\lim_{n\to \infty}\left(\theta_p \right)^{\circ n}(x_0)$. This completes the first part of the theorem.
\end{proofE}
By applying this theorem to our main examples, we recover versions of the Law of Large Numbers and the Central Limit theorem as immediate corollaries.
\begin{corollaryE}[Probabilistic limiting theorems]
\label{cor:plts}
\
\begin{enumerate}
\item
Let $\mu$ be a probability measure on a finite vector space $V$ with
$
\int \|x\|^{\lambda} d\mu < \infty \mbox{ for some } \lambda > 1.
$
Then
$
\frac{1}{2^n}\mu^{\conv 2^n}
$
converges as $n\to \infty$ and the limit depends only on the expectation of $\mu.$ 
\item
Let $\mu$ be a probability measure on a finite vector space $V$ with expected value $0$ and 
$
\int \|x \| ^{\lambda} d\mu < \infty \mbox{ for some } \lambda > 2.
$
Then
$
\frac{1}{\sqrt{2}^n}\mu^{\conv 2^n}
$
converges as $n\to \infty$ and the limit depends uniquely on the variance matrix of $\mu.$ 
\end{enumerate}
\end{corollaryE}
\begin{proofE}
The full statements are as follows and the proofs are given in section E.
\begin{corollaryE}[The law of large numbers]
Let $\mu$ be a probability measure on a finite vector space $V$ with
$$
\int x^{\lambda} d\mu < \infty \mbox{ for some } \lambda > 1
$$
then
$
\frac{1}{2^n}\mu^{\conv 2^n}
$
converges as $n\to \infty.$ 
\end{corollaryE}
\begin{corollaryE}[ Probabilistic CLT]
Let $\mu$ be a probability measure on a finite vector space $V$ with expected value $0$ and 
$$
\int \|x \| ^{\lambda} d\mu < \infty \mbox{ for some } \lambda > 2
$$
Then
$
\frac{1}{\sqrt{2}^n}\mu^{\conv 2^n}
$
converges as $n\to \infty$ and the limit depends uniquely on the variance matrix of $\mu.$ 
\end{corollaryE}
To prove these results, we first have the following easy lemma
\begin{lemmaE}
\label{lem:LLN}
Let $\mu$ be a probability measure on a space $V$. Then for all $n \in \mathbb{N}$:
\label{lem:LLN:eq}
$$
\frac{1}{2^n} (-)^{\conv 2^n} = \left( \frac{1}{2} (-)^{\conv 2} \right)^{\circ n};
$$
\end{lemmaE}

\begin{proofE}
We prove this statement by induction on $n$.
\emph{Base case} ($n = 1$):  
$$
\left( \frac{1}{2} \mu^{\conv 2} \right)^{\circ 1} = \frac{1}{2} \mu^{\conv 2}
$$
so item~\ref{lem:LLN:eq} holds. 

\emph{Inductive step}: Suppose the result holds for some $n \in \mathbb{N}$. Then,
$$
\left( \frac{1}{2} \mu^{\conv 2} \right)^{\circ (n+1)} 
= \frac{1}{2} \left( \left( \frac{1}{2} \mu^{\conv 2} \right)^{\circ n} \right)^{\conv 2}
= \frac{1}{2} \left( \frac{1}{2^n} \mu^{\conv 2^n} \right)^{\conv 2}
= \frac{1}{2^{n+1}} \mu^{\conv 2^{n+1}},
$$
where the final identity proves both item~\ref{lem:LLN:eq} for $n$. Thus, the first result holds for all $n$ by induction.
\end{proofE}
The first statement therefore follows from Theorem \ref{thm:gCLT} applied to the CLT-system  $(\Prob, |-|, \exp$ for $l > 1$, which was proven to be a CLT system in Theorem \ref{llnclt}. The convolution operator has the correct form by Proposition \ref{prop:axiomisedconvolution}.

The second theorem by the same argument applied to the CLT-system  $(\ProbZ,\PosDef, \Var)$ for $l > 1$.

\end{proofE}
\begin{example}
For the Cartesian CLT-system $(\ProbZ,\PosDef, \Var)$, the central limit is the Dirac delta measure; for the Cartesian CLT-system $(\ProbZ,\PosDef, \Var)$, the central limit is the Gaussian distribution. This can can be checked by showing that it is a fixed point of the $p$-restricted operator, we can then conclude by the uniqueness of central limits.
\end{example}
\begin{proofE}
The theorem therefore follows from the statement of Theorem \ref{thm:gCLT} applied to the CLT-system  $(\Prob, |-|, \exp, 2^l)$ for $l > 1$ 
\end{proofE}

The central limit has convenient functoriality properties.
\begin{theoremE}[Functoriality of the central limit]
\label{thm:functorialCLT}
In a Lipschitz category $\UCat{D}$, the map
    $$
    \eta_X\colon  \Cat{D}_{lin} \left(\1_{\UCat{D}},  \enr{G}(X) \right) \to \Cat{D}_{lin} \left(\1_{\UCat{D}},  \enr{F}(X)\right) \quad p \mapsto   i_p \circ \mathcal{N}_p
    $$
    where $i_p\colon  \enr{F}(X)_p  \to \enr{F}(X)$ is the canonical inclusion, defines a natural transformation in $\Set$ 
   $$
   \eta\colon  \Cat{D} _{lin}\left(\1_{\UCat{D}},  \enr{G}(-) \right) \to \Cat{D}_{lin} \left(\1_{\UCat{D}},  \enr{F}(-)\right)
    $$ which we call the \textbf{central limit natural transformation in $\Set$}.
\end{theoremE}
\begin{proofE}
The Lipschitz category condition, by Proposition \ref{prop:constant_morph},  guarantees that every point of $\1_{\Set} \to  \Cat{C}\left(1, \enr{F}(X)_p\right)$ in $\Set$ corresponds to a point $\1_{\SNorm} \to \DtoSNorm_\ast\UCat{D}\left(\1_{\UCat{D}}, \enr{F}(X)_p\right)$ in $\SNorm$. This ensures that the map
   $$
    \eta_X\colon  \Cat{D}_{lin} \left(\1_{\UCat{D}},  \enr{G}(X) \right) \to \Cat{D}_{lin} \left(\1_{\UCat{D}},  \enr{F}(X)\right)
    $$
    $$
     p \mapsto   i_p \circ \mathcal{N}_p
    $$
    is well-defined.

\noindent\textbf{Step 1: The naturality of $\theta$ with respect to $\enr{F}_{lin}$ and $\enr{G}_{lin}$}
 First, we show that $\theta^F$ and $\theta^G$ both define endomorphisms of $\enr{F}_{lin}$ and $\enr{G}_{lin}$, respectively. This is a routine verification and without loss of generality, it suffices to consider the case of $\enr{F}$. This means showing that given a morphism $f\colon X \to Y$, one has the equality $ \enr{F}_{lin}(f) \circ \theta^F_X = \theta^F_Y \circ  \enr{F}_{lin}(f)$. We check this by decomposing $\theta^F_X$ and $\theta^G_X$ into components.

The diagonal map $\Delta$ is natural transformation $ \enr{F}_{lin}\to  \enr{F}_{lin}\ptens  \enr{F}_{lin}$. Thus, we have $\Delta_Y \circ  \enr{F}_{lin}(f) = ( \enr{F}_{lin}(f) \ptens  \enr{F}_{lin}(f)) \circ \Delta_X$.

Now we can rewrite the composition:
\begin{align*} \theta^F_Y \circ  \enr{F}_{lin}(f) &= [c, e]\act \enr{F}_{lin}(+_Y) \circ \mu_{Y,Y} \circ \Delta_Y \circ  \enr{F}_{lin}(f) \\ &= [c, e]\act \enr{F}_{lin}(+_Y) \circ \mu_{Y,Y} \circ ( \enr{F}_{lin}(f) \ptens  \enr{F}_{lin}(f)) \circ \Delta_X \end{align*}
Since $ \enr{F}_{lin}$ is a lax monoidal functor, its multiplication $\mu$ is a natural transformation. The naturality square for $\mu$ with respect to the morphism $f\colon  X \to Y$ is $\mu_{Y,Y} \circ ( \enr{F}_{lin}(f) \ptens  \enr{F}_{lin}(f)) =  \enr{F}_{lin}(f \ptens f) \circ \mu_{X,X}$. 

Furthermore, the linear subcategory consists only of those morphisms in $\UCat{C}$  compatible with the pointwise addition, meaning $f \circ +_X = +_Y \circ (f \ptens f)$. Applying the functor $ \enr{F}_{lin}$ yields $ \enr{F}_{lin}(f) \circ  \enr{F}_{lin}(+_X) =  \enr{F}_{lin}(+_Y) \circ  \enr{F}_{lin}(f \ptens f)$. Substituting this into our equation:
\begin{align*} \theta^F_Y \circ  \enr{F}_{lin}(f) &= [c, e]\act \enr{F}_{lin}(f) \circ  \enr{F}_{lin}(+_X) \circ \mu_{X,X} \circ \Delta_X \\ &=  \enr{F}_{lin}(f) \circ \left([c, e]\act \enr{F}_{lin}(+_X) \circ \mu_{X,X} \circ \Delta_X\right) \\ &=  \enr{F}_{lin}(f) \circ \theta^F_X \end{align*}
This confirms that $ \enr{F}_{lin}(f) \circ \theta^F_X = \theta^F_Y \circ  \enr{F}_{lin}(f)$, so $\theta$ is a natural transformation from $ \enr{F}_{lin}$ to itself.

A symmetric argument establishes the naturality of  $\theta^{\enr{G}}$ with respect to  $\enr{G}_{lin}$.

\noindent\textbf{Step 2: Naturality descends to fibres}

The construction of fibres is functorial and therefore we have a factorization in the following diagram 
$$
\begin{tikzcd}
 \ & \ &  F(Y)_{\enr{G}_{lin}(f) \circ p}  \arrow[d,  "i_{\enr{G}_{lin}(f) \circ p}"]
\\
\enr{F}_{lin}(X)_p \arrow[r, "i_p"]  \arrow[urr, dotted, bend left, "q^f_p"] \arrow[d]& \enr{F}_{lin}X \arrow[d, "p_X"] \arrow[r, "\enr{F}_{lin}(f)"] & \enr{F}_{lin}Y \arrow[d, "p_Y"]
\\
\1_{\UCat{D}} \arrow[r, "p"] & \enr{G}_{lin}X  \arrow[r, "\enr{G}_{lin}(f)"]  & \enr{G}_{lin}Y
\end{tikzcd}
$$
 To be precise, the morphism $q^f_p$  can be explicitly constructed the limit of the following morphism of cospans
\begin{equation}
\label{eq:cospan}
\begin{tikzcd}
\1_{\UCat{D}} \arrow[r, "p"]  \arrow[d]& \enr{G}_{lin}X \arrow[d, "\enr{G}_{lin}(f)"]& \enr{F}_{lin}X \arrow[d, "\enr{F}_{lin}(f)" ] \arrow[l, "p_X", swap] 
\\
\1_{\UCat{D}} \arrow[r, "\enr{G}_{lin}(f)\circ p"] & \enr{G}_{lin}Y   & \enr{F}_{lin}Y  \arrow[l, "p_Y"]
\end{tikzcd}
\end{equation}
Let $\mathcal N_x$ be the unique fixed point of $\theta^F_X$ in the fibre $\UCat{D}(\1_{\UCat{D}}, \enr{F}_{lin}(X)_p)$, then apply the pushforward 
 $$(q^f_p)_\ast\colon  \UCat{D}(\1_{\UCat{D}}, \enr{F}_{lin}(X)_p) \to  \UCat{D}(\1_{\UCat{D}}, F_{\enr{G}_{lin}(f) \circ p}). $$
 One obtains a generalised point of $F_{\enr{G}_{lin}(f) \circ p}.$ If we can show that it is a fixed point of $\theta_{\enr{G}_{lin}(f)\circ p}$, then the result will follow from Theorem \ref{thm:gCLT}  on the uniqueness of such points. To that end, we next establish equivariance on fibres.

\noindent\textbf{Step 3: Equivariance on fibres}

In this step, we shall prove that the following diagram commutes.
$$
\begin{tikzcd}
\enr{F}_{lin}(X)_p \dar{q^f_p} \arrow[r, "\theta_{p}"] & \enr{F}_{lin}(X)_p \dar{q^f_p}
\\
\enr{F}_{lin}(Y)_{\enr{G}_{lin}(f)\circ p} \arrow[r, "\theta_{\enr{G}_{lin}(f)\circ p}"] & \enr{F}_{lin}(Y)_{\enr{G}_{lin}(f)\circ p}
\end{tikzcd}
$$
This is easy to see as the diagram above is the limit of the following diagram, and we showed in Step 1 that all the arrows in this diagram commute.
\[
\begin{tikzcd}[row sep=2.5em, column sep=3em]
    \1_{\UCat{D}} 
        \arrow[rr, "id"]                 
        \arrow[dr, "id"]                 
        \arrow[dd, "p"']                 
    & & 
    \1_{\UCat{D}} 
        \arrow[dr, "id"] 
        \arrow[dd, "\enr{G}_{lin}(f)\circ p"] 
    \\
    & \1_{\UCat{D}} 
        \arrow[rr, "id" near start]      
        \arrow[dd, "p" near start]       
    & & 
    \1_{\UCat{D}} 
        \arrow[dd, "\enr{G}_{lin}(f)\circ p"] 
    \\
    \enr{G}_{lin}(X) 
        \arrow[rr, "\enr{G}_{lin}(f)", near start, crossing over] 
        \arrow[dr, "\theta^G" near start] 
    & & 
    \enr{G}_{lin}(Y) 
        \arrow[dr, "\theta^G"] 
    \\
    & \enr{G}_{lin}(X) 
        \arrow[rr, "\enr{G}_{lin}(f)" near start] 
    & & 
    \enr{G}_{lin}(Y) 
    \\
    \enr{F}_{lin}(X) 
        \arrow[rr, "\enr{F}_{lin}(f)", near end, crossing over] 
        \arrow[dr, "\theta^F" near end] 
        \arrow[uu, "p_X" near start, crossing over]
    & & 
    \enr{F}_{lin}(Y) 
        \arrow[dr, "\theta^F"] 
        \arrow[uu, "p_Y"' near start, crossing over] 
    \\
    & \enr{F}_{lin}(X) 
        \arrow[rr, "\enr{F}_{lin}(f)"] 
        \arrow[uu, "p_X"]            
    & & 
    \enr{F}_{lin}(Y) 
        \arrow[uu, "p_Y"'] 
\end{tikzcd}
\]

It follows that they induce maps at the fibre level as it is a limit of these diagrams.

\noindent\textbf{Step 4: The induced map on fixed points.}

We may now reduce to an argument about fixed points in the concrete category $\Set$.  In our previous step, we showed that the map $\V$ is natural. 

It follows that  $q_\ast: \UCat{D}(\1_{\UCat{D}}, \enr{F}_{lin}(X)_p) \to  \UCat{D}(\1_{\UCat{D}}, F_{\enr{G}_{lin}(f) \circ p}) $ is equivariant under the action of $\theta$.  It therefore sends fixed points to fixed points since these are unique in Theorem \ref{thm:gCLT} . The conclusion follows.

\textbf{Step 5: Concluding naturality}

Recall that define the components of the transformation $\eta$ in $\Set$ by $\eta_X(p) := \mathcal{N}_p$, where $\mathcal{N}_p$ is the unique fixed point of the operator $\theta_p$ in the fibre $\enr{F}(X)_p$. We have  shown that 

$$ 
p \to \mathcal N_p
$$
commutes with any linear map $f\colon X\to Y$, in the sense that .
$$ 
f(p) \to \mathcal N_{f(p)}
$$
This completes the proof that in $\Set$, the map
 $$
   \eta\colon  \Cat{D} _{lin}\left(\1_{\UCat{D}},  \enr{G}(-) \right) \to \Cat{D}_{lin} \left(\1_{\UCat{D}},  \enr{F}(-)\right)
  $$  
  is a natural transformation.
\end{proofE}
\begin{proof}[Sketch of proof]
The result essentially follows from the uniqueness of fixed points. One shows by diagram chasing that given a map $f\colon X\to Y,$ there is an induced $F(X)_p \to F(Y)_{f(p)}$, and this map is equivariant with respect to the $p$-restricted convolution $\theta_p$. In particular, this equivariance implies that it sends fixed points to fixed points, and therefore must send the central limit in each fibre to the central limit in the other fibre. The conclusion follows.
\end{proof}
\begin{example}
The category $\eCMet$ is Lipschitz, so in the case of the law of large numbers, this theorem immediately gives the rather obvious result that given a linear map  $f\colon  \mathbb R^n \to \mathbb R^n$ the pushforward of the Dirac delta distribution $\delta_x$ along $\Prob(f)\colon  \Prob(\mathbb R^n) \to \Prob(\mathbb R^n)$ will be the Dirac delta distribution $\delta_{f(x)}$. For the case of the central limit theorem, it says the image of the Gaussian distribution $\mathcal N(0, M)$ along $\Prob(f)\colon  \Prob(\mathbb R^n) \to \Prob(\mathbb R^n)$ along a linear map will be the Gaussian distribution $\mathcal N(0, fMf^T)$.
\end{example}
\section{The CLT for Observables}
\label{sec:pushforward}
\pratendSetLocal{category=pushforward}
We conclude our examples by formulating and proving a novel Central Limit Theorems for Observables (CLTO). This illustrates how more complicated CLTs can be built from simpler ones.
\subsection*{The category of observables}
The intuition behind the following definition is that a measurable space $M$ acts as a black box, while functions $H\colon M \to \mathbb{R}$ represent readable outputs.

A \emph{internal dilated (i.d.) category} is a pair $(\Cat{C},\UCat{D})$ such that $\UCat{D}$ is a monoidal dilated category and $\Cat{D}$ is a monoidal subcategory of $\Cat{C}$. A \emph{i.d.\ functor} is a pair $(F, \enr{G}): (\Cat{C}, \UCat{D})  \to (\Cat{A}, \UCat{B})$ such that  the underlying functor $G$ is a restriction of $F$ to $\UCat{D}.$
\begin{definition}
Let $(\CC, \UCat{D})$ be an i.d.\  category.  The \emph{category of observables $\eObs{M, \UCat{D}}$ from $M\in \Cat{C}$} is the comma category $M/ \Cat{D}$. Explicitly,  $\eObs{M, \UCat{D}}$, has for objects pairs $(V, H)$ where $V \in\UCat{D}$ and $H \in  \Cat{C}(H,V)$ is a morphism we call the \emph{observable}. A morphism from $(V, H_1)$ to $(W, H_2)$ is a morphism $\beta \in \UCat{D}(V, W)$.
    \begin{equation}
    \label{cat_def}
    \begin{tikzcd}
    M \arrow[r, "H_1"] \arrow[rd, "H_2"'] & V  \arrow[d, "\beta"'] \\
      \ & W
    \end{tikzcd}
    \end{equation}
    This is a dilated category.  The distance between two morphisms $(\alpha_1, \beta_1), (\alpha_2, \beta_2)$ is 
     $$
    d_{\eObs{\Cat{C}, \UCat{D}}}\left((\alpha_1, \beta_1), (\alpha_2, \beta_2)\right)  = d_{\UCat{D}}(\beta_1, \beta_2)
        $$
    The seminorm is the Lipschitz constant $| \beta |$ of the map.
 \end{definition}
\begin{propositionE}
The category $\eObs{M, \UCat{D}}$ is a symmetric monoidal dilated category with pointwise addition. Moreover, if $(F, \enr{G}): (\Cat{C}, \UCat{D}) \to (\Cat{A}, \UCat{B})$ is an i.d.\ dilated functor; there is an induced monoidal dilated  functor 
 $$\eObs{F, \enr{G}} \colon \eObs{M, \UCat{D}} \to  \eObs{FM, \UCat{E}}.$$
\end{propositionE}
\begin{proofE}
The tensor product is defined component-wise using the monoidal products in $\Cat{C}$ and the tensor product in $\UCat{D}$. Given two objects $X =(V, H_1)$ and $Y =(W, H_2)$:
$$
 X \tens Y  :=(V \times W, H_1 \times  (H_2 \circ) \triangle_{M}) 
 $$
where $\triangle_{M}: M\to M\times M$ is the diagonal, which is always in $\Cat{C}$ because $\Cat{C}$ is assumed Cartesian. 

The unit object is the terminal morphism $M \to 1_{\Cat{D}}$ The various monoidal structure maps come from post composing $\Cat{D}$.

The pointwise addition is the operation defined by the following morphism:
$$
+_{\UCat{D}}:M \xrightarrow{\triangle_M} M \times M \xrightarrow{H \times H } V\times V \xrightarrow{+} V
$$
There is a functor of sets
$$
\Fpush:\eObs{M, \UCat{D}} ((H_1, V),(H_2, W) ) \xrightarrow{\sim} \UCat{D}(V, W) $$
$$
(H, V) \mapsto H.
$$
which is compatible with composition and we equip $\eObs{M, \UCat{D}}$ (we shall refer to this functor as enriched in the next definition) we can thus the induced $\DNorm$ structure.

The functorial claim follows from observing that $\Fpush$ induces a commutative diagram 
$$
\begin{tikzcd}
    \Obs{M, \UCat{D}} \arrow[d, "{\Obs{F, "\enr{G}}}"] \arrow[r, "{\Fpush(M, \UCat{D})}"] 
    & \UCat{D} \arrow[d, "\enr{G}"] 
    \\
    \Obs{FM, \UCat{B}} \arrow[r, "{\Fpush(FM, "\UCat{B})}"] 
    & \UCat{B}
\end{tikzcd}
$$
as the functor $\Fpush$ depends only on the codomain and not on $M$ or $H$ directly.
\end{proofE}

There is a strong monoidal dilated functor $\Obs{M, \UCat{D}}$ to $\UCat{D}$.

\begin{definition}
The \emph{observation functor} is a functor
$$\Fpush(M, \UCat{D}) \colon \eObs{M, \UCat{D}}\to \UCat{D}, \qquad(H, V) \mapsto H.$$
\end{definition}
This functor is just one of the projections in the comma category.
\begin{proposition}
\label{obsfunctor}
The observation functor is strong monoidal, preserves the seminorm and commutes with $\Cat{C}$-dilated functors.
\end{proposition}
\begin{proofE}
We explicitly verify the properties sequentially. 
The functor induces $\DNorm$-identities on the objects 
$$
\Fpush:\eObs{M, \UCat{D}} ((H_1, V),(H_2, W) ) \xrightarrow{\sim} \UCat{D}(V, W) 
$$
compatible with composition by the construction of dilated category structure. To establish that $\Fpush$ is dilated, let 
$$(\alpha, \beta) \in \Obs{\Cat{C}, \UCat{D}}((V, H_1), (N, W, H_2))$$
 be a morphism. 
By definition, the seminorm $\abs{(\alpha, \beta)}_{\Obs{\Cat{C}, \UCat{D}}}$ is the  seminorm of the underlying map $\beta \in\UCat{D}.$
Thus, the functor is dilated and preserves the seminorm.

It is clearly strong monoidal as it induces that structure on $\eObs{M, \UCat{D}}$.

Finally, it straightforwardly follows from the definition that we have a diagram
$$
\begin{tikzcd}
    \Obs{M, \UCat{D}} \arrow[d, "{\Obs{F, "\enr{G}}}"] \arrow[r, "{\Fpush(M, \UCat{D})}"] 
    & \UCat{D} \arrow[d, "\enr{G}"] 
    \\
    \Obs{FM, \UCat{B}} \arrow[r, "{\Fpush(FM, "\UCat{B})}"] 
    & \UCat{B}
\end{tikzcd}
$$
as the functor $\Fpush$ depends only on the codomain and not on $M$ or $H$.
\end{proofE}
We now prove the main result:
\begin{theoremE}
\label{thm:pushforwardCLT}
Let $M \in \Cat{A}$ be an object, and let  $(F, \enr{G}): (\Cat{A}, \UCat{C}) \to (\Cat{B}, \UCat{D})$ be i.d.\ functor, such that $(\enr{G}, \enr{U}, p)$ is a Cartesian CLT-system for some $\enr{U}$. Suppose that the inclusion $\Cat{D} \to \Cat{B}$ is continuous and cocontinuous. Then the tuple
$$
(\eObs{F M, \enr{G}}, p_{\Fpush(-)}\circ \eObs{M, \enr{G}}, \eObs{M, p})
$$
is also a Cartesian CLT-system.
\end{theoremE}
\begin{proofE}
We follow the two usual steps: first we verify it is a pre-CLT system and then we check it is a CLT-system.

 Most of this follows from the diagram chasing; in general we only need to check the commutativity of diagrams post-application of $\Fpush$. The only place where we need be more careful than this is when computing the fibre itself during the CLT-verification. The limit will be computed in the comma category $ \eObs{M, \UCat{D}}$ which has underlying category $M / \Cat{D}$. Note that, by assumption, the inclusion $\Cat{D}\to \Cat{B}$ is limit preserving. It follows by that $\Fpush$ preserves this limit.

We move onto the formal verification of the commutativity of the pre-CLT diagrams.

First, we note that $\eObs{FM, p}$ is correctly typed. The component of this transformation at an object $X = (V, H)$ in $\eObs{M, \UCat{D}}$ is defined as the morphism 
$$
FM \xrightarrow{FH} \enr{G}V \xrightarrow{p_V} \enr{U}V
$$ 
in the category $\eObs{FM, \UCat{E}}.$

Naturality means that we need to check that the following diagram commutes:
   \begin{equation}
    \label{cat_def}
    \begin{tikzcd}
    FM \arrow[r, "FH_1"] \arrow[rd, "FH_2"'] & FV  \arrow[d, "F\beta"']  \rar{p_V}& U V \dar{U\beta}\\
      \ & FW \rar{p_W} & UW
    \end{tikzcd}
    \end{equation}
     and this follows from the fact that the top line is precisely $p_{\Fpush(-)}\circ \eObs{M, \enr{U}}(H_1, V) $ and commutativity of the whole diagram follows from the naturality of $p$ itself with respect to maps in $\beta \in \UCat{D}(V, W).$

Next, we check that the grading commutes with the lax monoidal maps and the addition operation.

The tensor product in $\eObs{M, \UCat{D}}$ is defined component-wise on the codomain, and the observation functor $\Fpush$ is strong monoidal by Proposition \ref{obsfunctor}. Consequently, the diagram describing the compatibility of the grading with the lax monoidal structure in $\eObs{M, \UCat{D}}$ maps under $\Fpush$ precisely to the corresponding diagram in $\UCat{D}$:
\begin{equation*}
\begin{tikzcd}
\Fpush(\enr{G}X \ptens \enr{G}Y) \arrow[d, "\Fpush(\mu^F)"] \arrow[r, "p \ptens p"]& \Fpush(\enr{U}X \ptens \enr{U}Y) \arrow[d, "\Fpush(\mu^G)"]\\
\Fpush(\enr{G}(X \ptens Y)) \arrow[r, "p"]& \Fpush(\enr{U}(X \ptens Y))
\end{tikzcd}
\end{equation*}
Since $(\enr{G}, \enr{U}, p)$ is a pre-CLT system in the base category, this projected diagram commutes. It follows that the original diagram in $\eObs{M, \UCat{D}}$ commutes.

Similarly, the addition morphism $+_X$ in $\eObs{M, \UCat{D}}$ is defined such that $\Fpush(+_X) = +_{\Fpush(X)}$. The requirement that the grading commutes with addition projects via $\Fpush$ to the condition:$$p_{\Fpush(X)} \circ \enr{G}(+_{\Fpush(X)}) = \enr{U}(+_{\Fpush(X)}) \circ p_{\Fpush(X \ptens X)}$$This holds by the hypothesis that the base system is a pre-CLT system. 

The final condition that the bottom convolution is perfectly rescalable follows easily since the same holds for its image under $\Fpush.$   Thus, the lifted grading satisfies all the conditions of a pre-CLT system.
\\
Next, we verify the CLT conditions. We have the following
$$
\Fpush \left(\eObs{M, \enr{G}}(X)_y \right)= \enr{G}(V)_{\Fpush(y)}.
$$
We have already shown that $\Fpush$ is an isometry. Therefore, this is an isometric isomorphism. Since the base fibre $\enr{G}(V)_{\Fpush(y)}$ is metrically small by the hypothesis that the original system is CLT, the lifted fibre $\eObs{M, \enr{G}}(X)_y$ must be metrically small, as the set of points is determined in $\UCat{D}$.

Consider the convolution operator $\theta_y$ on the fibre $\eObs{M, \enr{G}}(X)_y$. As established in the pre-CLT verification, it follows that the image of $\theta_y$ under $\Fpush$ is precisely $\theta_{\Fpush(y)}$. The seminorm in $\eObs{M, \UCat{D}}$ is defined as the seminorm of the underlying morphism and since $(\enr{G}, \enr{U})$ is a CLT-system, the operator $\abs{\theta_{\Fpush(y)}}_{\UCat{D}} < e$. It follows immediately that the desired fibre is less than $e$. Thus, the system satisfies all desired CLT conditions.

Therefore the tuple 
$$
(\eObs{M, \enr{G}}, p_{\Fpush(-)}\circ \eObs{M, \enr{G}}, \eObs{FM, p})
$$
is a Cartesian CLT-system as required.
\end{proofE}
As a concrete example of this, we work with a category $\Meas_b$ defined as follows: the union of finite-dimensional vector spaces $\FinVect$ and the other measurable spaces $\Meas'$ where morphisms into $\FinVect$ are restricted to be \textbf{bounded} measurable functions. The category $\Meas^0_{b}$ will be the category of pointed such spaces where members of $\FinVect$ are taken to be canonically pointed by the origin. This defines a pair of i.d.\ categories $(\Meas_b, \eFinVect)$ and $(\Meas^0_b, \eFinVect)$. There are then i.d.\ functors $(\P, \Prob)$  and $(\P^0, \ProbZ)$ given by the Giry monad on generic measurable spaces but restricted to $\Prob$ and $\ProbZ$ on $\FinVect.$ The following is immediate.
\begin{corollaryE}
\label{pushing}
Let $M\in \Meas_b$. The tuples $$(\Obs{ \P M, \Prob},   \exp_{\Fpush(-)}\circ \Obs{\P M, \ProbZ}, \Obs{\P M, \exp})$$ 
 $$(\Obs{\P^0 M, \ProbZ},  \Var_{\Fpush(-)}\circ \Obs{\P^0 M, \ProbZ}, \Obs{\P^0 M, \Var})$$ 
 are both Cartesian CLT-system.
\end{corollaryE}
\begin{proofE}
The only thing that needs to be checked is that for $(\Meas_b, \eFinVect)$ and  $(\Meas^0_b, \eFinVect)$, the pushforward of measures from $M \in \Meas_b$ to $\mathbb R^n$ preserves moment assumptions. This is automatic since the functions into are bounded and thus the pushforward measure will automatically have bounded support, which implies that all moments are finite.

The corollary then immediately follows from Theorem \ref{thm:pushforwardCLT} applied to Theorem \ref{llnclt} along with the fact the inclusion of finite vector spaces into measurable spaces preserves finite limits (as it obviously preserves finite products and equalisers of linear maps).
\end{proofE}
\begin{example}
In the setting of statistical mechanics on a symplectic manifold~\cite{Barbaresco23, Marle16},  a \textbf{classical Hamiltonian system} is a triple $(M, \lambda_\omega, H)$, where $M$ is a compact symplectic manifold with its Liouville measure $\lambda_\omega$ (with expectation $m \in M$), and $H\colon  M \to \mathbb{R}$ is the Hamiltonian. Consider a collection of $2^N$ non-interacting, identical systems. The total energy of the ensemble is given by summing the individual energies, $H_{tot}(x_1, \dots, x_{2^N}) = \sum_{i=1}^{2^N} H(x_i)$. Since manifolds are measurable spaces, we have $(H, \mathbb{R}) \in \Obs{M, \eFinVect}.$ Assuming the pushforward of the initial energy distribution has mean zero and finite variance, \Cref{pushing} implies that the normalization of the total energy converges to the fixed point of the CLT system as the number of manifolds in the ensemble approaches infinity.
\end{example}
\begin{proofE}
Since $M$ is a compact manifold, the Hamiltonian function $H: M \to \mathbb{R}$ is continuous on a compact domain and is therefore has bounded image on $\mathbb R$. Thus, $M$ is an object of $\Meas_b^0$ (where the base point is given $m$), and, forgetting the symplective structure the pair $(\mathbb{R}, H)$ therefore constitutes an object in the category of observables $\Obs{\Meas^0_b, \ProbZ}$. The pushforward of the initial energy distribution $\lambda_\omega$ along $H$ is a probability measure on $\mathbb{R}$ with compact support, ensuring it has finite moments of all orders. This pushforward has mean zero because the expectation of $\lambda_\omega$ was $m$, so  the conditions of \Cref{pushing} are satisfied. The theorem then guarantees that the normalized $n$-fold convolution of this energy distribution converges to the normal distribution determined by whatever the variance of  of $H_\ast(\lambda_\omega)$ was.
\end{proofE}
\section{Conclusion}
In this paper, we introduced dilated/seminorm categories as a unified framework for quantitative reasoning about convergence. We established a categorical Banach Fixed Point Theorem and applied it to prove a structural Central Limit Theorem. We demonstrated that this abstract theorem recovers the classical CLT and the Law of Large Numbers, and allows for the systemic derivation of new results such as the CLT for Observables.

Future work offers several promising directions. First, the framework naturally suggests a unification with Markov categories, leading toward a theory of dilated Markov categories that combines structural and quantitative reasoning. The link the enrichment and type theory is also worth investigating. Finally, of course, we would be interested in seeing some practical developments linked to this theory; ie.\ implementations for formal verification in probabilistic programming.
\bibliography{MyBib}

\clearpage         %
\appendix
\textbf{\huge Appendix for review only}
\input{appendix.tex}

\section{Proof details for \cref{sec:prelim}}
\printProofs[prelim]
\section{Proof details for \cref{sec:spaces}}
\printProofs[spaces]
\section{Proof details for \cref{sec:seminormspaces}}
\printProofs[seminormspaces]
\section{Proof details for \cref{sec:seminormcategories}}
\printProofs[seminormcategories]
\section{Proof details for \cref{sec:catBanach}}
\printProofs[catBanach]
\section{Proof details for \cref{sec:probabilityandconvolution}}
\printProofs[probabilityandconvolution]
\section{Proof details for \cref{sec:CLT}}
\printProofs[CLT]
\section{Proof details for \cref{sec:pushforward}}
\printProofs[pushforward]

\end{document}

%% file: appendix.tex
\section{Notation}
\label{sec:notation}
\
\begin{tabular}{l|l}
  $\ECat{C}$, $\eAC$, $\enr{\Ban}$ & Enriched categories (underlined) \\
  $\Cat{C}$, $\CC$ & Meta-variables for underlying/ordinary categories \\
  $\Ban$ & Definite categories (boldface) \\
  $\tens$ & Tensor product of enrichment base ($\VC$, $\SNorm$ or $\DNorm$) \\
  $\tensU$ & Unit of $\tens$ ($\{x\}$ in $\SNorm$, $\Vint$ in $\DNorm$) \\
  $\ptens$ & Tensor on seminorm-enriched categories \\
  $\lin{\eCC}$ & Linear subcategory of $\eCC$ \\
  $\V$ & Quantale of values for distance/seminorm \\
  $\VSp$ & Category of $\V$-spaces and nonexpansive maps \\
  $\CDS$ & Category of complete $\V$-spaces \\
  $\VAct$ & Category of $\V$-action spaces \\
  $\SNorm$ & Category of seminorm spaces \\
  $\DNorm$ & Category of dilated spaces \\
  $\DtoSNorm$ & Change of basis functor $\DNorm \to \SNorm$ \\
  $\eEBan$ & Category of homogenous Banach spaces (dilated) \\
  $\eBan$ & Category of Banach spaces (dilated) \\
  $\eCMet$ & Category of complete extended metric space (seminorm) \\
  $\ekCMet$ & Category of rescaled metric spaces (dilated) \\
  $\eFinVect$ & Category of finite dimensional vector spaces (dilated)  \\
  $\Meas_b$ & Category of bounded measurable spaces (definite) \\
  $\eObs{M, \UCat{D}}$ & Category of observables over $M$ \\
  $\UCat{C}(\1, -)$ & Functor of points valued in $\SNorm/\DNorm$ \\
  $\Cat{C}(\1, -)$ &Functor of points valued in $\Set$\\
  $\Prob$ & Functor of probability measures: finite $l$-th moment \\
  $\ProbZ$ & Probability measures functor: finite $l$-th moment, exp 0 \\
  $\PosDef$ & Functor of positive semidefinite matrices \\
  $\abs{-}$ & Forgetful functor $\FinVect$ to $\ekCMet$  \\
  $\theta^{\enr{F}}_X$ & Convolution operator for functor $\enr{F}$ \\
  $\mathcal{N}_p$ & Central limit in fibre over $p$ \\
  $\ast$ & Convolution product \\
  $\act$ & Action of $\Vint$ (dilation)\\
  $\abs{-}$ & Seminorm applied to morphisms  \\
  $\1$, $\1_{\Cat{C}}$&  Terminal object
\end{tabular}

\section{Details for~\cref{sec:prelim}}
\paragraph*{Coherence Laws of Monoidal Categories}
The coherence laws of a monoidal category are expressed by the commutative diagram below, where the components of natural transformations have been omitted for readability:
\begin{equation}
\begin{tikzcd}[row sep=small, column sep=tiny]
												&A\tens B		
												& 
\\
(A\tens 1)\tens B\arrow[ur, "\rho\tens\id"]\arrow[rr, "\alpha"]	
												&			
												&A\tens (1\tens B)\arrow[ul, swap, "\id\tens\lambda"]
\end{tikzcd}
\end{equation}
\begin{equation}
\begin{tikzcd}[row sep=small, column sep=tiny]
((A\tens B)\tens C)\tens D\arrow[d, swap, "\alpha\tens\id"]\arrow[r, "\alpha"]			
					&(A\tens B)\tens(C\tens D)\arrow[r, "\alpha"]		
					&A\tens(B\tens(C\tens D)) 		
\\
(A\tens(B\tens C))\tens D\arrow[rr, "\alpha"]			
					&							
					&A\tens ((B\tens C)\tens D)\arrow[u, swap, "\id\tens\alpha"]
\end{tikzcd}
\end{equation}
These diagrams are referred to as the \emph{identity} and \emph{pentagon} law, respectively.
\paragraph*{Coherence Laws of Enriched Categories}
The coherence laws of a~$\VC$-category~$\AC$ are expressed by the commutative diagrams below where, as before, we have omitted the components from natural transformations for readability.
\begin{equation}
\begin{tikzcd}[row sep=small, column sep=tiny]
I \tens\eAC(A,B)\arrow[r, "j_B\tens\id"]\arrow[dr, swap, "\lambda"]
		&\eAC(B,B)\tens\eAC(A,B)\arrow[d, "\comp"]
		& \eAC(A,B)\tens I \arrow[l, swap, "\id\tens j_A"]\arrow[dl, "\rho"]
\\
    		& \eAC(A,B) 
		&
\end{tikzcd}
\end{equation}
\begin{equation}
\begin{tikzcd}[column sep=small, row sep=large]
    (\eAC(C,D)\tens\eAC(B,C))\tens\eAC(A,B) 
        \arrow[r, "\alpha"] 
        \arrow[d, swap, "\comp\tens\id"]
    & \eAC(C,D) \tens (\eAC(B,C)\tens\eAC(A,B)) 
        \arrow[d, "\id\tens\comp"] 
    \\
    \eAC(B,D)\tens\eAC(A, B) 
        \arrow[d, "\comp"]
    & \eAC(C,D)\tens\eAC(A,C) 
        \arrow[dl, swap, "\comp"] 
    \\
    \eAC(A,D) & 
\end{tikzcd}
\end{equation}
We refer to these diagrams as the \emph{unit} and \emph{associative laws} of composition in~$\AC$, respectively.